\documentclass[12pt]{article} 
\usepackage[utf8]{inputenc}
\usepackage[a4paper]{geometry}
\usepackage[T1]{fontenc}
\usepackage{setspace}
\doublespace  
\usepackage{amsthm,amsmath}
\usepackage{amssymb}
\usepackage{natbib}
\usepackage[colorlinks=true,linkcolor=black, citecolor=black, urlcolor=black]{hyperref}
\usepackage{bbm}
\usepackage{multicol}
\usepackage{multirow}
\usepackage{mathtools}
\usepackage{url}          
\usepackage{float}
\usepackage{bm}
\usepackage{xcolor}
\usepackage{graphicx} 
\usepackage{caption} 
\usepackage{subcaption} 
\usepackage{algpseudocode}
\usepackage{algorithm}
\usepackage[noEnd=false]{algpseudocodex}
\renewcommand{\algorithmicrequire}{\textbf{Input:}}
\renewcommand{\algorithmicensure}{\textbf{Output:}}

\usepackage{booktabs}
\usepackage{soul}

\newcommand{\p}[1]{{\color{orange}#1}}

\renewcommand*\d{\mathop{}\!\mathrm{d}}

\newcommand{\VaR}{\mathrm{VaR}}
\newcommand{\ES}{\mathrm{ES}}
\newcommand{\MES}{\mathrm{MMES}}
\newcommand{\DCTE}{\mathrm{DCTE}}
\newcommand{\e}{\mathrm{e}}

\newtheorem{theorem}{Theorem}[section]

\newtheorem{proposition}[theorem]{Proposition}

\newtheorem{lemma}[theorem]{Lemma}
\theoremstyle{definition}

\title{Assessing Extreme Risk using Stochastic Simulation of Extremes}
\author{Nisrine Madhar$^{1}$, Juliette Legrand$^{2}$ and Maud Thomas$^{3}$}

\begin{document}

\maketitle

{\small \noindent
$^1$ Universit\'e Paris Cit\'e, CNRS, Laboratoire de Probabilit\'es, Statistique et Mod\'elisation, LPSM, F-75013 Paris, France,\\
$^2$ Univ Brest, CNRS UMR 6205, Laboratoire de Mathématiques de Bretagne Atlantique, France \\
$^3$ Sorbonne Universit\'e, CNRS, Laboratoire de Probabilit\'es, Statistique et Mod\'elisation, LPSM, 4 place Jussieu, F-75005 Paris, France,\\

E-mails: madhar@lpsm.paris, juliette.legrand@univ-brest.fr, maud.thomas@sorbonne-universite.fr}

\begin{abstract}
Risk management is particularly concerned with extreme events, but analysing these events is often hindered by the scarcity of data, especially in a multivariate context. This data scarcity complicates risk management efforts. Various tools can assess the risk posed by extreme events, even under extraordinary circumstances. This paper studies the evaluation of univariate risk for a given risk factor using metrics that account for its asymptotic dependence on other risk factors. Data availability is crucial, particularly for extreme events where it is often limited by the nature of the phenomenon itself, making estimation challenging. To address this issue, two non-parametric simulation algorithms based on multivariate extreme theory are developed. These algorithms aim to extend a sample of extremes jointly and conditionally for asymptotically dependent variables using stochastic simulation and multivariate Generalised Pareto Distributions. The approach is illustrated with numerical analyses of both simulated and real data to assess the accuracy of extreme risk metric estimations.
\end{abstract}


\section{Introduction}\label{sec:Intro}

Risk management is of crucial importance in various sectors. It involves the identification, assessment, monitoring and mitigation of potential risks, regardless of the area of application. Risk is a common factor to various domains, yet its manifestation is specific to each sector. In climatology, meteorological and marine hazards such as droughts, floods and landslides can cause significant material damages and endanger the affected areas. The identification and anticipation of these events is thus of paramount importance in order to prevent them.
In finance, market movements can result to substantial losses. It is thus essential to anticipate and mitigate these risks, which can be achieved through in-depth analysis of market trends and the implementation of tailored risk management strategies. In the sequel, any variable that could result in a loss or damage is referred to as a risk factor.  For instance, in climatology, a risk factor could be any physical quantity, such as wave heights \citep[][]{legrand2023joint}, wind gusts \citep[][]{goegebeur2024dependent} or precipitation levels \citep[][]{grundemann2023extreme}.  In finance, risk factors are typically market parameters,  such as interest rates or exchange rates, which may induce potential losses for financial institutions if they experience unfavourable fluctuations \citep[see e.g.][]{mcneil2015quantitative}. The risk factors can be represented by random variables that quantify the magnitude of potential losses. It is typically the case that the magnitudes with the strongest impact result from events that occur with a very low probability of occurrence. This type of event is often referred to as a tail risk \citep{mcneil2015quantitative}. A primary objective of risk managers is to quantify, for a given target risk factor, the associated tail risk.

Modern risk management has witnessed a surge in the development of tail risk measures and methodologies for their estimation. The simplest and most common measure often considered for evaluating tail risks is the Value at Risk (VaR), which represents the quantile at a level $\alpha \in (0,1)$ of a given random variable. Nevertheless, \citet{artzner1999coherent} have identified significant shortcomings of the $\VaR$ approach, indicating that tail risk assessment through $\VaR$ estimation could result in an inaccurate quantification of the risk. 
 Several limitations of the $\VaR$ have been identified in the literature. In particular, this risk metrics does not account for the severity of the target risk factor, which may be a loss in financial applications \citep{yamai2005value}, or the magnitude of a climatic variable in environmental applications, e.g. the magnitude of daily precipitation \citep{grundemann2023extreme}.
In the literature, a number of alternative approaches have been proposed \citep{bellini2017risk,daouia2018estimation,tasche2002expected}. In this paper, we focus on three risk metrics (introduced in Section~\ref{sec:TRM}), referred to as tail risk metrics (TRMs), which are defined conditionally on the tail distribution of a given random variable, or target risk factor. 

Extreme Value Theory (EVT) provides the theoretical foundations for the analysis and modelling of tail distributions, making it an ideal framework for the estimation of TRMs \citep[see e.g.][]{resnick2007heavy}. In this paper, we propose two non-parametric simulation schemes for multivariate extremes by extending the procedures developed by \citet{legrand2023joint} in the bivariate case to higher dimensions.  The first algorithm---joint simulation of multivariate extremes---allows for the estimation of TRMs and the second one---conditional simulation of multivariate extremes---enables the inference of quantities involving some conditional tail distribution. 

Several parametric approaches based on EVT have been proposed for the estimation of TRMs \citep[see e.g.][]{mcneil2015quantitative,singh2013extreme}. However, the parametric setting is quite restrictive in the multivariate framework. Indeed, in contrast to the univariate framework,  the limiting distribution of multivariate extremes is no longer parametric \citep{rootzen2006multivariate}. To address this issue, \citet{rootzen2018multivariate2} have introduced several parametrisations and stochastic representations that enable the calibration of parametric models whose densities could be expressed through explicit formulas. Given the diversity of potential models, which may include nested models, model selection is necessary and can prove to be challenging and time-consuming. An alternative approach would be to model the univariate margins using univariate extreme value theory, and then  to model the dependence structure of extremes through the typical tools for dependence modelling, such as copulas. For extreme modelling, Clayton and Gumbel copula allow for the modelling of lower and upper tail dependence, respectively and are characterised by a single parameter  \citep[see e.g.][]{nelsen2006introduction}. This implies uniform dependence across the considered risk factors, which may not be a valid assumption. Another limitation of empirical estimates is that for high risk scenarios, the number of available tail observations becomes exceedingly limited.

In this paper, we thus rely on a non-parametric approach to expand the number of observations above the extreme level of interest based on joint simulation of multivariate extremes, by generalising the approach of \citet{legrand2023joint}.  Consequently, in addition to ensuring more reliable estimations, it is also possible to extrapolate beyond the range of observed data, which is not feasible when considering only historical data, and to go beyond the bivariate case, which is typically discussed in the literature.

A by-product of this first procedure concerns the estimation of quantities describing the tail of some conditional distribution, which could be analogous to an extremal regression approach. Existing approaches that address the issue of linear regression in the context of extreme values include quantile regression \citep[see e.g.][for recent developments]{pasche2022neural, gnecco2022extremal}. However, the performance of these approaches is highly dependent on the available data. It is crucial to acknowledge that the inherent limitations of the existing methodologies can be attributed to two key factors: the dimensionality of the random vector of risk factors and the number of available observations of this random vector. This latter issue is, in fact, the very  same problem that was encountered for TRMs estimation. However, the joint simulation algorithm does not apply in this context. Once more, we build upon the work of \citet{legrand2023joint}, who introduced a second algorithm for conditional simulation of bivariate extremes. In this paper, we extend their approach to higher dimensions. This second non-parametric algorithm enables the generation of new samples of a given random variable conditionally on joint event of simultaneous extreme risk factors, thus facilitating the estimation of quantities involving the conditional distribution of a  target risk factor given the observation of other risk factors with which the target risk factor has a strong tail dependence.

The outline of this paper is as follows. After a formal presentation of the TRMs of interest in Section~\ref{sec:TRM}, the joint and conditional simulation algorithms for multivariate extremes are developed in Section~\ref{sec:algos}. Their performances are evaluated first on synthetic data in Section~\ref{sec:DataSynthe} and then on real data in Section~\ref{sec:DataReal}.

\paragraph{Notations.} Throughout the paper, we use the following notations for vectors. Symbols in bold denote $d$-dimensional vectors.  For example, a $d$-dimension random vector is denoted by $\bm a = \left(a_1,\ldots,a_d\right)$, and $\bm a_{-j}$ denotes the vector $\bm a$ deprived from its $j$-th component for $j=1,\ldots,d$.  Operations and relations are meant component-wise, that is, for example, $\bm a\geq\bm b = \left(a_1\geq b_1,\ldots,a_d \geq b_d\right)$. 

\section{Tail Related Risk metrics}\label{sec:TRM}

In this section, we formally define the three TRMs considered in this paper. 
Hereinafter, $\bm X=(X_1,\ldots, X_d)$ denotes a $d$-dimensional vector with marginal continuous probability density function $f_j$, for $j=1,\dots,d$.

As all TRMs are based on the $\VaR$,  it is first necessary to recall its definition. The $\VaR$ at level $\alpha\in (0,1)$ of $X_j$, denoted $\textrm{VaR}_{\alpha}(X_j)$, is defined as the $(1-\alpha)$-quantile of the distribution of $X_j$, that is 
\[
\VaR_{\alpha}(X_j) = \inf\{x \in \mathbb{R} : \mathbb{P} \left(X_j > x\right)\leq \alpha\} \, 
\]
with $\alpha=\alpha_n \xrightarrow[n \to \infty]{} 0$. The $\VaR_\alpha$ represents the quantity that will be exceeded with probability $\alpha$. In other words, it corresponds to an extreme quantile of $X_j$. It should be noted that the term ``Value at Risk'' is primarily used in the field of financial risk management \citep[e.g.][]{mcneil2015quantitative}; in other domains, it is often referred to as the ``return level'' \citep[for more details, see][]{coles2001}.

 Several limitations of the $\VaR$ have been identified in the literature. In particular, this TRMs does not account for the severity of the random variable $X_j$, which may be a loss in financial applications \citep{yamai2005value}, or the magnitude of a climatic variable in environmental applications, e.g. the magnitude of daily precipitation \citep{grundemann2023extreme}. To address this shortcoming, alternative TRMs have been proposed, including the {\it expected shortfall} ($\ES$) \citep{artzner1999coherent}. The $\ES$  at level $\alpha\in(0,1)$ is defined as 
\begin{equation}\label{eq:ES}
    \ES_{\alpha}(X_j) = \mathbb{E}\left[X_j \mid X_j>\VaR_{\alpha}(X_j)\right] = \frac{1}{\alpha}\int_{\VaR_{\alpha}(X_j)}^{\infty}x f_j(x)\d x \,. 
\end{equation}
In words, the $\ES_\alpha$ corresponds to the expected loss above the $\VaR_\alpha$  \citep[for more details, see e.g.][]{artzner1999coherent,rockafellar2000optimization}. The $\ES$ is a useful tool for gaining insights into the severity of losses above a high quantile (i.e. the $\VaR$). It is therefore commonly used 
 in financial institutions for the calculation of the minimum capital requirements for market risk as specified by financial regulators \citep[see e.g.][]{basel2013}. Given its high profile, numerous techniques proposed for $\ES$ estimation have been proposed  \citep[see e.g.][for a comprehensive review]{nadarajah2014estimation}. In a parametric modelling approach, $\ES$ is computed as the empirical mean value of observations above the $\VaR$ which is estimated as the quantile of an extreme value distribution \citep[see e.g.][]{coles2001}. Subsequently, more advanced estimation approaches were suggested \citep[see e.g.][]{mcneil2000estimation, singh2013extreme}. One limitation of $\ES$ lies in its univariate risk assessment that ignores the potential asymptotic dependence that $X_j$ could exhibit with other risk factors. To address this issue, we consider two alternative risk metrics accounting for the asymptotic dependence in their univariate risk evaluation of a target risk factor $X_j$.

We thus introduce a second TRM, namely the {\it marginal expected shortfall} (MES) that has been intensively studied \citep[see e.g.][]{cai2015estimation,acharya2017measuring}. \citet{cai2015estimation} define the MES as the following conditional expectation $\mathbb{E}\left[X_j \mid S(\bm X) \geq v\right]$, where $S$ is a given statistic of the random vector $\bm X$ (e.g. the sum, the minimum, the maximum of $\bm X$)  and $ v\in\mathbb{R}$ a threshold defining the occurrence of an extreme event.  Building on the aforementioned definition, we propose to study a novel version by considering a conditioning on a joint event of simultaneous extreme risk factors deprived from the target risk factor $X_j$. That is, we define the {\it multivariate  marginal expected shortfall} ($\MES$) of $X_j$ at level $\alpha$ as follows
\begin{equation}\label{eq:MES}
    \MES_{\alpha}(X_j;\bm X) = \mathbb{E}\left[X_j \mid \bm X_{-j}\geq \bm v^\alpha_{-j}\right] = \int_{\mathbb R} x_j f_{X_j\mid\bm X_{-j}\geq \bm v^\alpha_{-j}}(x_j)\d x_j, 
\end{equation}
where $\bm v^\alpha =\VaR_\alpha(\bm X)\in\mathbb{R}^{d}$ and $f_{X_j\mid\bm X_{-j}\geq \bm v^\alpha_{-j}}$ denotes the conditional density of $X_j$ given the event $\bm X_{-j}\geq \bm v^\alpha_{-j}$. Through this formulation the aim is to capture the behaviour of $X_j$ when similar risk factors reach extreme levels.

In a bivariate setting, authors have investigated the incorporation of a dependent risk factor in order to improve the quantification of tail risk associated with a target  risk factor \citep[see e.g.][]{Josaphat2021,goegebeur2024dependent}. In this paper, we adopt the definition of {\it dependent conditional tail expectation} ($\DCTE$), as defined in the bivariate case by \citet{goegebeur2024dependent}. We extend this definition to the general multivariate case as follows 
\begin{equation}\label{eq:DCTE}
    \DCTE_{\alpha}(X_j; \bm X) = \mathbb{E}\left[X_j | \bm X\geq \bm v^\alpha \right]= \int_{v^\alpha_j}^{\infty}x_j f_{X_j \mid \bm X\geq \bm v^\alpha }(x_j)\d x_j.
\end{equation}
This third TRMs quantifies the risk of a target  risk factor when similar risk factors and the target risk itself reach extreme levels simultaneously. This metric could be relevant in a wide range of applications, including wind gust analysis \citep{goegebeur2024dependent} and car insurance policies \citep{syuhada2022estimating}.

  \citet{Josaphat2021} have proposed a parametric estimator of $\DCTE$ in the bivariate setting based on a model combining a Pareto distribution and a specific copula. However, the conditioning event  considered by \citet{Josaphat2021} constrains the fluctuations of $X_j$ to a specific interval, which hinders the extrapolation of the estimates beyond the observed data range.   \cite{syuhada2022estimating} have showed that the non-parametric estimators were more accurate than parametric ones at estimating DCTE. Finally, \citet{goegebeur2024dependent} have introduced an estimator using EVT arguments through a two-step approach, where the $\DCTE$ is first estimated at some intermediate level in order to then extrapolate the estimation at extreme levels. However, this approach is limited to the bivariate case.

In this paper, we propose two non-parametric simulation schemes for multivariate extremes. The first scheme is for the joint simulation of multivariate extreme events, which will  allow for the empirical estimation of our three TRMs of interest (Equations~\eqref{eq:ES},~\eqref{eq:MES} and~\eqref{eq:DCTE}).  The second scheme is for the conditional simulation of multivariate extreme events, which will be applied to the inference of quantities involving some conditional tail distribution.

\section{Extreme Scenario Simulation} \label{sec:algos}

This section outlines the two simulation algorithms developed in this study. These algorithms are natural extensions of the algorithms  presented in  \citep{legrand2023joint}, where the analysis is restricted to the bivariate setting.  

The  simulation procedure presented in \cite{legrand2023joint} is based on the stochastic representation of a standard multivariate generalised Pareto (MGP) vector proposed by \citet{rootzen2018multivariate2}. Building on the stochastic representation presented in Section~\ref{sec:MGP}, we propose two simulation algorithms. A joint simulation algorithm generating simulated samples of a MGP vector, while a conditional simulation algorithm is presented for simulating samples of a component of the MGP vector conditionally on the other components. Both algorithms are illustrated through numerical applications on synthetic data in Section~\ref{sec:DataSynthe}, and on real data in Section~\ref{sec:DataReal}.

\subsection{Multivariate Generalised Pareto Vectors}\label{sec:MGP}

Let $\bm X$ be a $d$-dimensional random vector. For a vector of suitably chosen thresholds $\bm u \in \mathbb R^d$, we define the vector of excesses above $\bm u$ as $\bm X - \bm u$ given that $\bm X \not \leq \bm u$, where $\bm X \not \leq \bm 0$ means that at least one of the components of $\bm X$ is positive. If $\bm X \not \leq \bm u$, $\bm X$ is said to be extreme. 
\citet{rootzen2006multivariate} have shown that, under certain conditions, the conditional excesses $\bm X - \bm u \mid \bm X \not \leq \bm u$ converge asymptotically to a non-degenerate distribution $H$ that belongs to the family of MGP distributions as $\bm u \to \bm \infty$. As in the univariate EVT \citep[e.g.][]{coles2001}, for $\bm u$ large enough, we may approximate the conditional excesses, namely $\bm X -\bm u\mid \bm X \not \leq \bm u$, by a random vector $\bm Z$ which is distributed as a MGP distribution, $H$.

For such vector $\bm Z$, the marginals $H_j$, $j=1,\ldots, d$, are in general not univariate generalised Pareto (GP) distributions. However, their restrictions to the positive subset are GP distributed \citep[see e.g.][]{rootzen2018multivariate2}. That is, for all $j=1,\ldots, d$,
\begin{equation}\label{eq:MGPmargin}
H_j^+(x) = \mathbb{P} \left[Z_j > x \mid Z_j >0 \right] = \left(1 + \gamma_j x/\sigma_j\right)_+^{-1/\gamma_j}, \mbox{ for }  x\geq 0 \mbox{ such that } \sigma_j+\gamma_j x>0,
\end{equation}
where $a_+ = \max(a,0)$, $\sigma_j>0$ is the j-th marginal scale parameter and $\gamma_j \in \mathbb R$ is the j-th marginal shape parameter. If $\gamma_j = 0$, we use the classical convention by taking the limit as $\gamma_j\to 0$ in Equation~\eqref{eq:MGPmargin}.

In the sequel, we denote $\bm \sigma = (\sigma_1,\ldots, \sigma_d)$ and $\bm \gamma = (\gamma_1,\ldots, \gamma_d)$ the vectors of marginal scale and shape parameters of a MGP vector. When $\bm \sigma = \bm 1$ and $\bm \gamma = \bm 0$, we say that $\bm Z$ is distributed according to a standard MGP distribution. Conversely, a standard MGP vector can be transformed into a general MGP vector using the following simple transformation: $Y$ is a MGP vector with parameters $(\bm \sigma,\bm\gamma)$ if and only if $Y=\bm \sigma (e^{\bm\gamma \bm Z}-1)/\bm\gamma$, with $\bm Z$ a standard MGP vector. 
In light of this result, we can therefore concentrate solely on the special case of standard MGP vectors. 

Note that any vector $\bm X$ can be transformed into a standard MGP vector by following the two steps outlined below:
\begin{enumerate}
  \item Standardise the data $\bm X$ to exponential margins $\bm X^{E}$, using the probability integral transform on each component
  \[
  X^E_j = \log \left(1- F_{j}(X_j)\right), \quad \mbox{for } j=1,\ldots,d, 
  \]
  where $F_{j}$ is the cumulative distribution function (c.d.f.) associated with the $j$-th risk factor for $j=1,\ldots,d$. $F_j$ can be either a parametric c.d.f. that fits the data or the empirical version of c.d.f.. 
  \item Compute the multivariate excess vector as defined in \cite{rootzen2006multivariate} as follows \begin{equation}\label{eq:Z_XE}
      \bm Z = \bm X^E - \bm u^E \mid \bm X^E  \not \leq \bm u^E \, ,
  \end{equation} 
where $\bm u^E$ is a suitably chosen threshold on the exponential scale. 
\end{enumerate}

The choice of the thresholds $u_j$, for $j=1,\ldots,d$ can be understood as a compromise between bias and variance. The smaller the threshold, the less valid the asymptotic approximation, which leads to bias. On the other hand, a threshold  too high  will generate few excesses to fit the model, which leads to high variance. In practice, threshold selection is a challenging task. The existing methods for the choice of the threshold $u$ rely on graphical diagnostics or on computational approaches based on supplementary conditions (that depend on unknown parameters) on the underlying distribution function $F$  \citep[see][]{scarrott2012review}.

\cite{rootzen2018multivariate2} have derived different stochastic representations of MGP distributions for which explicit density formulas can be obtained \citep[see also][]{kiriliouk2019peaks}. In particular, they have shown that a standard MGP vector $\bm Z$  can be decomposed as follows
\begin{equation}\label{eq:SMGPRep}
    \bm Z = E + \bm T - \max\left(\bm T\right),
\end{equation}
where  $E$ is a unit exponential variable and $\bm T$ a $d$-dimensional random vector independent of $E$. 

In Equation~\eqref{eq:SMGPRep}, it is necessary to ensure that the components of  the original vector $\bm X$ are asymptotically dependent \citep[see e.g.][]{coles2001}, simply meaning that the large values of the components $X_j$ occur simultaneously. In the applications that we have in mind, this hypothesis will be verified since the risk factors that will be considered are precisely those exhibiting asymptotic dependence. 

From Equation~\eqref{eq:SMGPRep}, we derive the cornerstone of our two simulation algorithms by simply considering the following multivariate difference

\begin{equation}\label{eq:wk_w1}
    \Delta^{j,k}  = Z_j - Z_k= T_j- T_k, \mbox{ for all } j,k=1,\ldots, d.
\end{equation}

Subsequently, Equation~\eqref{eq:SMGPRep} can be rewritten as follows 
\begin{equation}\label{eq:MGPEq}
    Z_j = E +  \sum_{k=1,k\neq j}^d \Delta^{j,k} \prod_{\ell=1,\ell\neq k}^d \mathbf{1}_{\Delta^{\ell,k}<0}, \mbox{ for all } j=1,\ldots, d,
\end{equation}
where $\mathbf{1}_\cdot$ denotes the indicator function.

Furthermore,  for a fixed $q\in\{1,\dots,d\}$, the $q$-th vector of differences is defined as follows
\[
\bm\Delta^{(q)}=\left(\Delta^{q,1},\dots,\Delta^{q,d}\right)\in\mathbb{R}^{d}.
\]

With Representation~\eqref{eq:MGPEq} at our disposal, we may proceed to the simulation of a MGP vector $\bm Z$, via the simulation of $\bm \Delta^{(q)}$ for a fixed $q \in \{1,\dots,d\}$ using bootstrapping techniques combined with a rejection algorithm in some specific cases. 

\subsection{Non-parametric joint MGP simulation}\label{subsec:joint}

In this section, without any loss of generality, we fix $q$ equal to 1. We then describe the joint simulation procedure, which is a generalisation of the algorithm originally proposed by \citet{legrand2023joint} for the case where $d=2$. 

For a given sample $(\bm Z_i)_{1\leq i \leq n}$ of independent and identically distributed (i.i.d.)  replicates of a MGP distributed vector $\bm Z$, the joint MGP simulation algorithm (see Algorithm~\ref{alg:JointMGP}) performs the stochastic generation of $m\in\mathbb{N}$ new i.i.d. copies of $\bm Z$, denoted $(\widetilde{\bm Z}_\ell)_{1\leq \ell \leq m}$. 

Algorithm~\ref{alg:JointMGP} is as follows. First, we simulate unit exponential random variables. Then, independently on this first step, we simulate new realisations $(\widetilde{\bm \Delta}^{(1)}_\ell)_{1\leq \ell \leq m}$ of the differences $\bm \Delta^{(1)}$ using a non-parametric bootstrap approach, i.e. by resampling among the observed differences  $\bm \Delta^{(1)}$. Finally, we use the stochastic relation in Equation~\eqref{eq:MGPEq} to merge both simulation steps, in order to generate new realisations of MGP vectors, noting that $\Delta^{r,s}=  \Delta^{1,s}- \Delta^{1,r}$, for all $1\leq r,s \leq d$.

\begin{algorithm}
\caption{Non-parametric joint MGP simulation.}
\label{alg:JointMGP}
\begin{algorithmic}[1]
\Require Observations  $\left(\bm Z_i\right)_{1\leq i \leq n} = \left(Z_{i,1},\ldots,Z_{i,d}\right)_{1\leq i \leq n}$ from a standard  MGP distribution;
\Ensure{A standard MGP simulated sample $\left(\widetilde{\bm Z}_m\right)_{1\leq \ell \leq m} = \left(\widetilde{Z}_{\ell,1},\ldots,\widetilde{Z}_{\ell,d}\right)_{1\leq \ell \leq m}$}
\Procedure{}{}
\State $\Delta_{i}^{1,k} \leftarrow Z_{i,1}-Z_{i,k}$, for $1\leq i\leq n$ and $1\leq k \leq d$
\State Generate $E_1,\ldots, E_m \overset{\text{i.i.d.}}{\sim} \mathrm{Exp}(1)$, and independent of $(\Delta_{i}^{1,k})_{1\leq i\leq n, 1\leq k \leq d}$ 

\State Generate a $m$-bootstrap sample  $\widetilde{\bm \Delta}^{(1)} = \left(\widetilde{\bm\Delta}^{(1)}_{\ell}\right)_{\ell=1,\ldots,m}$  from $\left(\bm \Delta_{i}^{(1)}\right)_{1\leq i\leq n}$
\State $\widetilde \Delta_{\ell}^{r,s} \leftarrow \widetilde \Delta_{\ell}^{1,s}-\widetilde \Delta_{\ell}^{1,r}$, for $1\leq \ell\leq m$ and all $1\leq r,s \leq d$
\State $\widetilde{Z}_{\ell,j} \leftarrow E_\ell + \sum_{s=1,s\neq j}^d \widetilde{\Delta}^{j,s}_{\ell} \prod_{r=1,r\neq s}^d \mathbf{1}_{\widetilde{\Delta}^{r,s}_{\ell}<0}$ for all $1 \leq \ell \leq m$ and $1\leq j \leq d$
\EndProcedure
\end{algorithmic}
\end{algorithm}

 Proposition~\ref{l:jointSim} (see Section \ref{sec:theo}) guarantees that the samples simulated with  Algorithm~\ref{alg:JointMGP} are indeed distributed according to a standard MGP distribution, for all $q \in \{1,\dots,d\}$, and in particular for $q=1$. The proof can be found in Section~\ref{sec:theo} with the statement of the proposition.

We then provide a numerical illustration of Algorithm~\ref{alg:JointMGP} in which we consider a 3-dimensional MGP vector $\bm Z$, with $\bm T$ (see Equation~\eqref{eq:SMGPRep}) distributed according to a centred multivariate Gaussian distribution with correlation coefficients $\rho_{1,2}=0.4,\rho_{1,3}=0.8$ and $\rho_{2,3}=0.1$.

\begin{figure}[h]
 \centering
  \begin{tabular}{ccc}\includegraphics[width=0.3\textwidth]{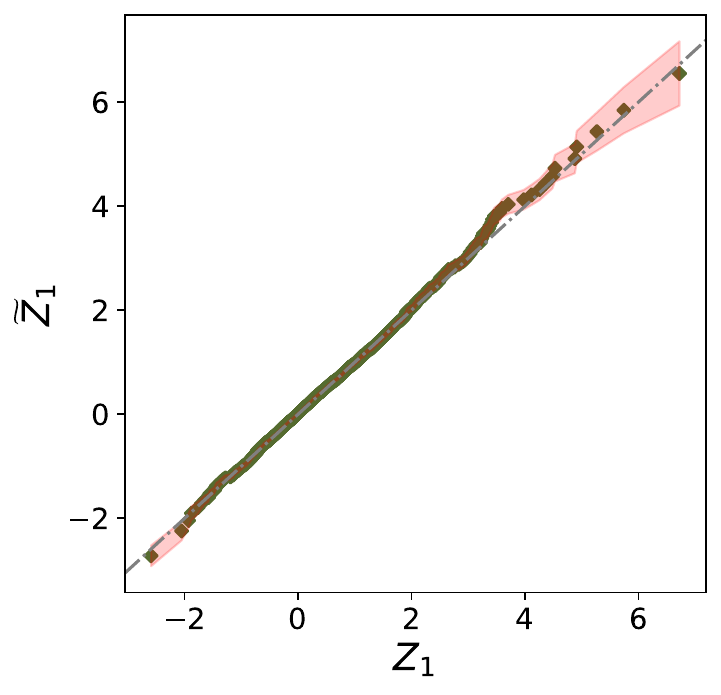} &  \includegraphics[width=0.3\textwidth]{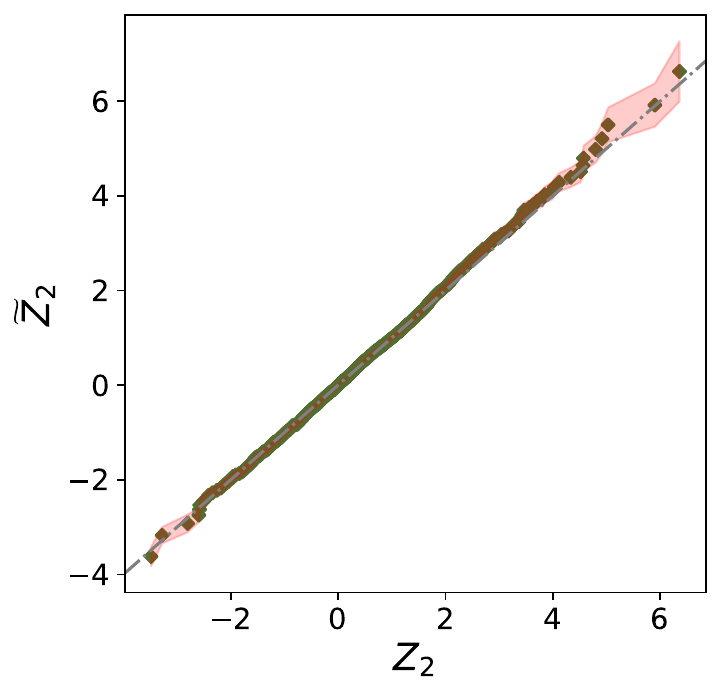} &
 \includegraphics[width=0.3\textwidth]{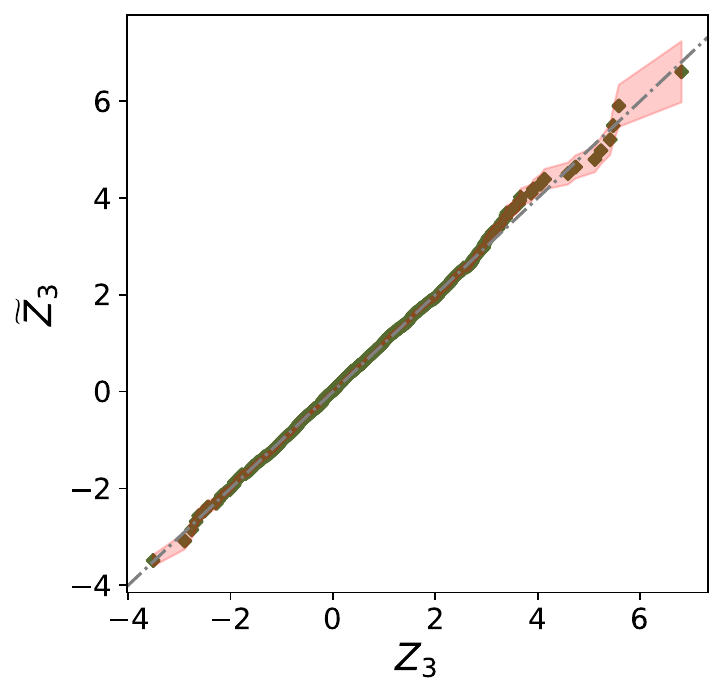} \\~\\
 a) & b)&  c) \\~\\
 
  \includegraphics[width=0.3\textwidth]{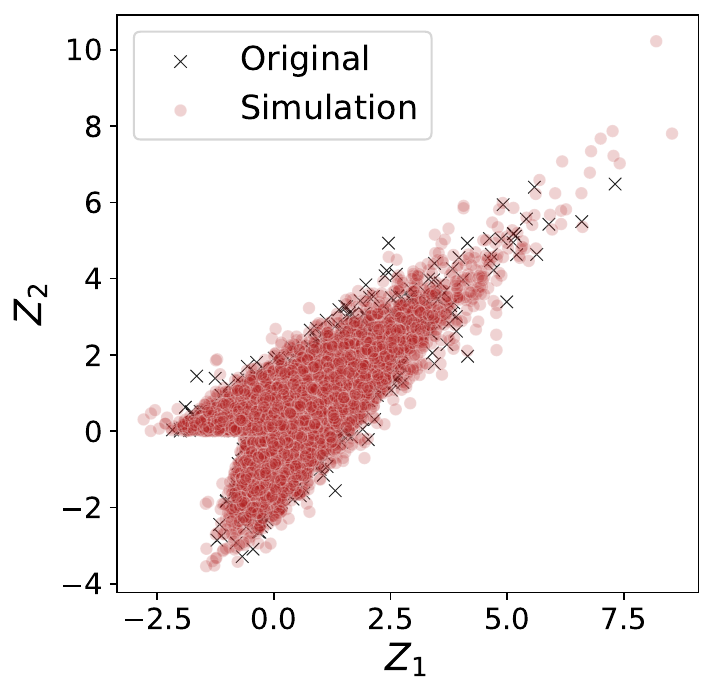}
& \includegraphics[width=0.3\textwidth]{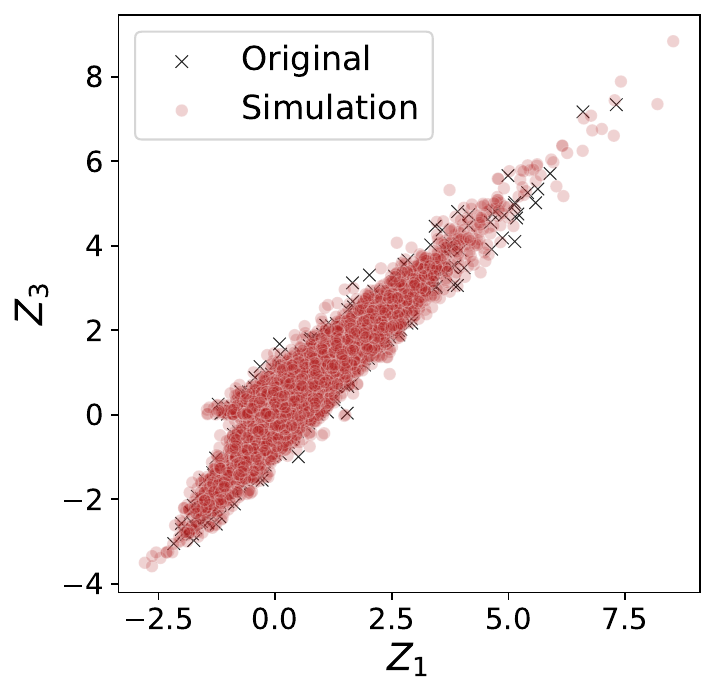}
&  \includegraphics[width=0.3\textwidth]{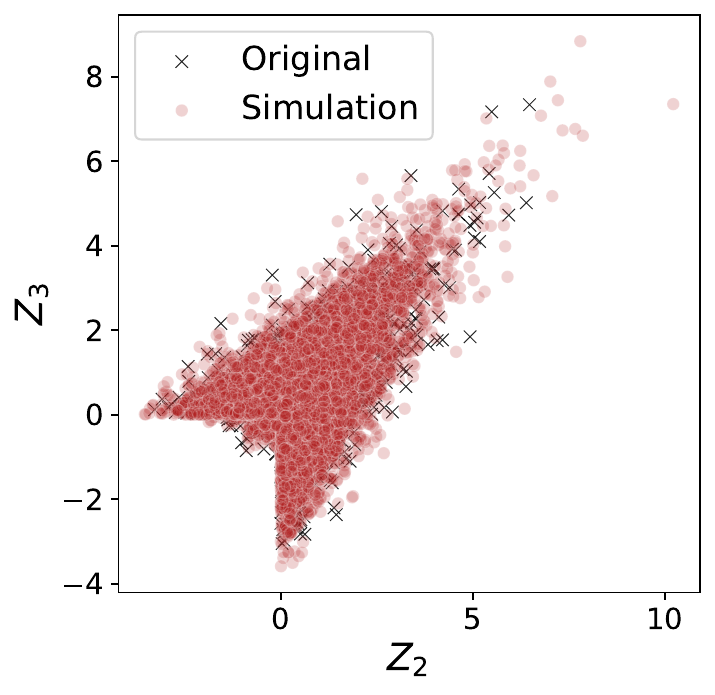}
\\~\\
 d) & e) & f) \\ 
\end{tabular}
\caption{Bivariate representations of the original sample of size 2,000 of a MGP vector $\bm Z \in \mathbb{R}^3$ with $\bm T$ (see Equation~\eqref{eq:SMGPRep}) distributed according to a centred multivariate Gaussian distribution with correlation coefficients $\rho_{1,2}=0.4,\rho_{1,3}=0.8$ and $\rho_{2,3}=0.1$ and the simulated sample $\widetilde{\bm Z}$ of size $10,000$. Figures a), b) and c) display the QQ plots of ($Z_1$, $\widetilde{Z}_1$), ($Z_2$, $\widetilde{Z}_2$) and ($Z_3$, $\widetilde{Z}_3$), respectively, and their associated 95\% point-wise confidence intervals based on $1,000$ bootstrap replications. Figures d), e) et f) represent scatter plots of the components  ($Z_1$, $Z_2$), ($Z_1$, $Z_3$) and ($Z_2$, $Z_3$) of the original sample (black crosses). On each scatter plot, the points of the simulated sample (red circles) are added.}
  \label{fig:Jointk3}
\end{figure}

Figure~\ref{fig:Jointk3} presents a bivariate representation of a simulated sample of size $10,000$ and the original sample of the 3-dimensional vector $\bm Z$  of size $2,000$. The top row (Figures a)-c)) of Figure~\ref{fig:Jointk3} depicts the quantile-quantile (QQ) plots  between the original observations and simulated observations for each component.
The QQ plots show a good fit between the marginal distributions of the simulated samples $\widetilde{Z}_j$ and the marginal distributions of the original sample $Z_j$ for $j \in \{1,2,3\}$. Consequently, the QQ plots indicate that the joint simulation algorithm is capable of generating samples with the same marginal distributions as the ones of the original sample. The scatter plots in the bottom row of Figure~\ref{fig:Jointk3} show the simulated samples of $\widetilde{\bm Z}$ in red and the original sample of $\bm Z$ in black. Each column of the second row focuses on a single plane of the $3$-dimensional space, and hence on a bundle of the component of $\bm Z$ and $\widetilde{\bm Z}$.  With regard to the dependence structure,  the scatter plots of Figure~\ref{fig:Jointk3} show that the shape displayed by the simulated sample seems to be similar to the one of the original sample. It can be observed that the simulated sample comprises a greater number of observations in the extreme regions.

\subsection{Non-parametric conditional MGP simulation}\label{sec:CondSim}

Thanks to the representation of the univariate components $Z_j$ of a MGP vector $\bm Z$ by Equation~\eqref{eq:MGPEq}, we were able to derive our joint simulation algorithm. In the same way, but using the definition of $\Delta^{j,k}$ in Equation~\eqref{eq:wk_w1}, we propose a second algorithm, namely the conditional MGP simulation algorithm (see Algorithm~\ref{algo:CondAlgo}). In fact, for practical application, we are also interested in generating observations of $Z_j$, when $\bm Z_{-j}$ is observed at some extreme level $\bm z_{-j}$. To be able to simulate $Z_j$ conditionally on the rest of the components $\bm Z_{-j}$, the conditional distribution of $\Delta^{q,j}$ given $\bm Z_{-j}$ must be computed for any fixed $q=1,\ldots,d$ with $q\neq j$. Then,  realisations of $\Delta^{q,j}$ are drawn, which can be used to compute the new observations of $Z_j$ according to Equation~\eqref{eq:MGPEq}.

The procedure of non-parametric conditional MGP simulation is based on the conditional distribution of $\Delta^{q,j}$ given $\bm Z_{-j}$.  Denoting $z_\star \coloneqq \max \bm z_{-j}$, the conditional distribution of $\Delta^{q,j}$ given $\bm Z_{-j}$ can then be expressed as follows 
\begin{enumerate}
    \item If $z_\star>0$ and $z_\star=z_q$, 
    \[
    f_{\Delta^{q,j} \mid \bm Z_{-j}=\bm z_{-j}}(\delta^{q,j} ) = \frac{1}{I_1(0) + I_2(0)} \left(\mathbf{1}_{\delta^{q,j}>0} + \e^{\delta^{q,j}}
         \mathbf{1}_{\delta^{q,j} \leq 0}\right)f_{\bm \Delta^{(q)}}\left(\bm{\delta}^{(q)} \right).
         \]
    \item If $z_\star>0$ and $z_\star\neq z_q$, denoting $\delta_\star \coloneqq z_q - z_\star$, 
    \[
    f_{\Delta^{q,j} \mid \bm Z_{-j}=\bm z_{-j}}(\delta^{q,j} ) =\frac{1}{I_1(z_q-z_\star) + I_2(z_q-z_\star)} \left( \e^{\delta^{q,j}}  \mathbf{1}_{\delta_{\star}>\delta^{q,j}} + \e^{\delta_{\star}}\mathbf{1}_{\delta_{\star}\leq\delta^{q,j}}\right) f_{\bm \Delta^{(q)}}\left(\bm{\delta}^{(q)} \right).
    \]
    \item If $z_\star\leq0$, 
\[
f_{\Delta^{q,j} \mid \bm Z_{-j}=\bm z_{-j}}(\delta^{q,j} ) = \frac{1}{I_1(z_q)} \e^{\delta^{q,j}}  f_{\bm \Delta^{(q)}}\left(\bm{\delta}^{(q)} \right) \mathbf{1}_{\delta^{q,j} <z_q}.
\]
\end{enumerate}
where \begin{itemize}
    \item  $\bm \delta^{(q)}=(\delta^{(q,1)},\dots,\delta^{(q,d)})$
    \item for any $x\in\mathbb{R}$, 
\[
I_1(x)\coloneqq \int_{-\infty}^x \e^{\zeta} f_{\bm \Delta^{(q)}}\left(\bar{\bm{\delta}}^{(q)} \right) \d \zeta \quad \text{ and } \quad
     I_2(x)\coloneqq \e^x\int_x^{\infty}  f_{\bm \Delta^{(q)}}\left(\bar{\bm{\delta}}^{(q)} \right) \d \zeta
 \]
\end{itemize}

where $\bar{\bm{\delta}}^{(q)}$ has the same components as $\bm{\delta}^{(q)}$ except for $\delta^{q,j}$ which is replaced by $\zeta= z_q - z_j$.

The theoretical result and the proof can be found in Section~\ref{subsec:theo:cond} (Proposition~\ref{prp:cond:deltaj}).  It can be observed that in Cases 1 and 3, the conditional distribution of $\Delta^{q,j}\mid\bm Z_{-j}= \bm z_{-j}$ is independent of $\bm z_{-j}$. Hence, in these specific cases, bootstrapping will be used for the conditional simulation of $\Delta^{q,j} \mid\bm Z_{-j}= \bm z_{-j}$. However, in Case 2 ($z_{\star}> 0$ and $z_{\star}\neq z_q$), the conditional distribution of $\Delta^{q,j}\mid\bm Z_{-j}= \bm z_{-j}$ does depend on $\bm z_{-j}$. In that case, bootstrapping needs to be combined with the rejection sampling method. 
However, it is not common practice to use the rejection algorithm to sample a univariate variable $\Delta^{q,j}$ through observations of random vector $\bm \Delta^{(q)}$. Therefore, Propositions~\ref{prop:RS} and~\ref{prop:RS2} are devoted to the rejection algorithm and its application in the conditional MGP simulation procedure. Finally, Algorithm~\ref{algo:CondAlgo} outlines the procedure in question.

\begin{algorithm}[H]
\caption{Non-parametric conditional MGP simulation}
\label{algo:CondAlgo}
\begin{algorithmic}[1]
\Require Observations $\left(\bm{\Delta}^{(1)}_i \right)_{1\leq i \leq n}$; $j \in \{2,\dots,d\}$ the index of risk factor of interest; a realisation $\bm z_{-j}$ of $\bm Z_{-j}$
\Ensure{A simulated sample $\left(\widetilde{Z}_{\ell,j}\right)_{1\leq \ell \leq m}$ conditionally on $\bm Z_{-j} =\bm z_{-j} $}
\Procedure{}{}
\If{$\exists k \neq j \mid z_k >0$}
    \If{$z_{\star}=z_1$}
    \State  Define $\bm{\Delta}^{(1)}_{| z_{\star}=z_1}$, the subset of $\left(\bm{\Delta}_i^{(1)} \right)_{1\leq i \leq n }$ such that $\max\bm z_{-j} = z_1$
    \State Bootstrap $m$ realisations $\left(\Delta_{\ell}^{1,j}\right)_{1\leq \ell \leq m}$, from $\left(\Delta^{1,j}_{| z_{\star}=z_1}\right)$ independently of $\bm Z_{-j}$
    \Else

    \For{$1\leq \ell \leq m$}
    \State Sample a realisation $\Delta^{1,j}_m$ from  $\left(\Delta^{1,j}_i \right)_{1\leq i \leq n}$ independently of $\bm Z_{-j}$
    \State Compute $\Delta_{\star}:=z_1 - z_{\star}$
    \State Generate a random number $u \in [0,1]$
    \While{$u> \left(e^{\Delta^{1,j}_\ell}  \mathbf{1}_{\Delta_{\star}>\Delta^{1,j}_\ell}+ e^{\Delta_{\star}} \mathbf{1}_{\Delta_{\star} \leq \Delta^{1,j}_\ell}\right)$}
     \State Repeat last two steps 
      \EndWhile
   \EndFor
\EndIf
\Else
\For{$1\leq \ell \leq m$}
\State Sample a realisation $\Delta^{1,j}_\ell$ from  $\left(\Delta^{1,j}_i \right)_{1\leq i \leq n}$ independently of $\bm Z_{-j}$
\State Generate a random number $u \in [0,1]$ 
\While{$u> e^{\Delta^{1,j}_\ell}  \mathbf{1}_{\Delta^{1,j}_\ell <z_q}$}
\State  Repeat last two steps 
\EndWhile
\EndFor
\EndIf
\State $\widetilde{Z}_{\ell,j} \leftarrow  z_1 - \Delta_{\ell}^{1,j}$ for all $1 \leq \ell \leq m$
\EndProcedure
\end{algorithmic}
\end{algorithm}

We illustrate the performances of the conditional MGP simulation Algorithm~\ref{algo:CondAlgo} through the same parametric example than in Section~\ref{sec:CondSim}, that is a 3-dimensional MGP vector $\bm Z$, with $\bm T$ (see Equation~\eqref{eq:SMGPRep}) distributed according to a centred multivariate Gaussian distribution, but with different correlation coefficients $\rho_{1,2}=0.6,\rho_{1,3}=0.8$ and $\rho_{2,3}=0.5$. We consider the conditional distribution of $Z_2 \mid \bm Z_{-2} = \bm z_{-2}$. The numerical experiment consists in generating realisations of the $2$-th component of a 3-dimensional standard MGP vector given observations $\bm z_{-2}=(z_1,z_3)$ for the original sample $\bm Z$ of size 2,000 and a sample simulated using Algorithm~\ref{algo:CondAlgo} of size 10,000.  We illustrate each case (Cases 1, 2 and 3) depending on $\bm z_\star$. For each simulation case, the empirical density of simulated samples of $Z_2\mid\bm Z_{-2}=\bm z_{-2}$ is compared to the theoretical density. Figure~\ref{fig:CondSimParametric} shows that regardless the value of $\bm z_{-2}$, the histogram shape and the kernel estimates of the density of the simulated samples of $Z_2\mid\bm Z_{-2}=\bm z_{-2}$ match systematically the theoretical density curve.

\begin{figure}[H]
    \centering
    \begin{tabular}{ccc}
        \includegraphics[width=0.3\linewidth]{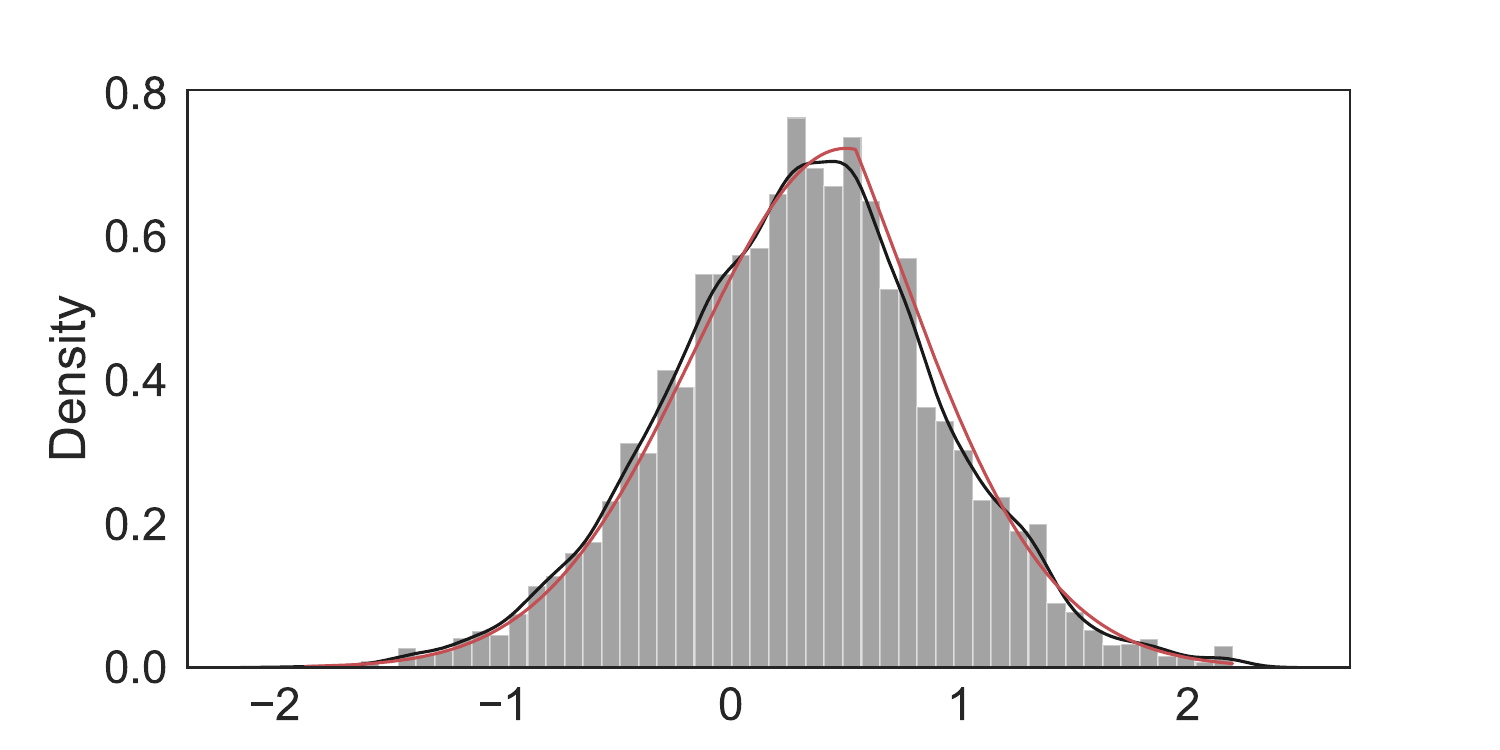}
   &
        \includegraphics[width=0.3\linewidth]{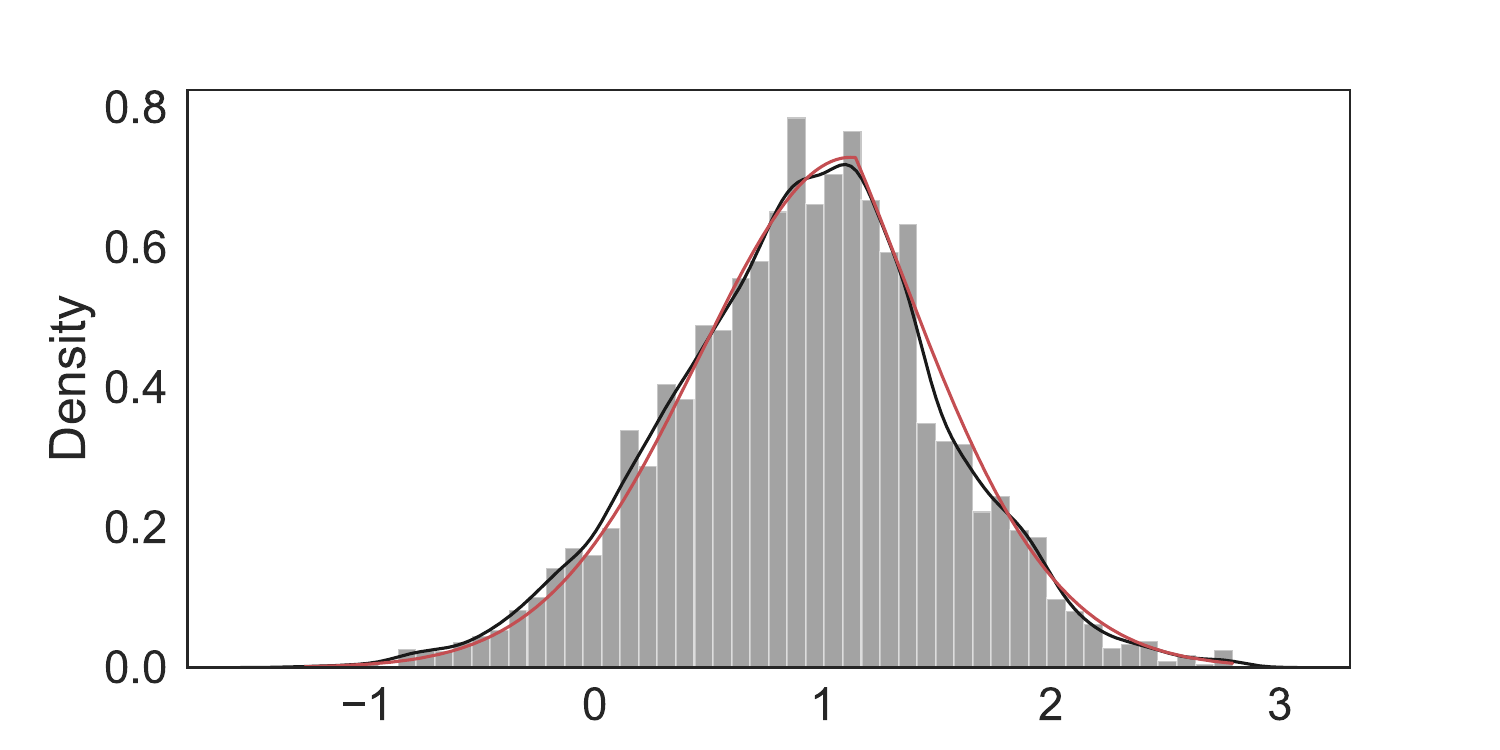}

 &
        \includegraphics[width=0.3\linewidth]{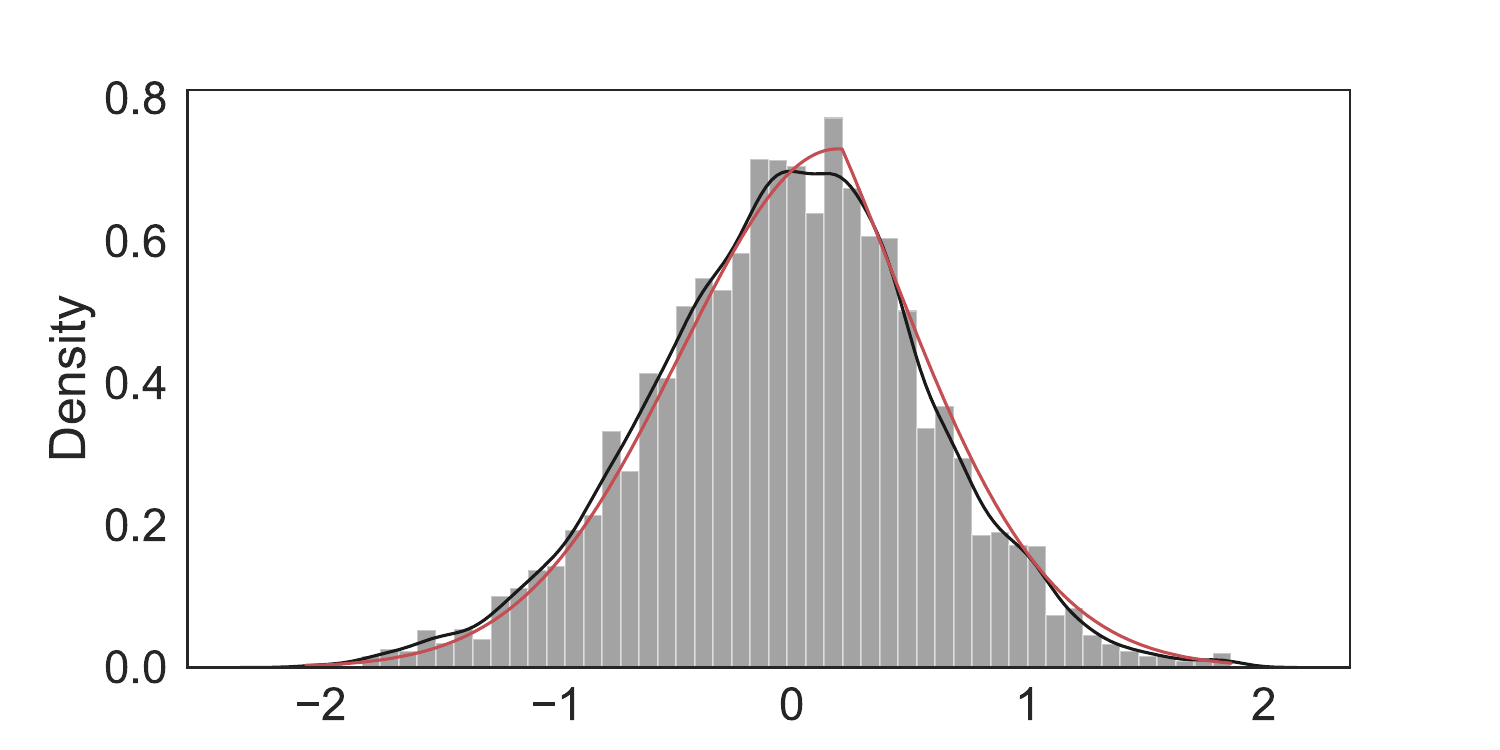} \\~\\
  $\bm z_{-2}=(0.54,0.31)$ &  $\bm z_{-2}=(1.14,1.03)$ & $\bm z_{-2}=(0.21,0.11)$
    \\~\\
  &   a) Case 1: $z_\star >0$ and $z_\star=z_1$& \\~\\

        \includegraphics[width=0.3\linewidth]{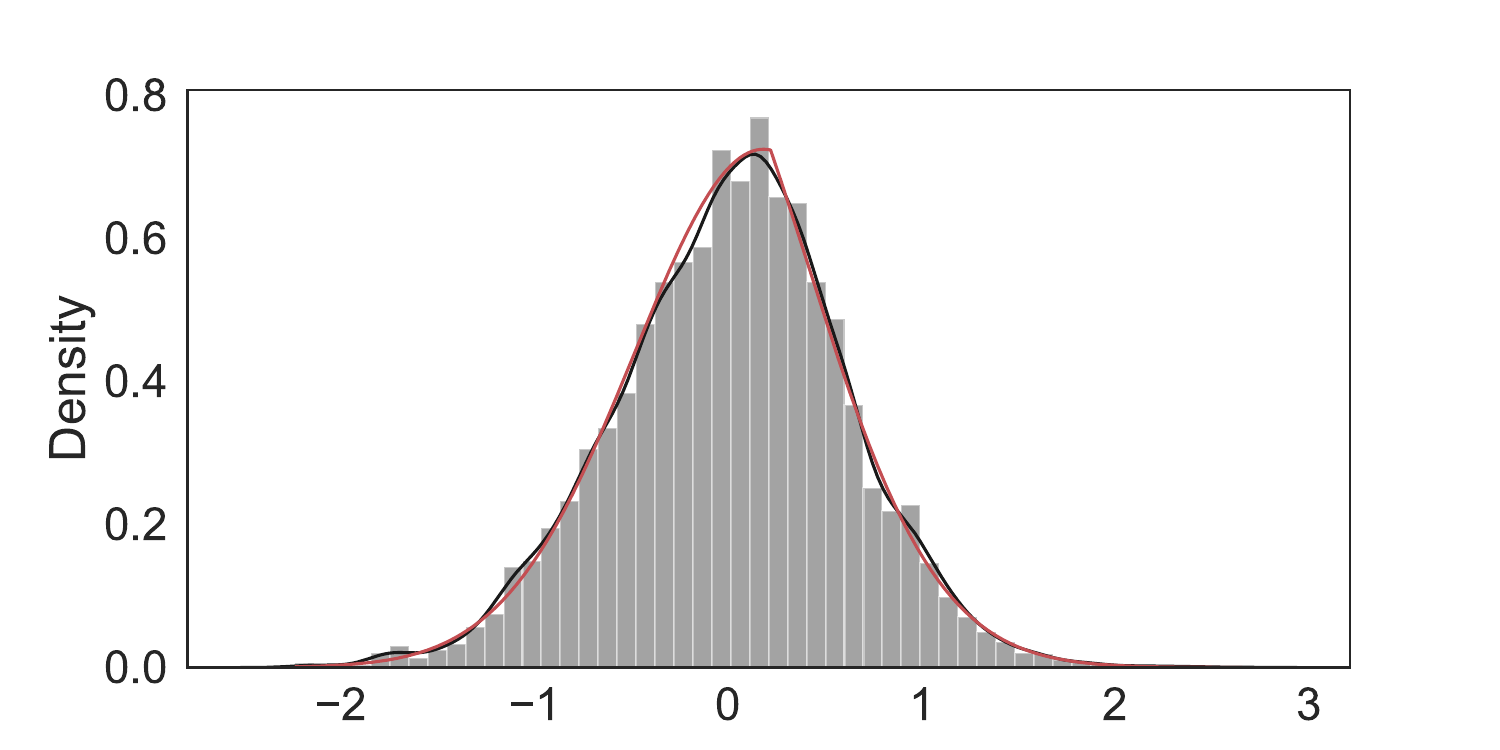}&

        \includegraphics[width=0.3\linewidth]{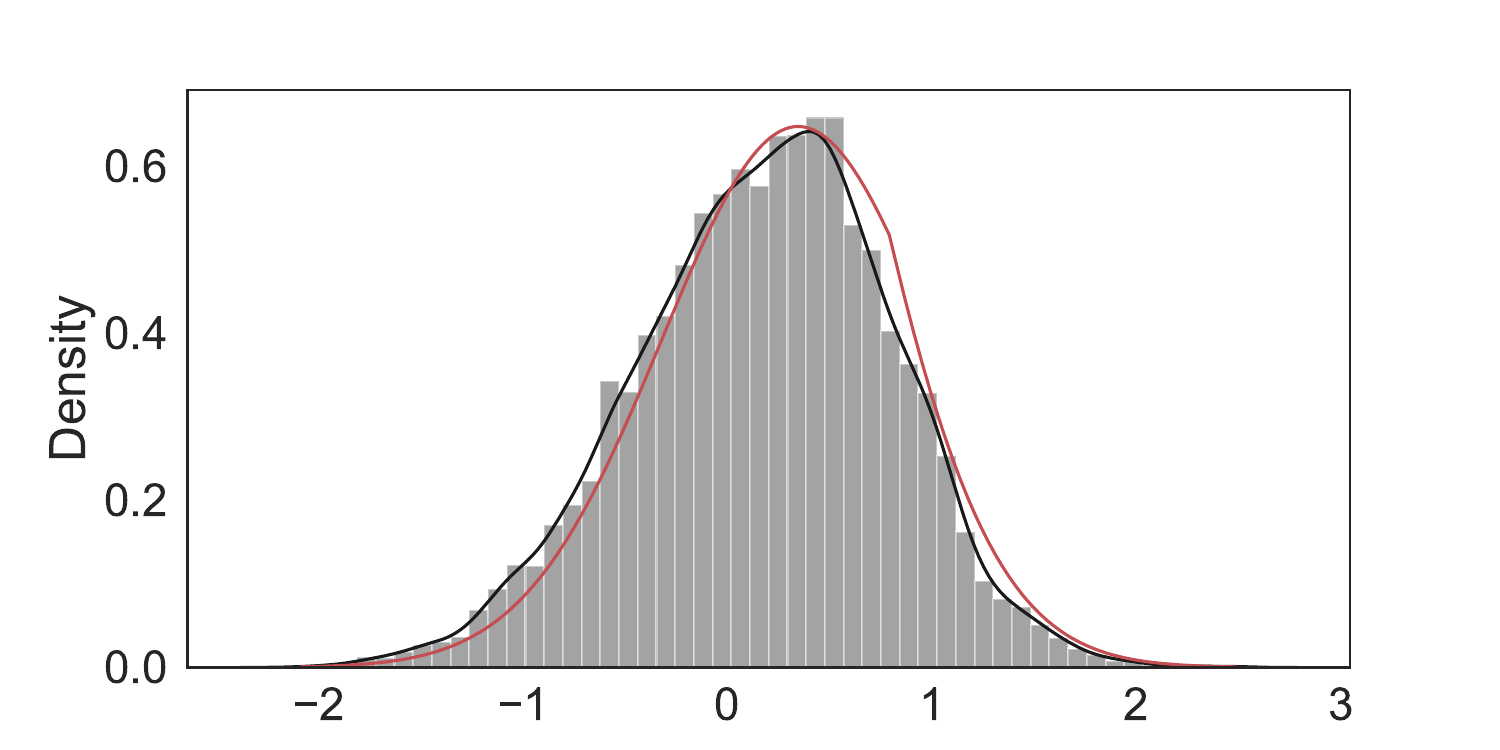} &

        \includegraphics[width=0.3\linewidth]{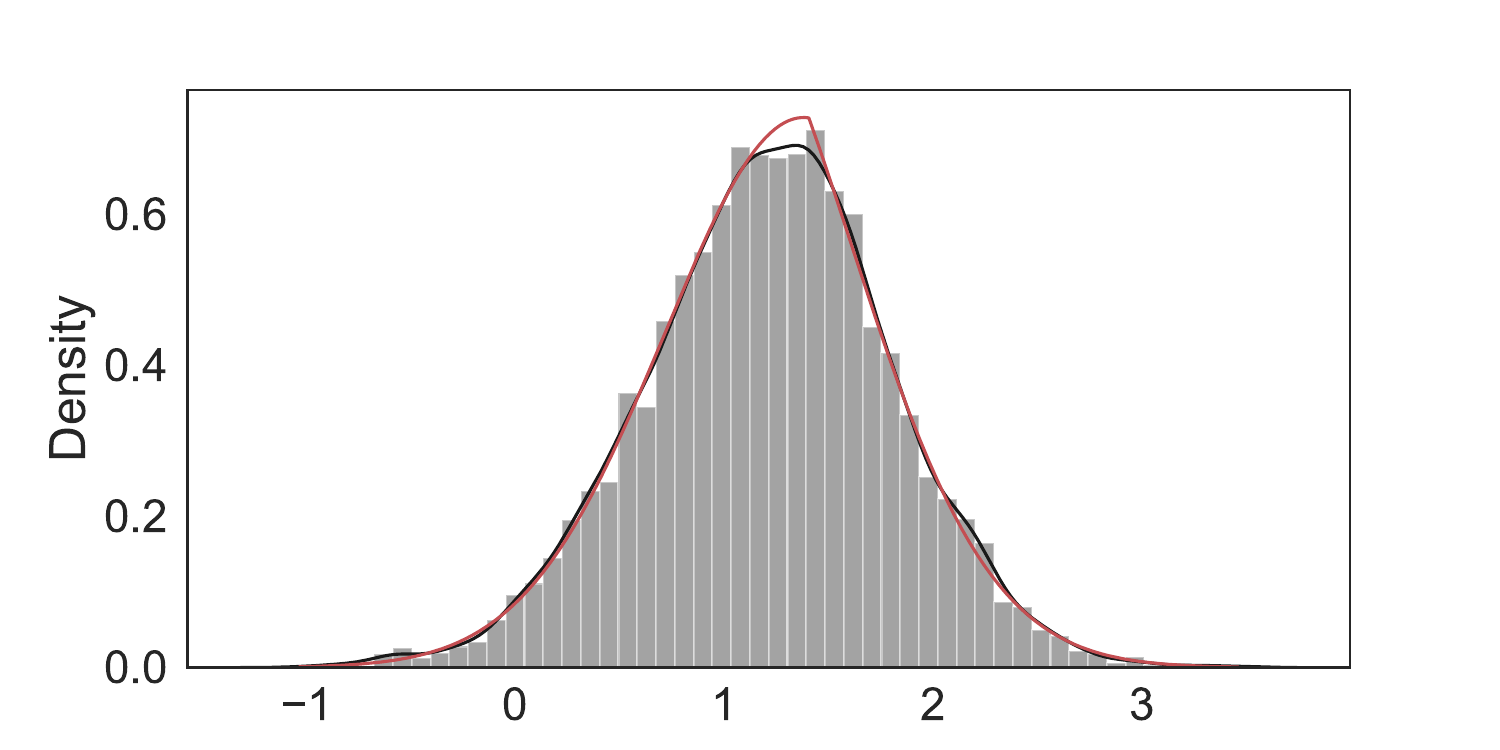} \\~\\

$\bm z_{-2}=(0.17,0.22)$ & $\bm z_{-2}=(0.24,0.79)$ & $\bm z_{-2}=(1.37,1.40)$ \\~\\
   & b) Case 2: $z_\star >0$ and $z_\star \neq z_1$ &\\~\\
    
        \includegraphics[width=0.3\linewidth]{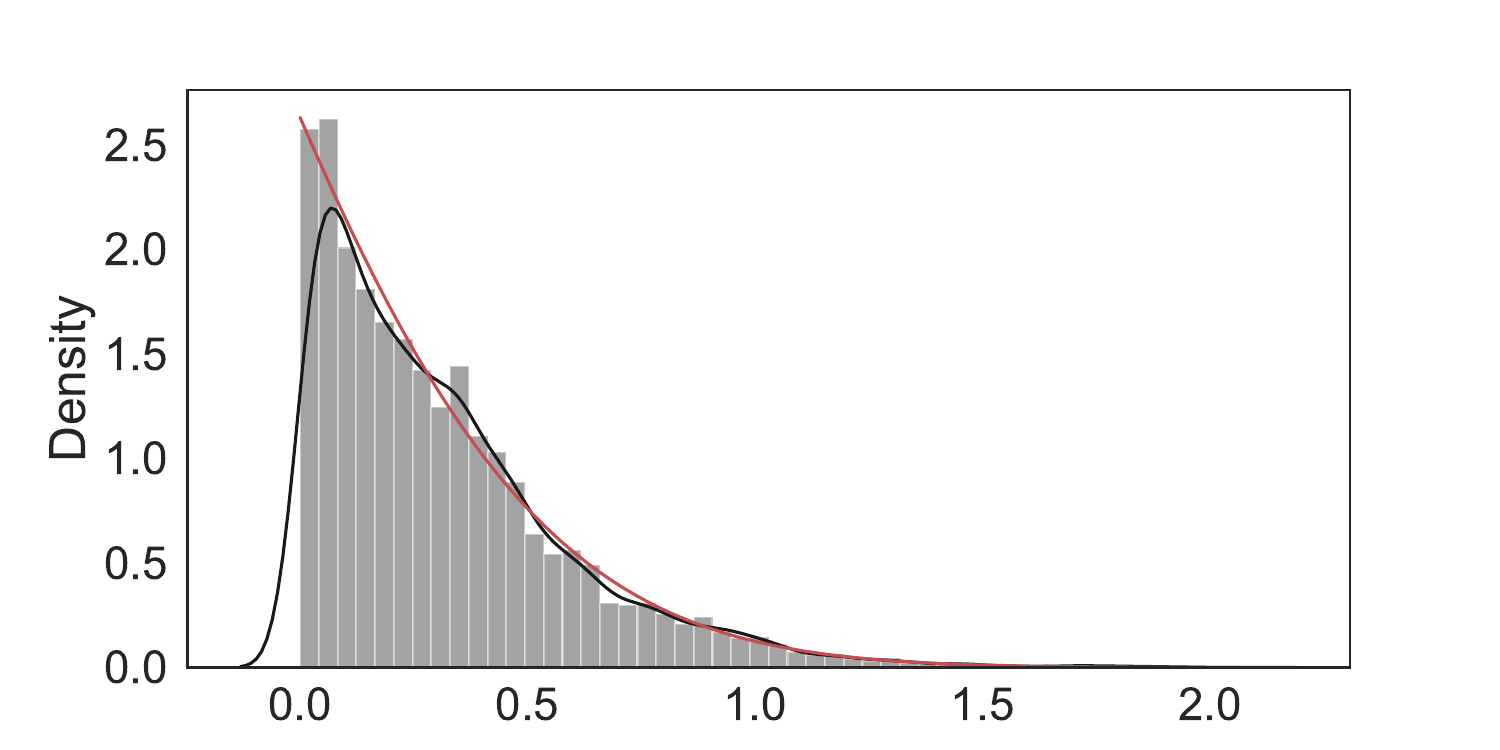} &

        \includegraphics[width=0.3\linewidth]{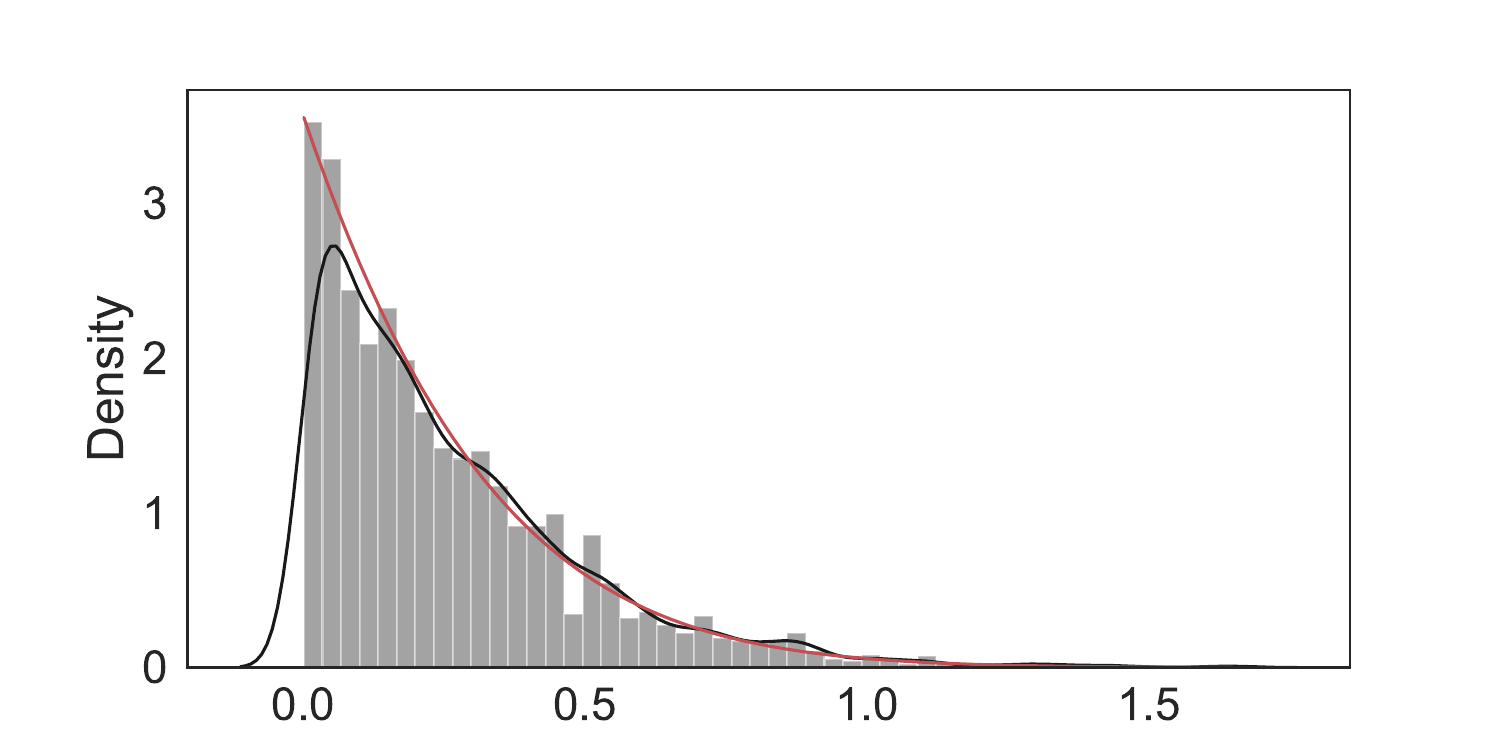} &

        \includegraphics[width=0.3\linewidth]{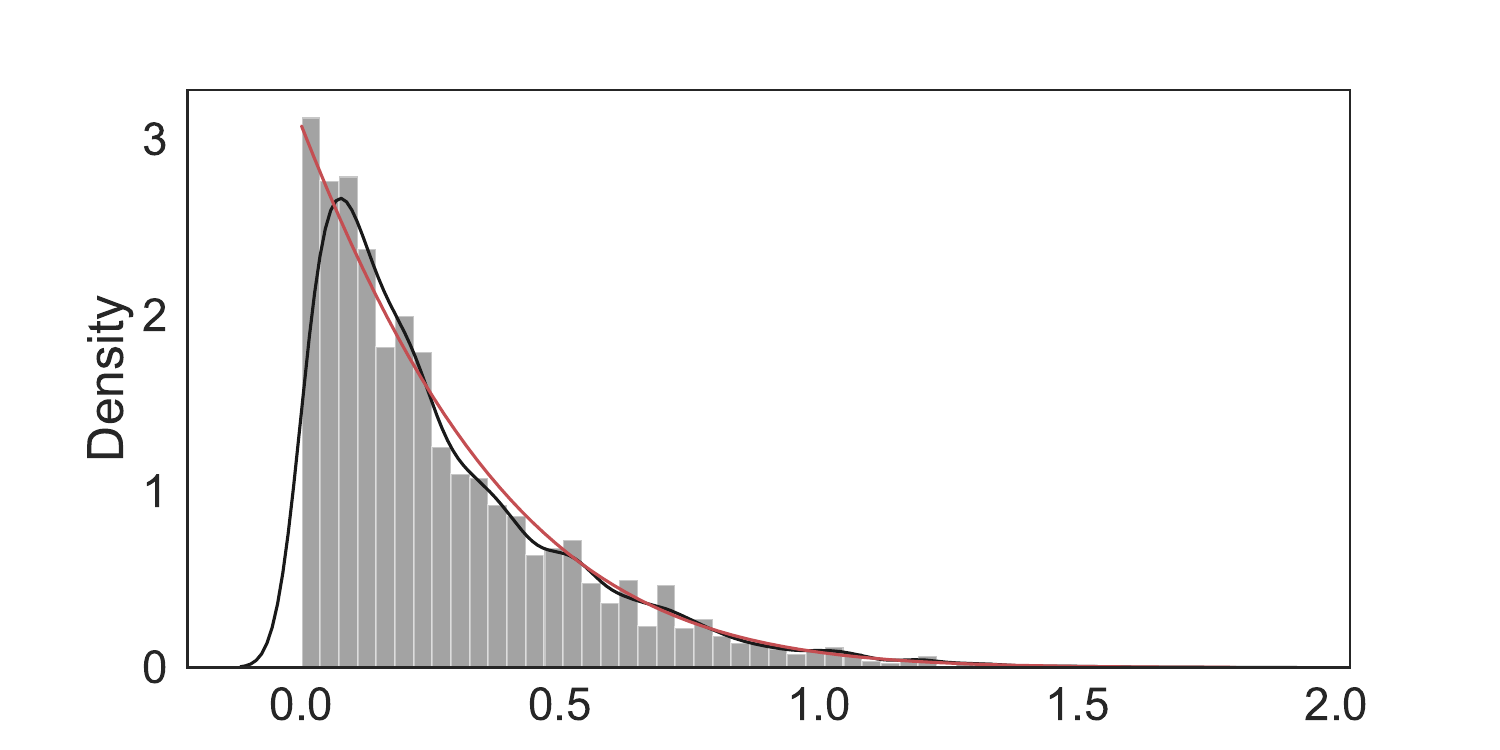}

    \\~\\
    $\bm z_{-2}=(-0.42,-0.35)$ & $\bm z_{-2}=(-0.86,-1.00)$ & $\bm z_{-2}=(-0.67,-0.49)$ \\~\\
  &  c) Case 3: $z_\star \leq0$ &
    \end{tabular}
    \caption{Histograms of simulated samples $Z_2\mid \bm Z_{-2}= \bm z_{-2}$ of size $m=10,000$. Kernel density estimates (curves in black) and  and theoretical densities of the original sample (curves in red). Figures a) correspond to Case 1 with $z_\star >0$ and $z_\star=z_1$, Figures b) to Case 2 with $z_\star >0$ and $z_\star \neq z_1$  and Figures c) to Case 3 with $z_{\star}<0$.}
    \label{fig:CondSimParametric} 
\end{figure}

\section{Illustrations on simulated data}\label{sec:DataSynthe}

As motivated in the introduction, the main goal of the MGP-based simulation approaches is to address the problem of scarcity of available observations in extreme regions. These regions will take different forms depending on the application of interest. 

In the following, we illustrate the two simulation algorithms derived in Section~\ref{sec:algos} on a numerical setting inspired by the field of finance and more specifically financial returns (details are given below). This simulation framework is presented in Section~\ref{subsec:NS}.  The joint MGP simulation algorithm (Algorithm~\ref{alg:JointMGP}) is then applied in Section~\ref{subsec:TRMynthe} in order to estimate the three TRMs presented in Section~\ref{sec:TRM} from the simulated data. Finally, in Section~\ref{subsec:ShocksEst}, the conditional MGP simulation algorithm (Algorithm~\ref{algo:CondAlgo}) is applied to estimate financial shocks, defined as a specific conditional expectation.

\subsection{Simulation framework}\label{subsec:NS}

We consider a 3-dimensional random vector $\bm X =(X_1,X_2,X_3)$ with  marginal Student's $t$-distribution with a relatively low degree of freedom $\nu$. This choice of distribution for the marginals is chosen  so as to mimic financial returns, which are typically heavy tailed. Indeed, the tail of the Student's $t$-distribution is governed by $1/\nu$. So, the smaller $\nu$ the heavier the tail. 

To fully characterise the joint distribution of $\bm X$, we need to define the dependence structure in addition to the marginal distribution. Sklar's theorem \citep{sklar1959fonctions} states that given the continuous joint distribution $\bm F$ of a random vector $\bm X$ with marginal cumulative distribution functions $F_1,F_2,F_3$, there exists a unique copula $\mathcal{C}$ such that 
\begin{align*}
    \bm F(x_1,x_2,x_3) = \mathcal{C}\left(F_1(x_1),F_2(x_2),F_3(x_3)\right). 
\end{align*}

Copulas are powerful tools for modelling the dependence structure of a random vector separately from the marginal distributions. 

Let us also recall that an underlying assumption of our simulation framework is that the components of $\bm X$ are asymptotically dependent. To ensure that this hypothesis is satisfied, we consider the Gumbel copula \citep{nelsen2006introduction} to obtain dependent extremes in the upper tail. In the $3$-dimensional case, the Gumbel copula is defined as follows

\begin{align}\label{eq:copGumb}
    \mathcal{C}(\bm y) := \exp\left(- \left( \sum_{i=1}^3\left[-\log(y_i)\right]^{\theta
}\right)^{1/\theta}\right),
\end{align}
where $\theta \geq 1$ is the copula parameter. The larger $\theta$, the stronger the asymptotic dependence structure between the components of $\bm X$. Consequently, the pair $\left(\bm\nu,\theta\right)$, where $\bm\nu=(\nu_1,\nu_2,\nu_3)$ represents the vector of degrees of freedom associated with each component of $\bm X$ and $\theta$ parameterise the structure of dependence through Equation~\eqref{eq:copGumb}, fully characterises the joint distribution of $\bm X$.

The numerical experiments, unless otherwise specified, are performed on simulated data sets $\mathcal D\in \mathbb{R}^{1500\times 3}$, with $\nu_1 = 2$,  $\nu_2=3$, $\nu_3=2.5$. From this parametric framework, we can derive the theoretical values of the TRMs and infer quantities describing the tail of some conditional distribution, that will be used as benchmarks in the following sections.

As previously mentioned, Algorithms~\ref{alg:JointMGP} and~\ref{algo:CondAlgo} have to run on standard MGP vectors. To this end, the vectors $\bm X$ are transformed to standard MGP vectors $\bm Z$, following the procedure presented in Section~\ref{sec:MGP}. Conversely, once the algorithms have been run on $\bm Z$, the resulting simulated sample outputs of the algorithms $\widetilde{\bm Z}$ are projected back into the original space through 
\[
\widetilde{X}_j = F_j^{-1}\left(1-e^{\widetilde{Z}_j+u^E_j}\right), \mbox{ for } j\in\{1,\ldots,d\}
\] 
where $F^{-1}_j$ is the  quantile function associated with the $j$-th risk factor for $j=1,\ldots,d$.  $F^{-1}_j$ can be the quantile function of  any  parametric distribution that has been fitted to the data,  in our simulation framework, that is the Student's $t$-distribution quantile distribution.  Thus,  any quantity to be estimated on the $\bm X$-scale can be estimated on the $\widetilde{\bm X}$-samples. 

For comparison purpose, the estimation procedure is performed on the following three different samples.
\begin{itemize}
    \item {\bf Original sample} $\mathcal{D}$ generated using the parametric framework above, denoted {\it Orig};
    \item {\bf Simulated sample} $\mathcal{\widetilde{D}}$ obtained using the joint simulation algorithm (Algorithm~\ref{alg:JointMGP}), denoted {\it Simu};
    \item {\bf Extended sample} $\mathcal{D} \cup \mathcal{\widetilde{D}}$ corresponding to the concatenation of the original and simulated samples, denoted {\it Ext}.
\end{itemize}

\subsection{Tail-related risk metrics estimation through the joint simulation of multivariate extremes}\label{subsec:TRMynthe}

In this section, we investigate the performance of our joint simulation procedure (Algorithm~\ref{alg:JointMGP}) in estimating TRMs, as defined in Equations~\eqref{eq:ES}, \eqref{eq:MES} and \eqref{eq:DCTE}, with respect to the several parameters. Particular attention is paid to the strength of the asymptotic dependence, which can be controlled through the copula parameter $\theta$ and the level $\alpha$ (see Supplementary material, Section 1).

In the following, TRMs estimates are compared to their respective theoretical benchmark values (see Supplementary Material, Section 2). For the sake of brevity, we only considered the estimation procedure for the component $X_1$. However, it is important to note that it could also be performed for the other components, $X_2$ and $X_3$.

In order to run the joint simulation algorithm (Algorithm~\ref{alg:JointMGP}), it is necessary to specify the size of the simulated samples $m$. The simulation procedure was run with $m=10,000$ which yielded stable results while exhibiting a reasonable computational cost. Further results regarding the performance of the simulation algorithm with respect to the size $m$ can be found in the Supplementary Material (Section 3). 

The estimation of the $\VaR$ is of paramount importance in estimating the TRM, as all TRMs are defined from the $\VaR$ (see Section~\ref{sec:TRM}).   However, this is not the focus of our current investigation. The performance of our approach is evaluated using the theoretical $\VaR$, that is to say, the $\VaR$ computed as a Student's $t$-quantile with the corresponding degrees of freedom. Other techniques for estimating the $\VaR$  have been investigated, including an empirical approach and the estimation of the $\VaR$ as a univariate extreme value distribution quantile. The results of these analyses can be found in the Supplementary Material (Section 4). In particular, it can be seen that when the theoretical $\VaR$ is unknown, which is often the case in real-data applications, the estimation of the $\VaR$ using a univariate extreme value approach appears to be a good candidate.

Given that the $\VaR$ estimation technique and $m$ are fixed, we can now proceed to investigate the performance of our approach with respect to $\alpha$ and $\theta$. Moreover, since the estimation errors are sensitive to the sample $\mathcal{D}$, it is necessary to  ascertain that the approach performance is not specific to the original sample $\mathcal{D}$. To this end, 50 distinct original samples $\mathcal{D}$ were simulated (of sample size $n$ and with the same joint distribution as outlined in Section~\ref{subsec:NS}). Then, for each original sample $\mathcal{D}$, a set of 50 simulated samples $\widetilde{\mathcal{D}}$ of size $m=10,000$ was generated following the joint simulation algorithm (Algorithm~\ref{alg:JointMGP}). The TRMs are estimated on each of 50 original, simulated and extended samples, along with the relative errors on these TRMs approximations.  

Figure~\ref{fig:AlphaThetaTRM} and Table~\ref{tab:ImpactThetaTRM} summarise the results for TRMs estimations at different levels $\alpha \in \{0.9975,0.999,0.9997\}$ and for different copula parameters $\theta \in \{1.3,2.6,7.3\}$. Boxplots of the distributions of the relative errors are presented, showing that the relative errors for the TRMs estimations on the simulated samples {\it Simu} and on the extended samples {\it Ext} are smaller and present less dispersion  than those made on the original samples {\it Orig}.  It is also important to ascertain that the quality of TRMs do not deteriorate when $\alpha$ gets closer to 1. Figure~\ref{fig:AlphaThetaTRM} shows TRMs estimations at level $\alpha \in \{0.9975,0.999,0.9997\}$, one may observe that the estimation on {\it Simu} and {\it Ext} samples induce fairly low approximation errors regardless of the level $\alpha$. We notice that the joint simulation of multivariate extremes using Algorithm~\ref{alg:JointMGP} improves not only the estimation of $\MES$ and $\DCTE$ (defined with the joint distribution of $\bm X$) but also the estimation of $\ES$ even though it only evolves the marginal distribution of $X_j$.
\begin{figure}[h]
\begin{tabular}{ccc}
    \includegraphics[width=0.3\textwidth]{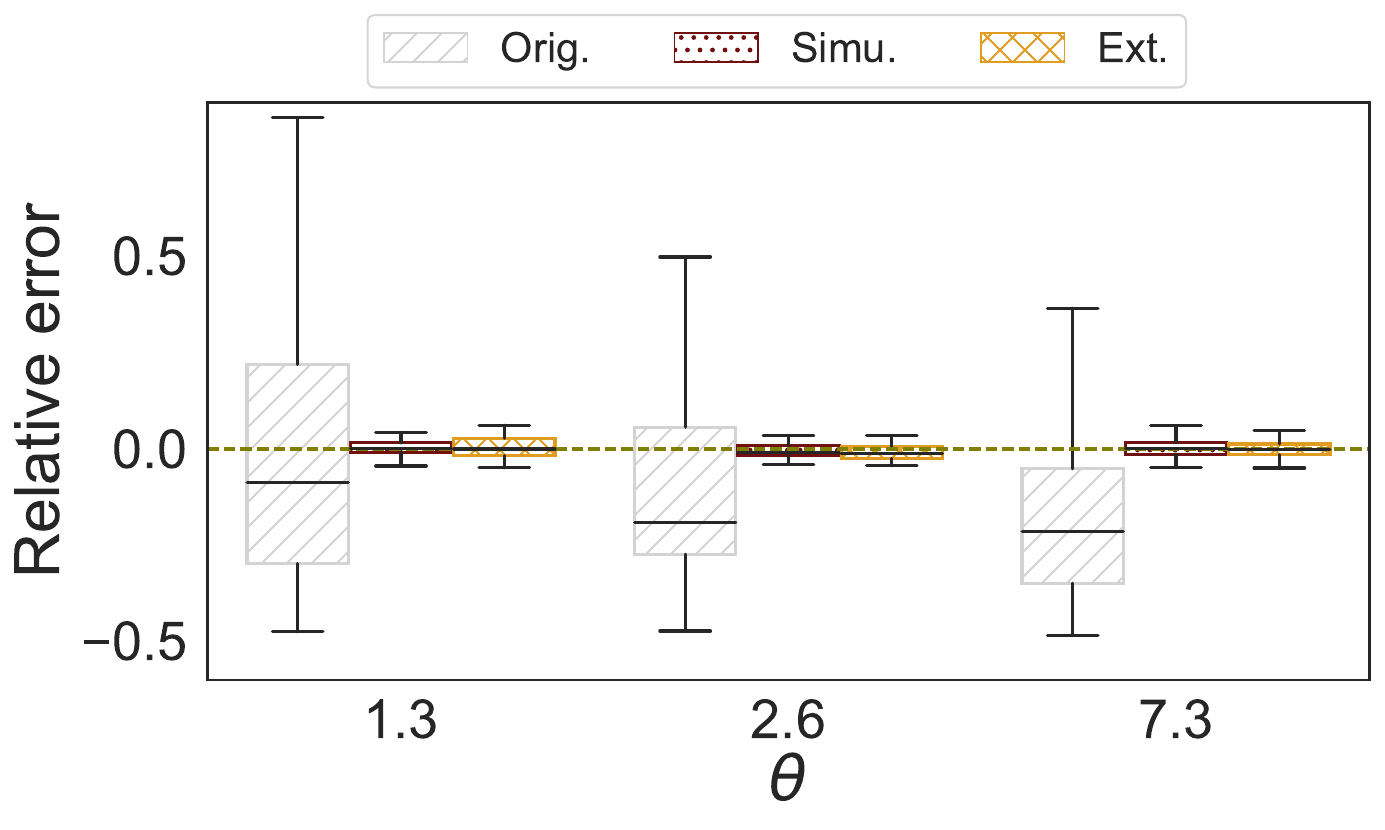} &  \includegraphics[width=0.3\textwidth]{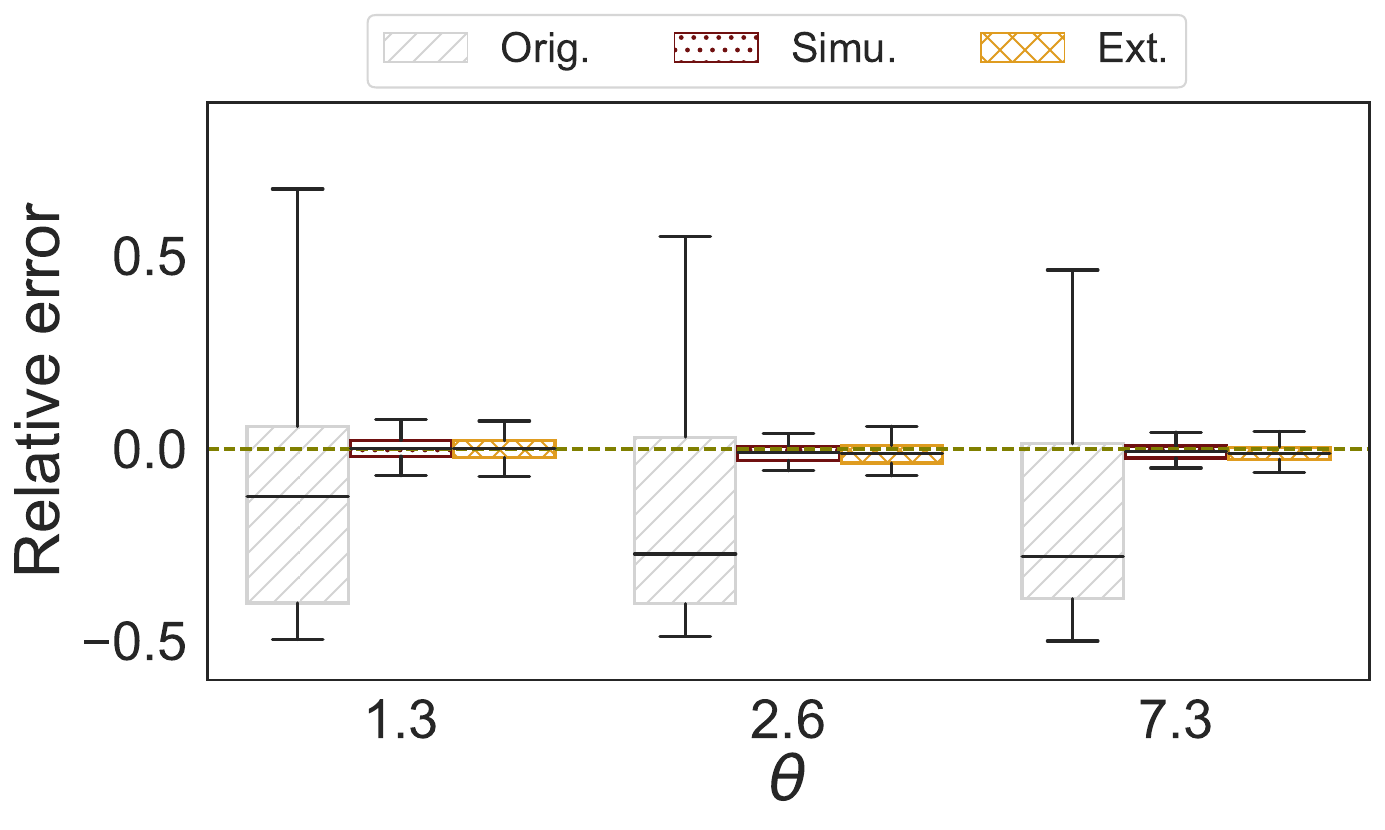} & \includegraphics[width=0.3\textwidth]{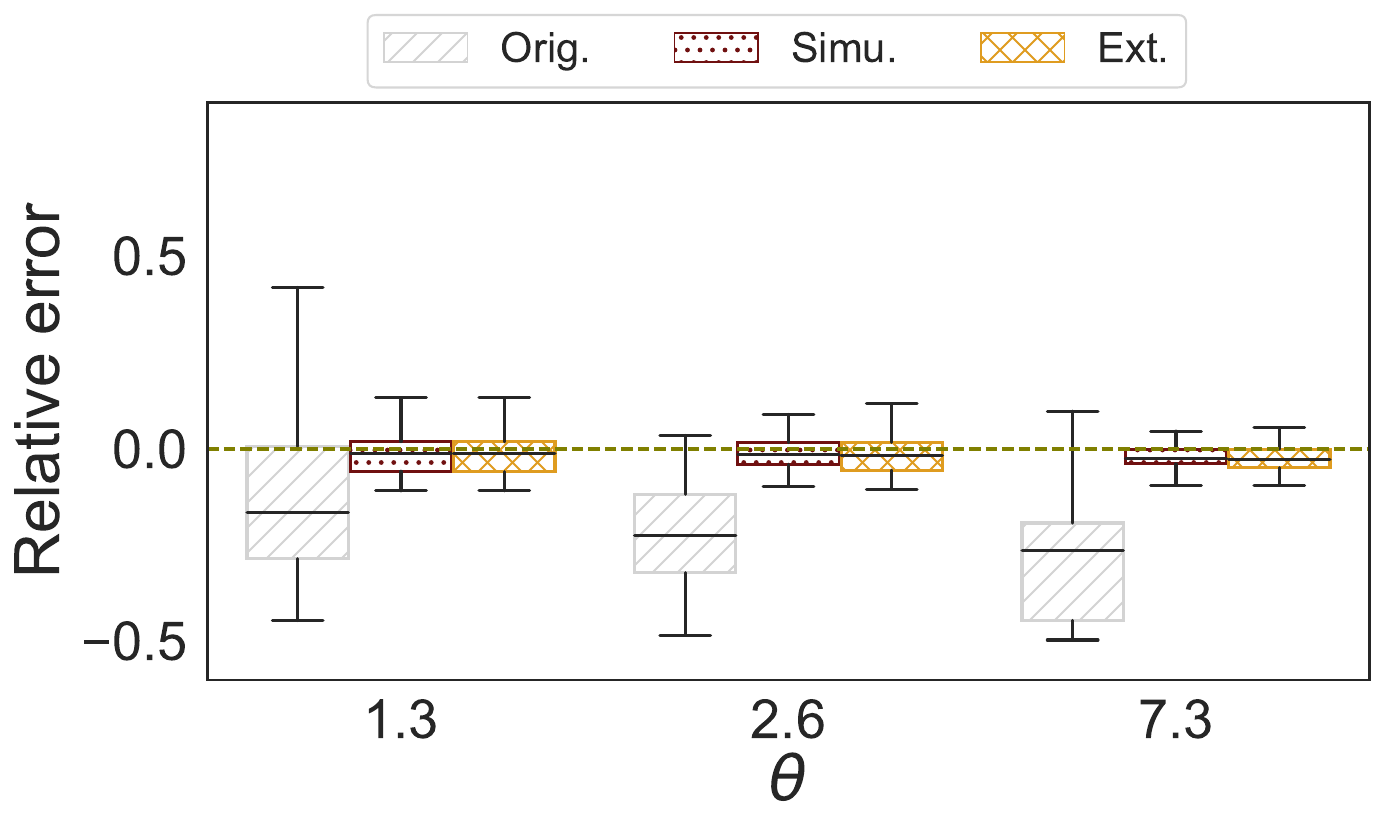} \\~\\
    a) $\ES_{0.9975}$ & b) $\ES_{0.999}$ & c) $\ES_{0.9997}$ \\~\\
 
    \includegraphics[width=0.3\textwidth]{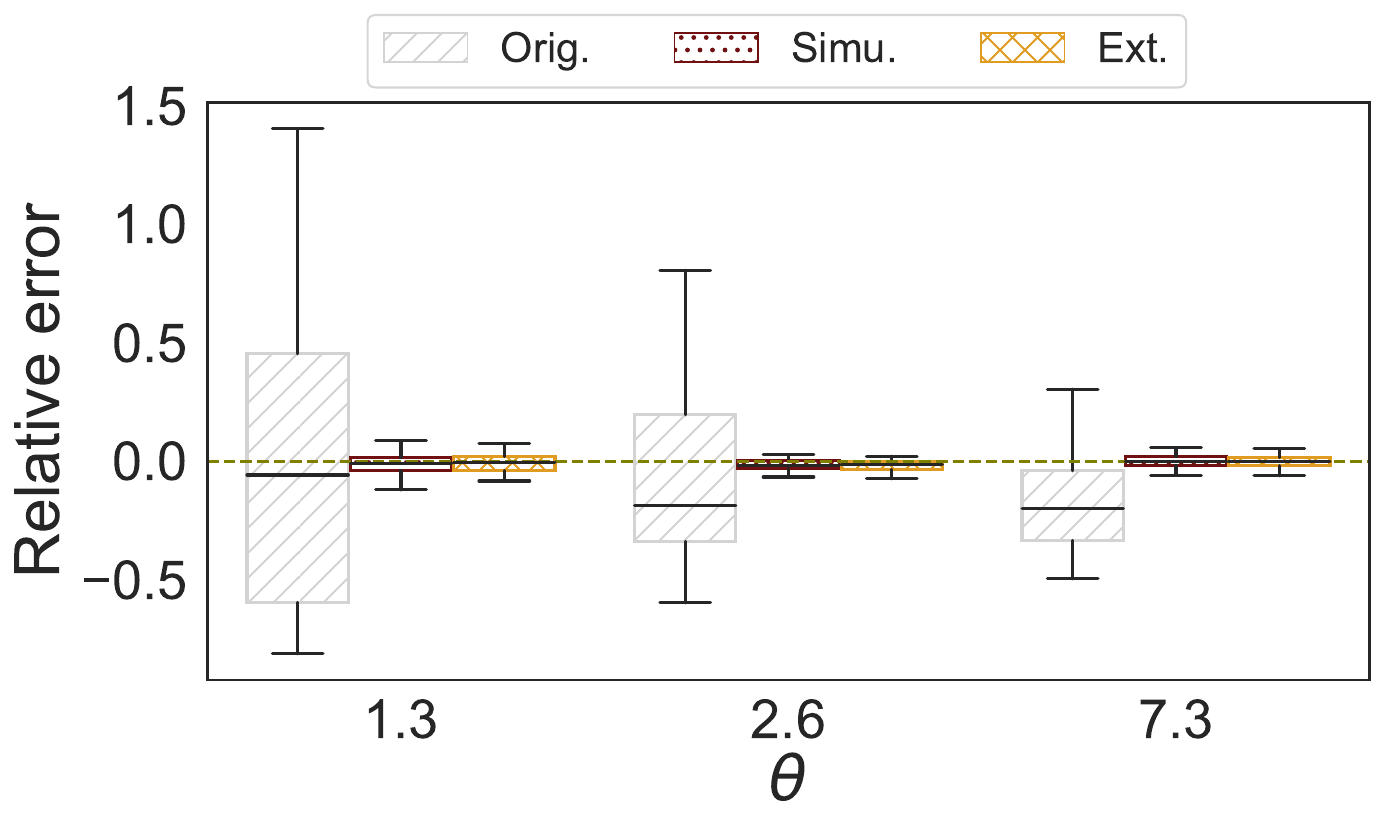} &   \includegraphics[width=0.3\textwidth]{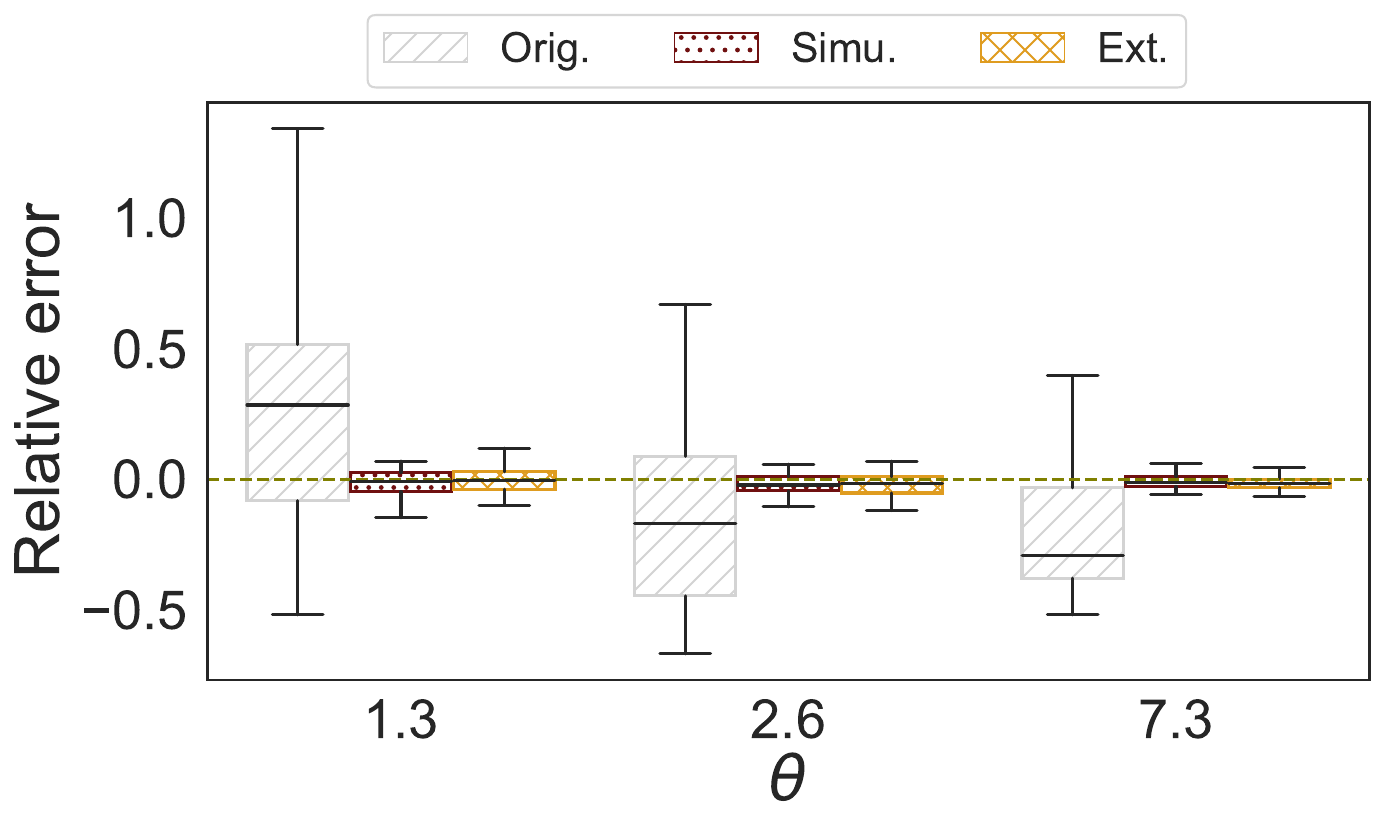} &    \includegraphics[width=0.3\textwidth]{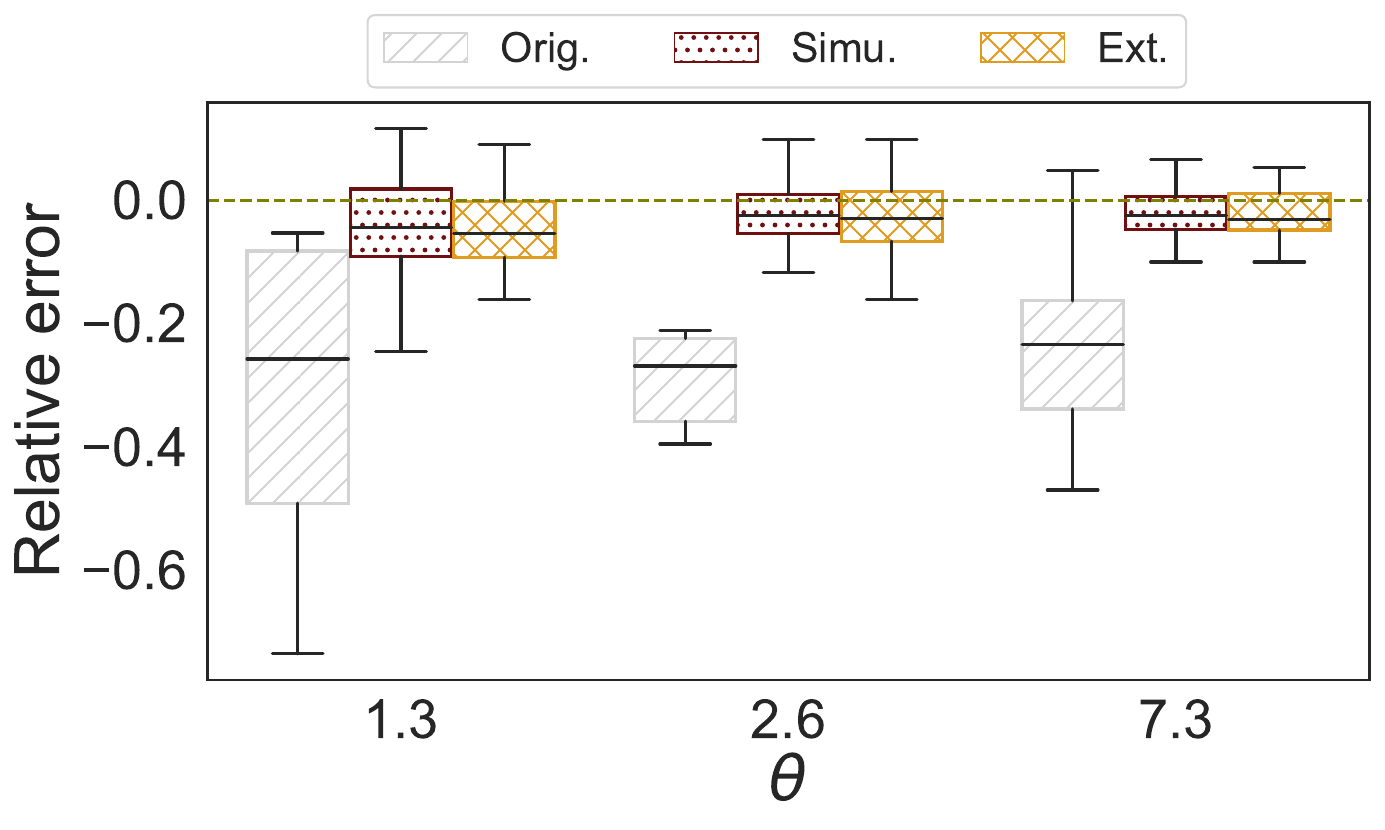}\\~\\
    d) $\MES_{0.9975}$ & e) $\MES_{0.999}$ & f) $\MES_{0.9997}$ \\~\\ 

    \includegraphics[width=0.3\textwidth]{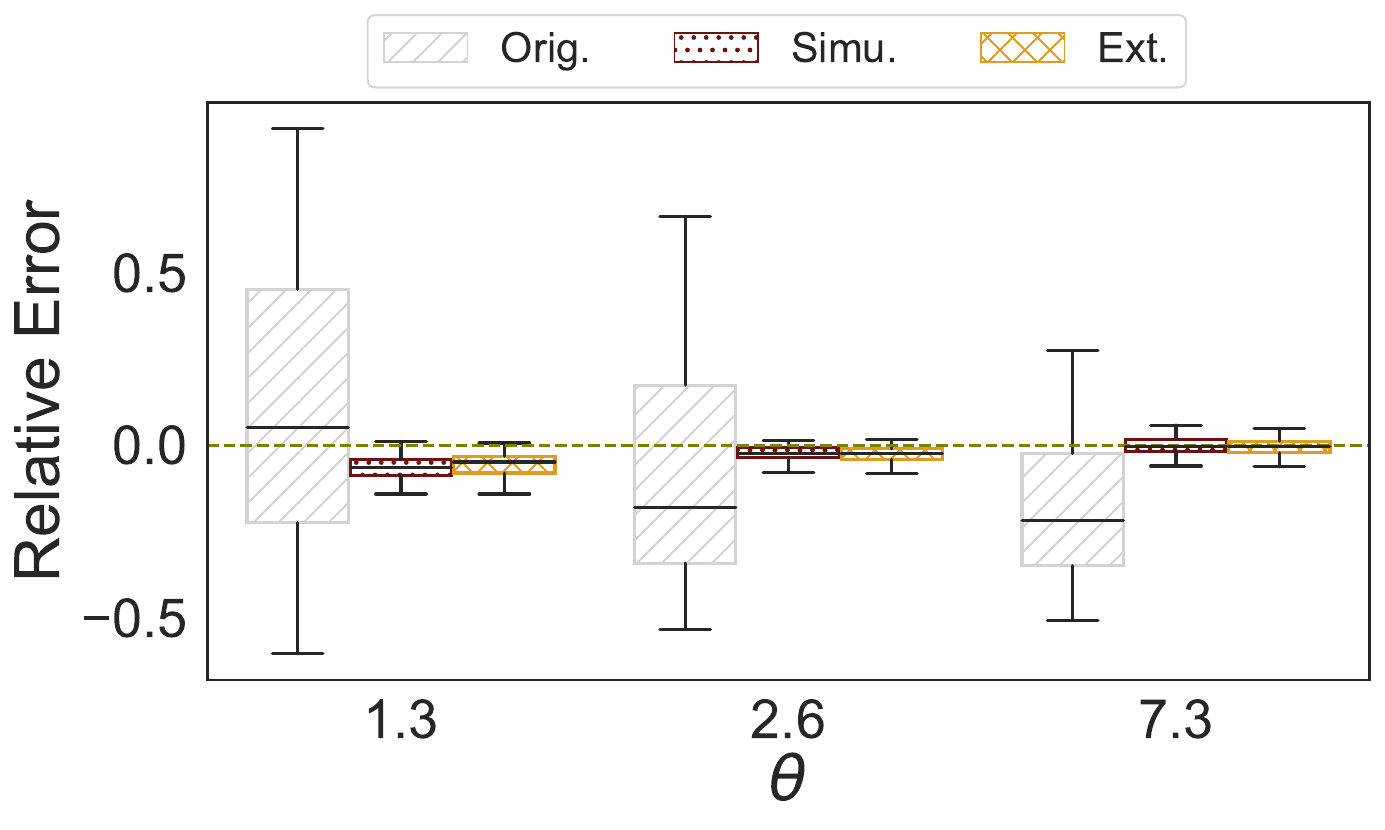} &  \includegraphics[width=0.3\textwidth]{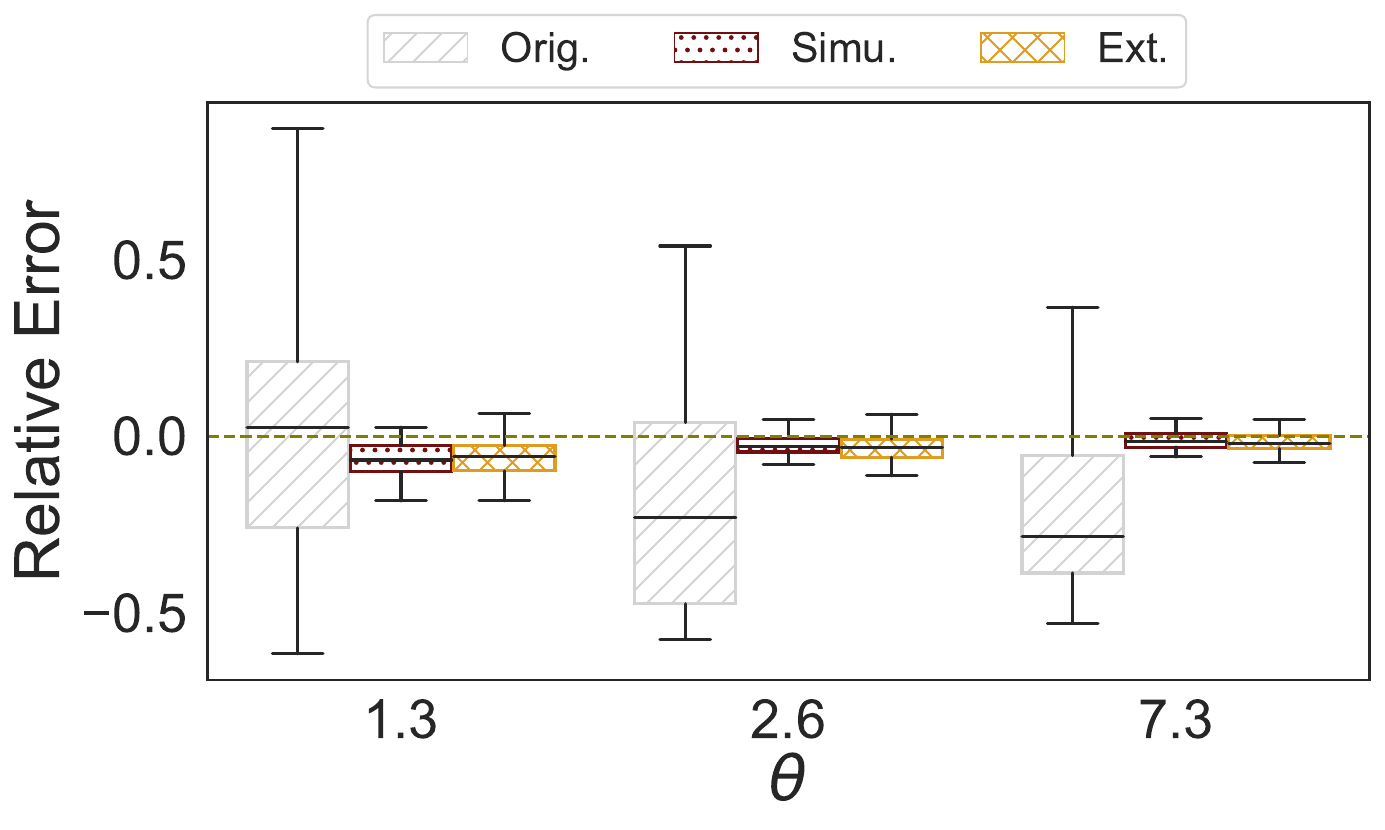} & 
     \includegraphics[width=0.3\textwidth]{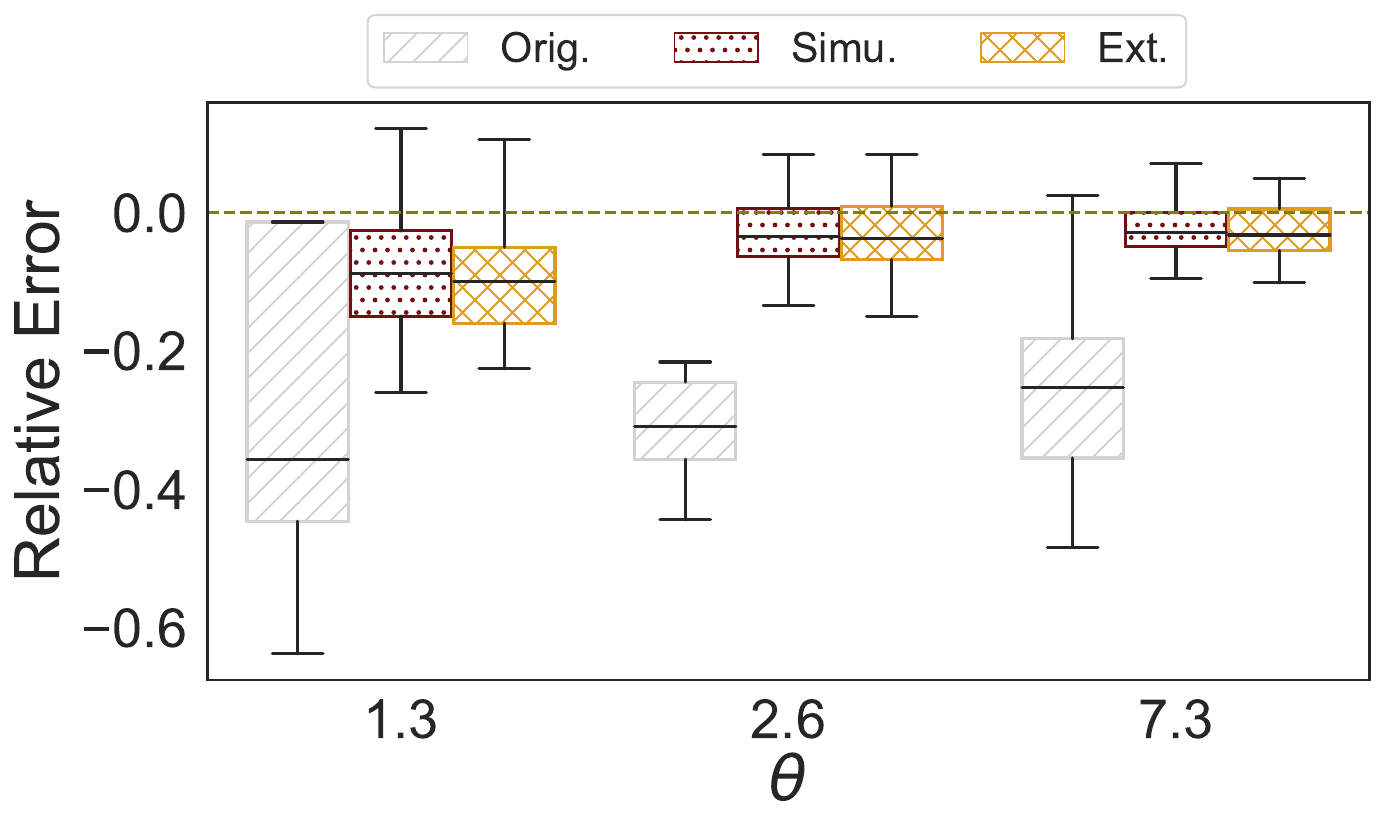}\\~\\
    g) $\DCTE_{0.9975}$ & h) $\DCTE_{0.999}$ & i) $\DCTE_{0.9997}$\\
    \end{tabular}
  \caption{Distribution of relative approximation errors of the estimations of TRMs on 50 original samples {\it (grey oblique lines)}, 50 simulated samples {\it (red dots)} and 50 extended samples {\it (yellow grid)} for the $\ES_\alpha$ (Figures a), b), c)), the $\MES_\alpha$ (Figures d), e), f)) and the $\DCTE_\alpha$ (Figures g), h), i)) at varying level $\alpha \in\{0.9975,0.999,0.9997\}$ with respect to copula parameter $\theta \in \{1.3, 2.6, 7.3\}$. On each graph, the dashed blue horizontal line corresponds to the associated theoretical value.}
  \label{fig:AlphaThetaTRM}
\end{figure}

\begin{table}[h]
\begin{center}
\begin{tabular}{c}
\begin{tabular}{c|rr|rr|rr}
\multirow{2}{*}{$\alpha$} & \multicolumn{2}{c|}{$\ES_\alpha$}  & \multicolumn{2}{c|}{$\MES_\alpha$}                                     
& \multicolumn{2}{c}{$\DCTE_\alpha$}     \\
\cline{2-7} & \multicolumn{1}{c|}{{\it Orig}}   & \multicolumn{1}{c|}{{\it Simu}} & \multicolumn{1}{c|}{{\it Orig}}        & \multicolumn{1}{c|}{{\it Simu}} & \multicolumn{1}{c|}{{\it Orig}}  & \multicolumn{1}{c}{{\it Simu}} \\
 \hline

0.9975  & \multicolumn{1}{r|}{3.7 (1.6)} & 87.4 (5.2)   & \multicolumn{1}{r|}{0.82 (0.9)} & 33.9 (2.2)  & \multicolumn{1}{r|}{1.3 (1.3)} & 42.9 (3.3) \\
0.999  & \multicolumn{1}{r|}{1.7 (1.4)} & 35.0 (2.3)    & \multicolumn{1}{r|}{0.4 (0.6)} & 13.5 (1.1)    & \multicolumn{1}{r|}{0.5 (0.7)}  & 17.2 (1.4) \\ 
0.9997   & \multicolumn{1}{r|}{0.6 (0.7)} & 10.4 (0.8)     & \multicolumn{1}{r|}{0.2 (0.49)} & 4.0 (0.4)   & \multicolumn{1}{r|}{0.2 (0.5)} &5.1 (0.5)               
\end{tabular}\\~\\

a) $\theta = 1.3$ \\~\\

\begin{tabular}{c|rr|rr|rr}
\multirow{2}{*}{$\alpha$} & \multicolumn{2}{c|}{$\ES_\alpha$}   & \multicolumn{2}{c|}{$\MES_\alpha$}    & \multicolumn{2}{c}{$\DCTE_\alpha$}  \\
\cline{2-7} & \multicolumn{1}{c|}{{\it Orig}} & \multicolumn{1}{c|}{{\it Simu}} & \multicolumn{1}{c|}{{\it Orig}}   & \multicolumn{1}{c|}{{\it Simu}} & \multicolumn{1}{c|}{{\it Orig}}    & \multicolumn{1}{c}{{\it Simu}} \\ 
 \hline

0.9975     & \multicolumn{1}{r|}{3.9 (2.4)} & 114.3 (6.6)   & \multicolumn{1}{r|}{2.3 (1.6)}  & 74.1 (3.9)                & \multicolumn{1}{r|}{2.7 (1.7)} & 83.2 (4.6) \\ 
\hline
0.999   & \multicolumn{1}{r|}{1.6 (1.2)} & 46.0 (2.8)    & \multicolumn{1}{r|}{0.9 (0.9)} & 29.8 (1.7)    & \multicolumn{1}{r|}{1.0 (0.9)} & 33.5 (2.1)  \\ 
\hline
0.9997   & \multicolumn{1}{r|}{0.4 (0.6)} & 13.6 (0.7)  & \multicolumn{1}{r|}{0.2 (0.5)} & 8.8 (0.53)  & \multicolumn{1}{r|}{0.3 (0.5)} & 10.0 (0.6)  
\end{tabular}\\~\\
b) $\theta=2.6$ \\~\\
\begin{tabular}{c|rr|rr|rr}
\multirow{2}{*}{$\alpha$} & \multicolumn{2}{c|}{$\ES_\alpha$}   & \multicolumn{2}{c|}{$\MES_\alpha$}   & \multicolumn{2}{c}{$\DCTE_\alpha$}   \\ 
\cline{2-7} & \multicolumn{1}{c|}{{\it Orig}}      & \multicolumn{1}{c|}{{\it Simu}} & \multicolumn{1}{c|}{{\it Orig}}    & \multicolumn{1}{c|}{{\it Simu}} & \multicolumn{1}{c|}{{\it Orig}}  & \multicolumn{1}{c}{{\it Simu}} \\ 
\hline
0.9975  & \multicolumn{1}{r|}{4.1 (2.1)} & 144.8 (10.0)  & \multicolumn{1}{r|}{3.4 (1.9)} & 126.0 (8.5)   & \multicolumn{1}{r|}{3.5 (2.0)}  & 131.0 (8.9)\\
\hline
0.999   & \multicolumn{1}{r|}{1.6 (1.4)} & 57.7 (4.1)  & \multicolumn{1}{r|}{1.5 (1.3)} & 50.2 (3.38)   & \multicolumn{1}{r|}{1.5 (1.4)} & 52.3 (3.6)  \\ 
\hline
0.9997    & \multicolumn{1}{r|}{0.4 (0.6)} & 17.3 (10.0)   & \multicolumn{1}{r|}{0.3 (0.5)} & 15.1 (1.1)    & \multicolumn{1}{r|}{0.3 (0.5)}  & 15.7 (1.2) \\
\end{tabular}\\~\\
c) $\theta=7.3$
\end{tabular}
\caption{Mean values and standard deviations in brackets of the number of observations on which the TRMs are estimated on the original samples ({\it Orig}) and on the simulated samples ({\it Simu}) for different values of $\alpha \in\{0.9975,0.999,0.9997\}$ and different copula parameters a) $\theta=1.3$, b) $\theta=2.6$  and c) $\theta=1.3$.}
 \label{tab:ImpactThetaTRM}
 \end{center}
\end{table}

The accuracy of the TRMs estimates is attributable to the enlargement of the sample size on which TRMs are being estimated compared to the original sample size. Table~\ref{tab:ImpactThetaTRM} shows the mean value of the sample size on which the three TRMs are estimated when considering the original samples and simulated samples for different levels $\alpha \in \{0.9975,0.999,0.9997\}$ and for different values of copula parameters $\theta \in \{1.3, 2.6, 7.3\}$. Overall, it can be noted that for the considered levels of $\alpha$, the TRMs estimation on the original samples is conducted with limited data, that is, here with at most 5 observations. In contrast, our simulation approach allows for a larger sample size to estimate the various TRMs, extending the original observations to 100 observations in some cases. As expected, the number of points above the $\VaR$ increases when the asymptotic dependence is stronger.

\subsection{Conditional expectation estimation through the conditional simulation of multivariate extremes}
\label{subsec:ShocksEst}

In Section~\ref{subsec:TRMynthe}, we proposed a numerical illustration of our first Algorithm~\ref{alg:JointMGP}, the joint simulation algorithm, through the empirical estimation of TRMs. As previously stated in the introduction, the second algorithm~\ref{algo:CondAlgo} can be used to infer quantities describing the tail of some conditional distribution. Since the TRMs under investigation can be expressed as conditional expectations, one might be inclined to consider the conditional simulation algorithm for their estimation. However, accurate estimation of TRMs  with the conditional simulation algorithm (Algorithm~\ref{algo:CondAlgo}) is not possible. Indeed, for  a MGP vector $\bm Z$, the aforementioned algorithm generates samples of $Z_j \mid \bm Z_{-j}=\bm z_{-j}$ for $j=1,\ldots, d$. However, when considering the estimation of $\MES$ for instance, we are attempting  to estimate the expectation of $Z_{j}\mid\bm Z_{-j}\geq\bm z_{-j}$ (see Equation~\eqref{eq:MES}).

Nevertheless, the conditional simulation algorithm could be employed in other contexts.  For a vector of risk factors $\bm X$, we consider here the inference of the expectation of  risk factor $X_j$ given that the other risk factors $\bm X_{-j}$ are observed, that is $\bm X_{-j}=\bm x_{-j}$, defined as follows 
\begin{equation}\label{eq:shocks}
\mu(\bm x_{-j})=\mathbb{E}\left[X_j \mid \bm X_{-j}=\bm x_{-j}\right] . 
\end{equation} 
Algorithm~\ref{algo:CondAlgo},  the conditional simulation algorithm,  is illustrated on this specific application using the simulated data generated with the parametric framework specified in Section~\ref{subsec:NS}.
In this simulation framework, the conditional distribution of $X_j$ given $\bm X_{-j}=\bm x_{-j}$ can be derived explicitly, thereby enabling the derivation of $\mu(\bm x_{-j})$ as defined in Equation~\eqref{eq:shocks}. Hence, a benchmark value is available, and can be used to measure the accuracy of the empirical estimations using simulated samples generated with Algorithm~\ref{algo:CondAlgo}.

For each original sample $\mathcal{D}$, we consider the estimation of the conditional expectation  denoted by $\mu(\bm x_{-2})$ (see Equation~\eqref{eq:shocks}). This estimation is based on the conditioning observations $\bm x_{-2}$,  the original sample $\mathcal{D}$ and the generated conditional simulated samples of $X_2 \mid \bm X_{-2}= \bm x_{-2}$.  Once the vector $\bm X$ has been transformed into a MGP vector $\bm Z$, the conditioning observations are selected such that each conditional simulation case (see Section~\ref{sec:CondSim}),  $z_{\star}>0$,  $z_\star=z_1$ (Case 1), $z_{\star}>0$, $z_{\star}\neq z_1$ (Case 2) and $z_{\star}\leq 0$ (Case 3), is represented. For each $\bm z_{-2}$, 10 simulated samples $Z_2 \mid \bm Z_{-2}= \bm z_{-2}$ of size $m=10,000$ are generated, and $\mu(\bm x_{-2})$ is estimated as the empirical mean of each simulated sample. The deviation of the  mean over these estimates from the reference value are computed for each conditioning observation $\bm x_{-2}$. Then, Figure~\ref{fig:pointplot} depicts the distribution of the mean absolute error across 50 original samples $\mathcal{D}$ for each conditional simulation case, with the conditioning observations being increasingly extreme. In other words, the selected conditioning observations are in the bulk of the distribution for Figure~\ref{fig:pointplot} a), in the tail of the distribution for Figure~\ref{fig:pointplot} b), and even further in the tail in Figure~\ref{fig:pointplot} c).  For comparison purposes, a linear regression model with $X_2$ as the response variable and $\bm X_{-2}$ as the covariates is fitted. It is evident that the conditional simulation errors are of smaller magnitude than those induced by linear regression, which result in substantial errors. It should be noted that while errors exhibit large amplitude as the conditioning observation becomes extreme with linear regression, the situation is reversed when the conditional simulation is considered. The lowest errors are observed for the most extreme conditioning observations, illustrated in  Figure~\ref{fig:pointplot} c). 

\begin{figure}[h]
\centering
\begin{tabular}{cc}
  \centering
  \includegraphics[width=.45\textwidth]{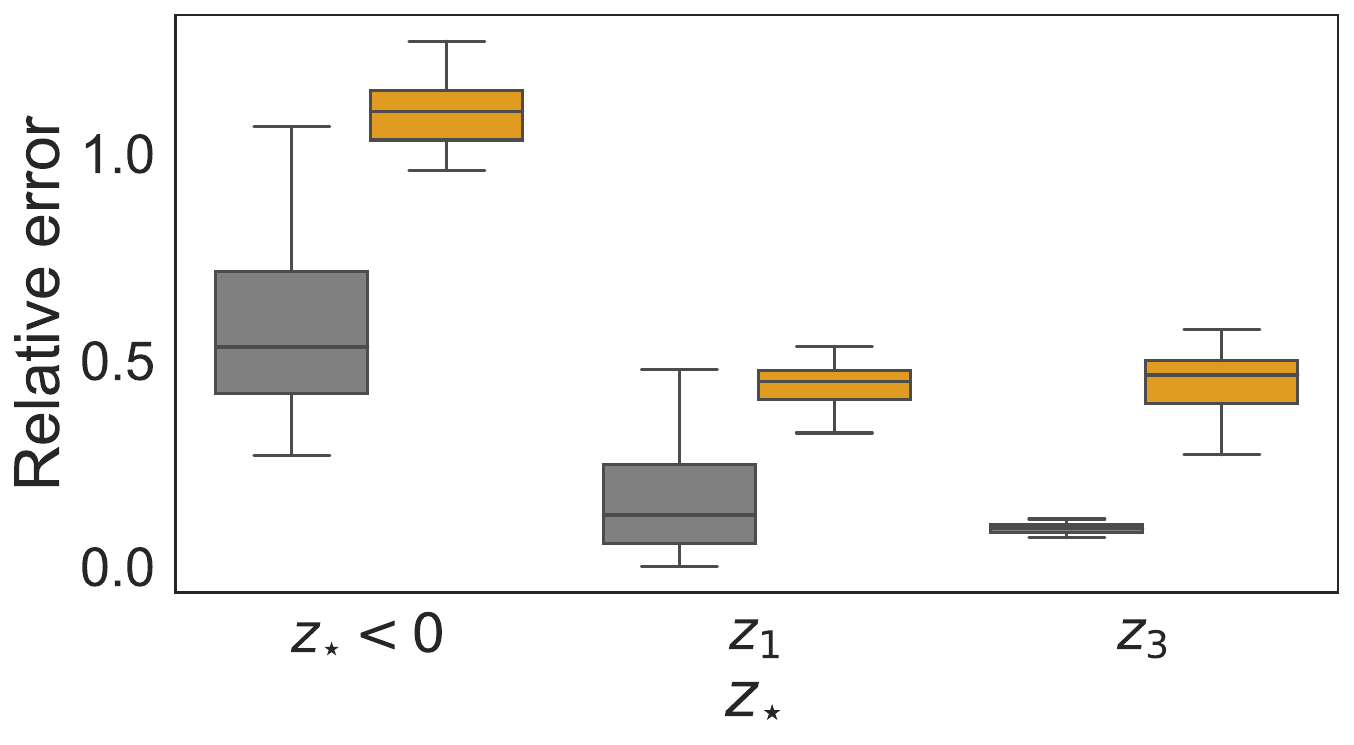} & \includegraphics[width=.45\textwidth]{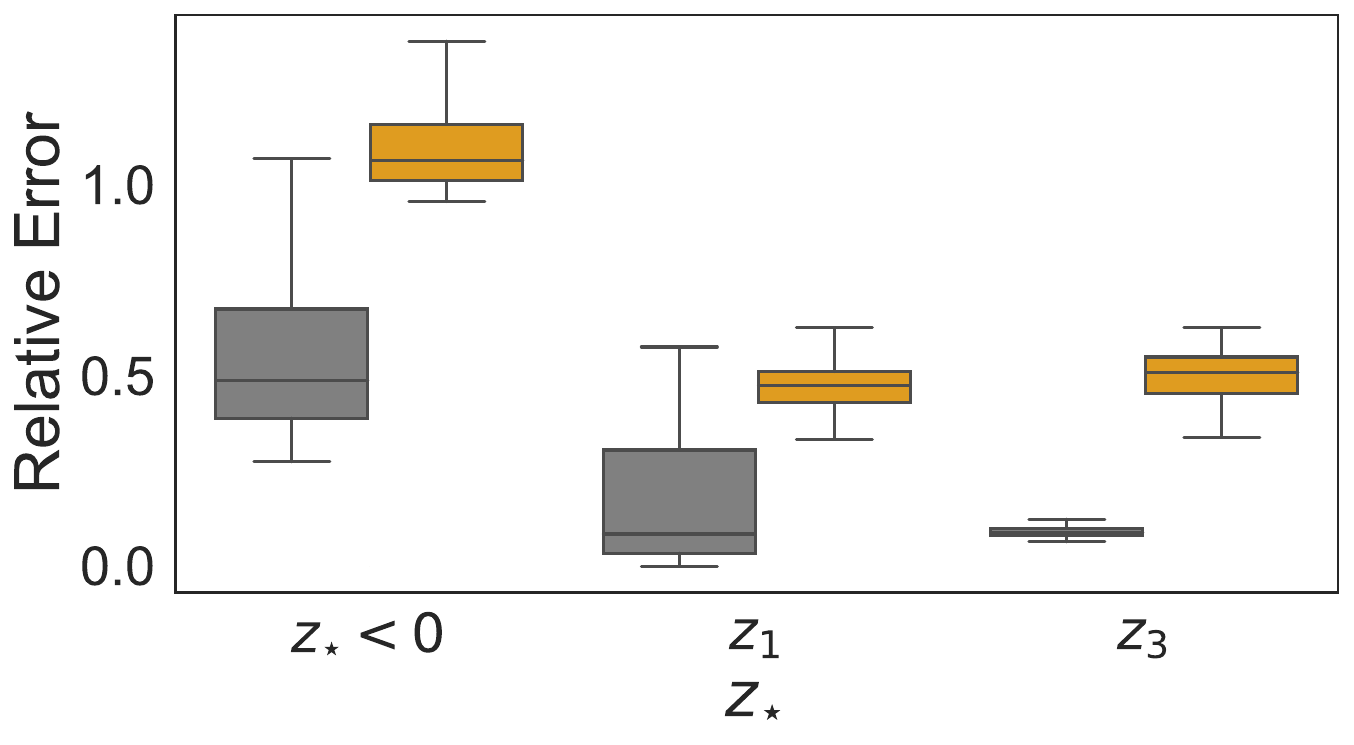} \\
  a) $\alpha=0.975$ & b) $\alpha=0.99$  \\~\\
 \multicolumn{2}{c}{\includegraphics[width=.45\textwidth]{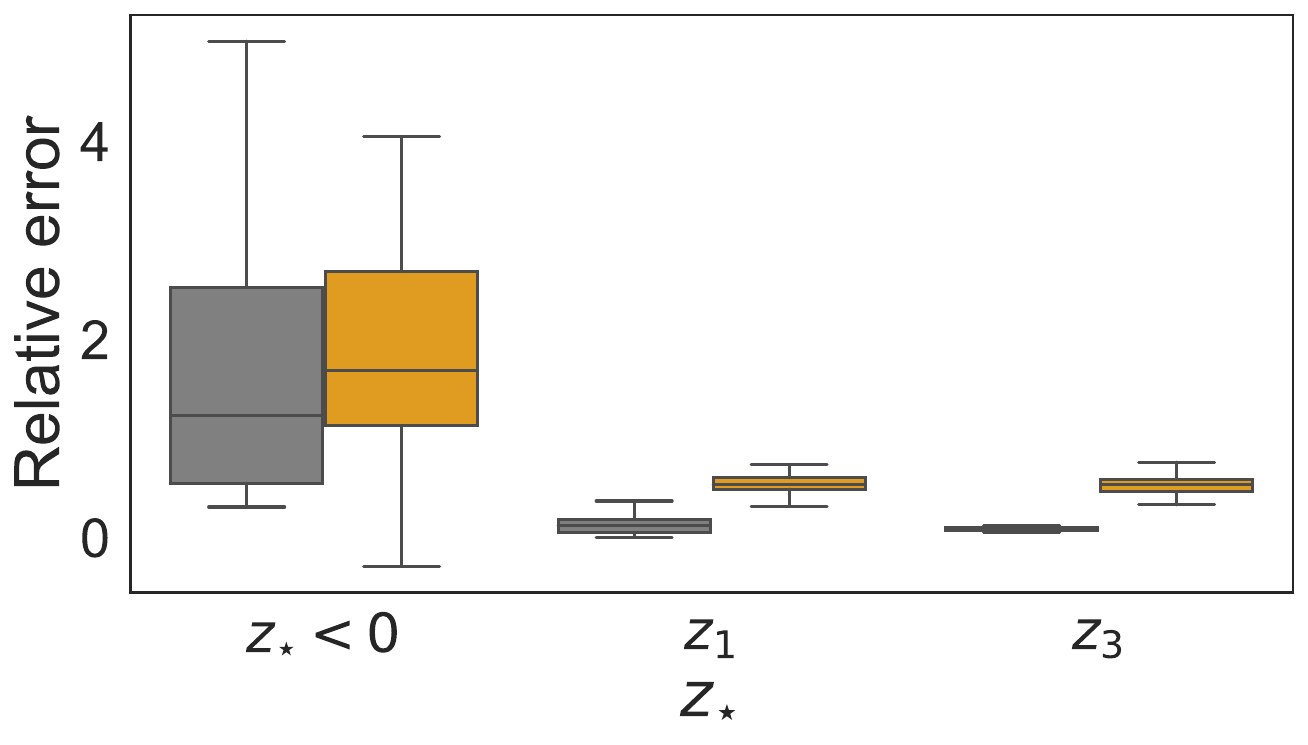}} \\
 \multicolumn{2}{c}{c) $\alpha\rightarrow1$ (i.e. the maximum)}\\
\end{tabular}
    \caption{Distribution of relative errors of the deviations from the reference value of $\mathbb{E}\left[X_2 \mid\bm X_{-2} = \bm x_{-2}\right]$ with respect to each conditional simulation case, using conditional simulation {\it (dark grey boxplot)} and linear regression {\it (light orange boxplot)} for increasingly extreme conditioning observations, that is $\bm x_{-2}$ chosen as three different $\alpha$-quantile of all-observations.}
  
\label{fig:pointplot}
\end{figure}

\section{Illustrations on real financial data}\label{sec:DataReal}

In the previous section, the two non-parametric simulation algorithms~\ref{alg:JointMGP} and~\ref{algo:CondAlgo} were illustrated on simulated data. These illustrations have demonstrated the advantages of these procedures for the estimation of TRMs and the inference of quantities involving the tail of some conditional distribution. This section presents a more practical illustration based on real financial data.

These data consist of weekly negative returns on  stocks from three UK banks: HSBC, Lloyds and RBS. The corresponding time series are denoted $X_1$, $X_2$ and $X_3$ respectively. The sample size of each time series is $n=470$, with temporal resolution spanning from 29/10/2007 to 17/10/2016. The data used in this study were initially extracted from Yahoo Finance, and considered in \citep{kiriliouk2019peaks}.

As in Section~\ref{sec:DataSynthe}, both algorithms are applied to this real data set $\bm X =(X_1,X_2,X_3)$ in two stages.  Firstly, Algorithm~\ref{alg:JointMGP} is employed in order to estimate TRMs (Section~\ref{subsec:TRMreal}). Secondly, Algorithm~\ref{algo:CondAlgo} is used to estimate the conditional expectation (Equation~\eqref{eq:shocks}, Section~\ref{subsec:ShockReal}).

Prior to this, and as done with the simulated data, the first step is to standardise the data to common margins (see Section~\ref{sec:MGP}).  In our case, we have chosen to fit a Student's $t$-distribution on each risk factor, which is often used with heavy-tailed financial data  \citep[see e.g.][]{mcneil2015quantitative}. However, other choices could have been considered. Goodness-of-fit plots can be found in the Supplementary Material (see Section 5) along with the estimated parameters of the fitted Student's $t$-distributions.

\subsection{Tail related risk metrics estimations}\label{subsec:TRMreal}

A set of $100$ new samples of weekly negative returns of size $m=10,000$ were generated using our joint simulation algorithm (Algorithm~\ref{alg:JointMGP}) in order to perform empirical estimation of the three TRMs of interest. The results are displayed in Table~\ref{tab:TRMRealData} along with the estimations on the original sample {\it Orig}, i.e. without increasing the amount of available data. The performance assessment conducted under the parametric framework has shown that TRMs estimates on the simulated samples {\it Simu}  were more precise than estimations on the extended sample {\it Ext}, hence, we exclusively discuss results from the simulated samples {\it Simu} in this section. 
 Table~\ref{tab:TRMRealData} shows that the joint simulation approach allows for the estimation of the TRMs that would be unfeasible if only the original sample was considered. This is to be expected, given that the objective of the simulation approach is to increase the number of available observations in extreme regions in order to enable the performance of empirical estimations in a reliable manner.

\begin{table}
    \centering
   \begin{tabular}{c|rr|rr|rr}
\multirow{2}{*}{Stock} & \multicolumn{2}{c|}{$\ES_{0.9975}$}  & \multicolumn{2}{c|}{$\MES_{0.9975}$}          & \multicolumn{2}{c}{$\DCTE_{0.9975}$}   \\
\cline{2-7} & \multicolumn{1}{c|}{{\it Orig}}  & \multicolumn{1}{c|}{{\it Simu}} & \multicolumn{1}{c|}{{\it Orig}}        & \multicolumn{1}{c|}{{\it Simu}} & \multicolumn{1}{c|}{{\it Orig}}    & \multicolumn{1}{c}{{\it Simu}} \\
\hline

HSBC & \multicolumn{1}{r|}{0.24 (-) } & 0.28 (0.015)    & \multicolumn{1}{c|}{NA} &0.29 (0.022)   & \multicolumn{1}{r|}{0.11 (-)} & 0.32 (0.026)     \\ 
 LL  & \multicolumn{1}{r|}{0.55 (-)} & 0.83 (0.078) & \multicolumn{1}{c|}{NA} & 0.91 (0.122)   & \multicolumn{1}{c|}{NA}  & 1.04 (0.150)   \\ 
 RBS  & \multicolumn{1}{r|}{0.62 (-)} & 0.56 (0.030)   & \multicolumn{1}{c|}{NA} & 0.62 (0.053)      & \multicolumn{1}{c|}{NA} &0.66 (0.058)              
\end{tabular}
\caption{Empirical mean estimates, and standard deviation in brackets, of the TRMs for the three negative returns of interest, using the original data {\it Orig} and using the simulation algorithm with $100$ generated samples of size $m=10,000$ each {\it Simu}. For each TRM, a threshold level $\alpha=0.9975$ is considered. NA meaning that there are no observations above the theoretical $\VaR$, so that the computation of the TRM was not possible.}
\label{tab:TRMRealData}
\end{table}

\subsection{Conditional expectation estimations} \label{subsec:ShockReal}

We now turn to the second application of our simulation framework, namely the estimation of the conditional expectation $\mu(\bm x_{-j}) = \mathbb{E}\left[X_j \mid \bm X_{-j}= \bm x_{-j}\right]$, for $j\in\{1,2,3\}$. For illustrative purposes, the conditioning observations $(\bm x_{-j})_{j\in\{1,2,3\}}$ are selected in such a way that they cover and illustrate the three possible cases of the conditional simulation algorithm (see Section~\ref{sec:CondSim}). However, in practice, if this approach is used, for instance, for shock inference, these values  would typically be provided by regulators. 

In order to simulate samples of $X_j \mid \bm X_{-j} = \bm x_{-j}$, Algorithm~\ref{algo:CondAlgo} considers the differences $\Delta^{j,q}$ with $q\neq j$ (see Equation~\eqref{eq:MGPEq}). The algorithm was presented with $q=1$ and thus $j\neq 1$. But for $j=1$, any $q\neq 1$ can be selected. In our case, to simulate $X_1 \mid \bm X_{-1} = \bm x_{-1}$, $q$ was chosen to be equal to 3. To illustrate the three conditional simulation cases identified in Section~\ref{sec:CondSim}, we consider values of the components of the conditioning vector $\bm x_{-j}$ as quantiles at different extreme levels of the marginal distribution of the $X_j$. For each conditional simulation case and each stock, we consider only a single conditioning observation vector, and we generate 100 samples of size $m=10,000$ using the conditioning simulation Algorithm~\ref{algo:CondAlgo}. Then $\mu(\bm x_{-j})$ is estimated as the empirical mean of each simulated sample with Algorithm~\ref{algo:CondAlgo}. The mean and the standard deviation over the 100 estimates are displayed in Table~\ref{tab:LevelEstReal}. As in the numerical illustration in Section~\ref{subsec:ShocksEst}, we compare the results obtained through our approach with those obtained through linear regression. The results are depicted in Table~\ref{tab:LevelEstReal}.

Table~\ref{tab:LevelEstReal} shows that while the shocks estimates for the HSBC negative returns using conditional simulation are comparable to those obtained with linear regression, the discrepancy between these two shock estimated is more pronounced for Llyods and RBS negative returns. These two negative returns exhibit a heavier tail than HSBC. However, this type of distribution is precisely the one where linear regression is not a suitable approach, and other methodologies should be considered. The results obtained under the simulation framework of Section~\ref{subsec:NS},  where the theoretical value was used as a reference, indicated that the approximation errors were relatively low, in contrast to those produced by linear regression. Consequently, for real data,  it might be preferable to use our shock estimates over those provided by linear regression. 

\begin{table}
\centering
\begin{tabular}{c}
\begin{tabular}{l|ccc}
 &  Case 1 ($z_{\star}>0, z_{\star}=z_q$)  &  Case 2 ($z_{\star}>0, z_{\star}\neq z_q$) & Case 3 ($z_{\star}\leq 0$) \\ 
 \hline
HSBC & 0.1526 ($10^{-3}$) & 0.1094 ($10^{-4}$) & 0.0436 ($10^{-4}$) \\
LLyods & 0.3031 ($10^{-3}$) & 0.0880 ($10^{-3}$) & 0.0793 ($10^{-3}$) \\
RBS  & 0.2421 ($10^{-3}$) & 0.3273 ($10^{-3}$) & 0.0951 ($10^{-4}$) 
\end{tabular}\\~\\
a) Conditional Simulation \\~\\
\begin{tabular}{l|ccc}
 & Case 1 ($z_{\star}>0, z_{\star}=z_q$) & Case 2 ($z_{\star}>0, z_{\star}\neq z_q$) & Case 3 ($z_{\star}\leq 0$) \\
 \hline
HSBC & 0.1481 & 0.1051 & 0.0187  \\
LLyods & 0.2042 & 0.1417 & 0.0295 \\
RBS & 0.1475 & 0.3931 & 0.0478 
\end{tabular}
\\~\\
  b) Linear Regression
\end{tabular}
\caption{Empirical mean estimates, and standard deviation in brackets when available, of the estimated shocks $\mu(\bm x_{-j}) = \mathbb{E}\left[X_j \mid \bm X_{-j}= \bm x_{-j}\right]$, $j\in\{1,2,3\}$, with $X_1$ (resp. $X_2$ and $X_3$) corresponding to negative returns from HSBC (resp. Lloyds and RBS). Estimations are performed using  our conditional simulation approach (Algorithm~\ref{algo:CondAlgo}) and b) linear regression. For both inference approaches are illustrated through the three cases identified for the conditional simulation approach.}
\label{tab:LevelEstReal}
\end{table}

\section{Conclusion and Perspectives}


One of the primary concerns when attempting to estimate risks at high levels is the scarcity of available data. In fact, data sparseness makes any estimation prone to instabilities and inconsistencies, thereby limiting the effectiveness of this estimate as a trustworthy risk metric. Another consequence of considering solely historical data  is the fact that the information contained in historical observations is incomplete. Essentially, this approach represents only past crisis information, implying that upcoming crises cannot exceed the severity of those that have already been experienced. 

The issue was addressed in this paper by the development of two non-parametric simulation approaches of multivariate extremes. Both algorithms are driven by multivariate extreme value theory. The first algorithm, referred to as the joint simulation algorithm, generates MGP vectors. The second algorithm, referred to as the conditional simulation algorithm, generates samples of a component of MGP vectors conditionally on the others components being extreme.  The aforementioned algorithms have enabled the expansion of the available data set, and the use of the newly generated observations has resulted in the accurate estimation of tail risk metrics.

The performance of each of the algorithms has been demonstrated through two potential applications on both simulated and real data. The first is the estimation of TRMs for the joint simulation algorithm, and the second one is the inference of quantities involving some conditional tail distribution for the conditional simulation algorithm.

The aim of this study was to enhance the quantity of observations in extreme regions. Following the conclusion of the parametric numerical experiment, it was established that even empirical methods applied to data simulated with our algorithms can yield accurate risk measurement estimates. Consequently, future works could focus on the combination of our simulation approaches with more sophisticated estimation techniques than empirical ones \citep[e.g.][]{padoan2023marginal, davison2023tail}.

\paragraph{Codes} All the codes and data are publicly available at \url{https://github.com/MadharNisrine/MultivariateExtremeSimulator.git}.

\section{Theoretical support}\label{sec:theo}

\subsection{Theoretical support for Algorithm~\ref{alg:JointMGP}}\label{subsec:theo:joint}
The following Proposition~\ref{l:jointSim} guarantees that the simulated samples through Algorithm~\ref{alg:JointMGP} are actually distributed according to a standard MGP distribution, for all $q \in \{1,\dots,d\}$, and in particular for $q=1$. 

\begin{proposition}\label{l:jointSim} For $j\in \{1,\ldots,d\}$, let $F$ be the distribution function of $\bm \Delta^{(j)}_1,\ldots,\bm \Delta^{(j)}_{n}$,  and $\widetilde F_{nm}$ the empirical distribution function of $\widetilde {\bm \Delta}^{(j)}_1, \ldots,\widetilde {\bm \Delta}^{(j)}_m$.
If $\widetilde F_{nm}$ converges in distribution to $F$, as $n$ and $m$ tend to infinity, then $\left(\widetilde{\bm Z}_{\ell}\right)_{1 \leq \ell \leq m}$ converges in distribution to a standard MGP vector distributed as  the sample $\left(\bm Z_{i}\right)_{1 \leq i\leq n}$. 
\end{proposition}

The two following technical lemmas are useful to prove Propostion~\ref{l:jointSim}. Their proofs are straightforward. 

\begin{lemma}\label{prop:alternativeMin}
    \begin{align*}
    \min_{j}\left(z_j - \mathcal{D}_j\right) = \min_j\left(z_1+\Delta^{q,1},z_2+\Delta^{q,2},\ldots,z_d+\Delta^{q,d}\right) - \min_j\left(\Delta^{q,1},\Delta^{q,2},\ldots,\Delta^{q,d}\right),
    \end{align*}
    where $\mathcal{D}_j = \sum_{k=1,k\neq j}^d \Delta^{j,k}\prod_{l=1,l\neq k}^d\mathbf{1}_{\Delta^{l,k}<0}$. 
\end{lemma}

\begin{lemma}\label{prop:MinAlternative2}
    \begin{align*}
    &\min_j \left(z_1+\Delta^{q,1},z_2+\Delta^{q,2},\ldots,z_d+\Delta^{q,d}\right) - \min_k \left(\Delta^{q,1},\Delta^{q,2},\ldots,\Delta^{q,d}\right)\\ 
    & = - \max_j \left(T_1-z_1,T_2-z_2,\ldots,T_d-z_d\right) + \max_j \left(T_1,T_2,\ldots,T_d\right) \\
\end{align*}
\end{lemma}

\begin{proof}
Proposition~\ref{l:jointSim} and Algorithm~\ref{alg:JointMGP} are written with $q=1$, but since the proof does not change with the value of $q$, we prove the lemma with a generic value $q=1,\ldots,d$. Let $\widetilde{\bm Z}=(Z_1,\ldots,Z_d)$ be a d-dimensional vector resulting from Algorithm~\ref{alg:JointMGP}. We wish to prove that $\widetilde{\bm Z}$ converges in distribution towards a standard MGP distribution. Then, $\widetilde{\bm Z}$ satisfies Equation~\eqref{eq:MGPEq}
\[
    \widetilde{Z}_j = E +  \sum_{k=1,k\neq j}^d \widetilde{\Delta}^{j,k} \prod_{\ell=1,\ell\neq k}^d \mathbf{1}_{\widetilde{\Delta}^{\ell,k}<0}, \mbox{ for all } j=1,\ldots, d,
\]
where $E$ is a standard exponential distribution and $\widetilde{\Delta}^{j,k}  = \widetilde{Z}_j - \widetilde{Z}_k$, for all $j,k=1,\ldots, d$. Recall also the notation :  for a fixed $q\in\{1,\dots,d\}$, the $q$-th vector of differences is defined as follows
\[
\widetilde{\bm\Delta}^{(q)}=\left(\widetilde{\Delta}^{q,1},\dots,\widetilde{\Delta}^{q,d}\right)\in\mathbb{R}^{d}.
\]

Then,  the multivariate distribution of $\widetilde{\bm Z}$ is given by  

\begin{align*}
    \mathbb{P}\left(\widetilde{\bm Z} \leq \bm z\right) &= \mathbb{P}\left(E \leq z_1 -\sum_{k=1,k\neq 1}^d \widetilde{\Delta}^{1,k} \prod_{l=1,l\neq k}^d \mathbf{1}_{\widetilde{\Delta}^{l,k}\leq0},\ldots,E\leq z_d -\sum_{k=1,k\neq d}^d \widetilde{\Delta}^{d,j} \prod_{l=1,l\neq k}^d \mathbf{1}_{ \widetilde{\Delta}^{l,k}\leq0}  \right)\\
    &=\mathbb{P}\left(E \leq \min_j\left\{z_j -\sum_{k=1,k\neq j}^d \widetilde{\Delta}^{j,k}\prod_{l=1,l\neq k}^d \mathbf{1}_{\widetilde{\Delta}^{l,k}\leq0}  \right\}\right)\\
    &= \int \mathbb{P}\left(\left.E \leq \min_j\left\{z_j -\sum_{k=1,k\neq j}^d \widetilde{\Delta}^{j,k}\prod_{l=1,l\neq k}^d \mathbf{1}_{\widetilde{\Delta}^{l,k}\leq0}  \right\} \right \vert \widetilde{\bm \Delta}^{(q)} = \widetilde{\bm \delta}^{(q)}\right) f_{\widetilde{\bm \Delta}^{(q)}}(\widetilde{\bm \delta}^{(q)}) \d \widetilde{\bm \delta}^{(q)} \\
    &= \int \mathbb{P}\left(E \leq \min_j\left\{z_j -\sum_{k=1,k\neq j}^d \widetilde{\delta}^{j,k} \prod_{l=1,l\neq k}^d \mathbf{1}_{\widetilde{\delta}^{l,k}\leq0}  \right\}\right) f_{\widetilde{\bm \Delta}^{(q)}}(\widetilde{\bm \delta}^{(q)}) \d \widetilde{\bm \delta}^{(q)}\\
     &= \int \left[1- \min\left( 1, \exp\left[-\min_j\left\{z_j -\sum_{k=1,k\neq j}^d \widetilde{\delta}^{j,k}\prod_{l=1,l\neq k}^d \mathbf{1}_{\widetilde{\delta}^{l,k}\leq0}  \right\}\right]\right)\right] f_{\widetilde{\bm \Delta}^{(q)}}(\widetilde{\bm \delta}^{(q)}) \d \widetilde{\bm \delta}^{(q)}\\
     &=1 - \mathbb{E}_{\widetilde{\bm \Delta}^{(q)}}\left[\min\left( 1, \exp\left[-\min_j\left\{z_j -\sum_{k=1,k\neq j}^d \widetilde{\Delta}^{j,k}\prod_{l=1,l\neq k}^d \mathbf{1}_{\widetilde{\Delta}^{l,k}\leq0}  \right\}\right]\right)\right]\\
     &= 1 - \mathbb{E}_{\widetilde{\bm \Delta}^{(q)}}\left[\min\left( 1, e^{-\min_j \left(z_1+\widetilde{\Delta}^{q,1},z_2+\widetilde{\Delta}^{q,2},\ldots,z_d+\widetilde{\Delta}^{q,d}\right) + \min_j\left(\widetilde{\Delta}^{q,1},\widetilde{\Delta}^{q,2},\ldots,\widetilde{\Delta}^{q,d}\right)  }\right)\right]
\end{align*}
Last equation arises from Lemma~\ref{prop:alternativeMin}. Then, assuming that the Portemanteau theorem holds in dimension $d>2$, 
\begin{align*}
    \mathbb{P}\left(\left.\widetilde{\bm Z} < \bm z\right\lvert \bm \Delta_1^{(q)}, \ldots,\bm \Delta_{n}^{(q)}\right) &\rightarrow 1 - \mathbb{E}_{\bm \Delta^{(q)}}\left[\min\left( 1, e^{-\min_j \left(z_1+\Delta^{q,1},z_2+\Delta^{q,2},\ldots,z_d+\Delta^{q,d}\right) + \min_j \left(\Delta^{q,1},\Delta^{q,2},\ldots,\Delta^{q ,d}\right)  }\right)\right]
\end{align*}
Now, noting that the term in the exponent could be written alternatively as described in Lemma~\ref{prop:MinAlternative2} injecting this result into the exponent term gives 

\begin{align*}
    \mathbb{P}\left(\left.\widetilde{\bm Z} < \bm z\right\lvert \bm \Delta_1^{(q)}, \ldots,\bm \Delta_{n}^{(q)}\right) &\rightarrow 1 - \mathbb{E}_{\bm \Delta^{(q)}}\left[\min\left( 1, e^{\max_j \left(T_1-z_1,T_2-z_2,\ldots,T_d-z_d\right) - \max_j \left(T_1,T_2,\ldots,T_d\right)}\right)\right]\\ 
    &=1 - \mathbb{E}_{\bm \Delta^{(q)}}\left[\min\left( 1, e^{\max\left(\bm T - \bm z\right) - \max \left(\bm T\right)}\right)\right],
\end{align*}
hence, one recovers the cumulative distribution function of a multivariate GP distributions as defined in \citep{rootzen2018multivariate2}. 
\end{proof}

\subsection{Theoretical support for Algorithm~\ref{algo:CondAlgo}}\label{subsec:theo:cond}

Proposition~\ref{prp:cond:deltaj} derives the density of the conditional distribution of $\Delta^{q,j}$ given $\bm Z_{-j}=\bm z_{-j}$ needed for Algorithm~\ref{algo:CondAlgo}. 

\begin{proposition}\label{prp:cond:deltaj}
    For $q\in \{1,\ldots,d\}$ with $q\neq j$, let $z_\star = \max \bm z_{-j}$. The conditional distribution of $\Delta^{q,j}$ given $\bm Z_{-j}=\bm z_{-j}$ is given by
\begin{enumerate}
    \item If $z_\star>0$ and $z_\star=z_q$, 
    \[
    f_{\Delta^{q,j} \mid \bm Z_{-j}=\bm z_{-j}}(\delta^{q,j} ) = \frac{1}{I_1(0) + I_2(0)} \left(\mathbf{1}_{\delta^{q,j}>0} + \e^{\delta^{q,j}}
         \mathbf{1}_{\delta^{q,j} \leq 0}\right)f_{\bm \Delta^{(q)}}\left(\bm{\delta}^{(q)} \right).
         \]
    \item If $z_\star>0$ and $z_\star\neq z_q$, denoting $\delta_\star \coloneqq z_q - z_\star$, 
    \[
    f_{\Delta^{q,j} \mid \bm Z_{-j}=\bm z_{-j}}(\delta^{q,j} ) =\frac{1}{I_1(z_q-z_\star) + I_2(z_q-z_\star)} \left( \e^{\delta^{q,j}}  \mathbf{1}_{\delta_{\star}>\delta^{q,j}} + \e^{\delta_{\star}}\mathbf{1}_{\delta_{\star}\leq\delta^{q,j}}\right) f_{\bm \Delta^{(q)}}\left(\bm{\delta}^{(q)} \right).
    \]
    \item If $z_\star\leq0$, 
\[
f_{\Delta^{q,j} \mid \bm Z_{-j}=\bm z_{-j}}(\delta^{q,j} ) = \frac{1}{I_1(z_q)} \e^{\delta^{q,j}}  f_{\bm \Delta^{(q)}}\left(\bm{\delta}^{(q)} \right) \mathbf{1}_{\delta^{q,j} <z_q}.
\]
\end{enumerate}
where \begin{itemize}
    \item  $\bm \delta^{(q)}=(\delta^{(q,1)},\dots,\delta^{(q,d)})$
    \item for any $x\in\mathbb{R}$, 
\[
I_1(x)\coloneqq \int_{-\infty}^x \e^{\zeta} f_{\bm \Delta^{(q)}}\left(\bar{\bm{\delta}}^{(q)} \right) \d \zeta \quad \text{ and } \quad
     I_2(x)\coloneqq \e^x\int_x^{\infty}  f_{\bm \Delta^{(q)}}\left(\bar{\bm{\delta}}^{(q)} \right) \d \zeta , 
 \]
\end{itemize}
where $\bar{\bm{\delta}}^{(q)}$ has the same components as $\bm{\delta}^{(q)}$ except for $\delta^{q,j}$ which is replaced by $\zeta= z_q - z_j$.
\end{proposition}

\begin{proof}
The conditional distribution of $Z_j$  given $\bm Z_{-j} = \bm z_{-j}$ is given by  
\begin{equation*}
    f_{Z_j | \bm Z_{-j}= \bm z_{-j}}(z_j) =  f_{\bm Z}(\bm z)/f_{\bm Z_{-j}}(\bm z_{-j} ),  
\end{equation*}
where the joint distribution function of $\bm Z$ is given by 
\begin{align*}
    f_{\bm Z}(\bm z) = \mathrm{e}^{-\max\{\bm z\}} f_{\bm \Delta^{(q)}}\left(\bm{\delta}^{(q)} \right) \mathbf{1}_{\bm z \nleqslant 0}
\end{align*}
for a fixed $q\in \{1,\dots,d\}$  such that $q\neq j$.

We distinguish between two cases depending on the sign of $z_{\star}:=\max \bm z_{-j}$. 

\paragraph{If $z_{\star}>0$.} First, the marginal distribution of $\bm Z_{-j}$ is obtained by 
\begin{eqnarray*}
    f_{\bm Z_{-j}}(\bm z_{-j} ) &=& \int_\mathbb{R}  \e^{-\max\{\bm z\}} f_{\bm \Delta^{(q)}}\left(\bm{\delta}^{(q)} \right) \d z_j \\
    &=&\int_\mathbb{R} \e^{-\max\{z_1-z_q+\zeta,z_2-z_q+\zeta,\ldots, z_d - z_1 +\zeta\}-z_q+\zeta} f_{\bm \Delta^{(q)}}\left(\bar{\bm{\delta}}^{(q)} \right) \d \zeta, \\
\end{eqnarray*}
where $\bar{\bm{\delta}}^{(q)}$ has the same components as $\bm{\delta}^{(q)}$ except for $\delta^{q,j}$ which is replaced by $\zeta= z_q - z_j$.
The last integral is split in two integrals of  $f_{\bm Z}(\bm z)$ on $\mathbb{R}_{-}$ and $\mathbb{R}_{+}$, respectively. We now compute each integral separately. 

For the integral $\mathbb{R}_{+}$, let us note that when $\zeta=z_q-z_j>0$, then $z_\star >z_j$, thus\\ ${\max\{z_1-z_q+\zeta,z_2-z_q+\zeta,\ldots, z_d - z_q +\zeta\}+z_q-\zeta = \max(\bm z) = z_\star}$, which yields
 \begin{equation*}
   \int_0^{\infty} \e^{-\max\{z_1-z_q+\zeta,z_2-z_q+\zeta,\ldots, z_d - z_q +\zeta\}-z_q+\zeta} f_{\bm \Delta^{(q)}}\left(\bar{\bm{\delta}}^{(q)} \right) \d \zeta= \e^{-z_{\star}} \int_0^{\infty}   f_{\bm \Delta^{(q)}}\left(\bar{\bm{\delta}}^{(q)}\right) \d \zeta.
\end{equation*}

For the integral $\mathbb{R}_{-}$, we split the integral on the intervals $(-\infty, z_q -z_\star]$  and $[z_q -z_\star,0]$. On the interval $(-\infty, z_q -z_\star]$, $z_j\geq  z_\star$, so that $\max(\bm z) = z_j$ and on the interval $[z_q -z_\star,0]$, $z_j\leq  z_\star$, so that $\max(\bm z) = z_\star$. Thus, we obtain 
\begin{eqnarray*}
\lefteqn{\int_{-\infty}^0 \e^{-\max\{z_1-z_q+\zeta,z_2-z_q+\zeta,\ldots, z_d - z_q +\zeta\}-z_q+\zeta} f_{\bm \Delta^{(q)}}\left(\bar{\bm{\delta}}^{(q)} \right) \d \zeta}\\
&=& \left(\e^{-z_q}\int_{-\infty}^{z_q -z_\star} e^{\zeta}  f_{\bm \Delta^{(q)}}\left( \bar{\bm{\delta}}^{(q)}\right) \d \zeta + e^{-z_\star} \int_{z_q -z_\star}^0   f_{\bm \Delta^{(q)}}\left( \bar{\bm{\delta}}^{(q)}\right) \d \zeta \right)\mathbf{1}_{z_{\star}\neq z_q} \\
 &&+ \left(e^{-z_q}\int_{-\infty}^0 e^{\zeta} f_{\bm \Delta^{(q)}}\left( \bar{\bm{\delta}}^{(q)}\right)\d \zeta \right)\mathbf{1}_{z_{\star}=z_q} \, . 
\end{eqnarray*}

Hence, the marginal distribution of $\bm Z_{-j}$ is equal to  
\begin{enumerate}
\item If $z_{\star} >0$ and $z_{\star} = z_q $
\[
f_{\bm Z_{-j}}(\bm z_{-j} ) = \e^{-z_q}\left(\int_{-\infty}^0 \e^{\zeta} f_{\bm \Delta^{(q)}}\left(\bar{\bm{\delta}}^{(q)} \right) \d \zeta + \int_0^{\infty}   f_{\bm \Delta^{(q)}}\left(\bar{\bm{\delta}}^{(q)} \right) \d \zeta  \right) =  \e^{- z_q} \left(I_1(0) + I_2(0) \right) 
\]
\item If $z_{\star} >0$ and $z_{\star} \neq z_q$
\begin{align*}
f_{\bm Z_{-j}}(\bm z_{-j} ) &=\e^{-z_q}\left(\int_{-\infty}^{z_q -z_\star} \e^{\zeta}  f_{\bm \Delta^{(q)}}\left(\bar{\bm{\delta}}^{(q)} \right) \d \zeta + \e^{z_q-z_\star }  \int_{z_q -z_\star}^{\infty} f_{\bm \Delta^{(q)}}\left(\bar{\bm{\delta}}^{(q)} \right) \d \zeta\right)\\
& = \e^{-z_q} \left(I_1(z_q-z_\star) + I_2(z_q-z_\star) \right)
\end{align*}
\end{enumerate}
where we have denoted, for any $x \in \mathbb R$
\begin{equation*}
    I_1(x)= \int_{-\infty}^x \e^{\zeta} f_{\bm \Delta^{(q)}}\left(\bar{\bm{\delta}}^{(q)} \right) \d \zeta \quad \text{ and } \quad
     I_2(x)= \e^x\int_x^{\infty}  f_{\bm \Delta^{(q)}}\left(\bar{\bm{\delta}}^{(q)} \right) \d \zeta
\end{equation*}
The conditional distribution of $Z_j$ given $\bm Z_{-j} =\bm z_{-j}$ is then obtained as
\begin{enumerate}
\item If $z_{\star} >0$ and $z_{\star} = z_q $
\[
f_{Z_j | \bm Z_{-j} =  \bm z_{-j}}(z_j) =  \frac{1}{I_1(0) + I_2(0)} \left(\mathbf{1}_{z_q>z_j} + \e^{z_q-z_j}
         \mathbf{1}_{z_q\leq z_j}\right)f_{\bm \Delta^{(q)}}\left(\bm{\delta}^{(q)} \right)
\]
\item If $z_{\star} >0$ and $z_{\star} \neq z_q$
\[
f_{Z_j | \bm Z_{-j} =  \bm z_{-j}}(z_j) =  \frac{1}{I_1(z_q-z_\star) + I_2(z_q-z_\star)} \left( \e^{(z_q -z_j)}\mathbf{1}_{z_j>z_{\star}} + \e^{(z_q -z_{\star})}\mathbf{1}_{z_j\leq z_{\star}}\right) f_{\bm \Delta^{(q)}}\left(\bm{\delta}^{(q)} \right)
\]
\end{enumerate}
Yielding, 
\begin{enumerate}
\item If $z_{\star} >0$ and $z_{\star} = z_q $
\[
f_{\Delta^{q,j} | \bm Z_{-j}=\bm z_{-j}}(\delta^{q,j} )  =   \frac{1}{I_1(0) + I_2(0)} \left(\mathbf{1}_{\delta^{q,j}>0} + \e^{\delta^{q,j}}
         \mathbf{1}_{\delta^{q,j} \leq 0}\right)f_{\bm \Delta^{(q)}}\left(\bm{\delta}^{(q)} \right)
\]
\item If $z_{\star} >0$ and $z_{\star} \neq z_q$
\[
f_{\Delta^{q,j} | \bm Z_{-j}=\bm z_{-j}}(\delta^{q,j} )  =    \frac{1}{I_1(z_q-z_\star) + I_2(z_q-z_\star)} \left( \e^{\delta^{q,j}}\mathbf{1}_{\delta_{\star}>\delta^{q,j}} + \e^{\delta_{\star}}\mathbf{1}_{\delta_{\star}\leq\delta^{q,j}}\right) f_{\bm \Delta^{(q)}}\left(\bm{\delta}^{(q)} \right)
\]
\end{enumerate}
where $\delta_\star = z_q - z_\star$.

Note that while the conditional distribution  does not depend on $\bm z_{-j}$ when $z_{\star}=z_q$, it does depend on $\bm z_{-j}$ when $z_{\star}\neq z_q$. 

\paragraph{If  $ z_{\star} \leq 0$ (Case 3).} The joint distribution is given by
\begin{equation*}
    f_{\bm Z}(\bm z) =  e^{-\max\{\bm z\}} f_{\bm \Delta^{(q)}}\left(\bm{\delta}^{(q)} \right) \mathbf{1}_{z_j >0}.
\end{equation*}
Thus, the marginal distribution of $\bm Z_{-j}$ is obtained via
\begin{equation*}
    f_{\bm Z_{-j}}(\bm z_{-j} )=  \e^{-z_q}\int_{- \infty}^{z_q} \e^{\zeta}f_{\bm \Delta^{(q)}}({\bar{\bm{\delta}}^{(q)}}) \d \zeta,
\end{equation*} the conditional distribution  is then given by 
\begin{equation*}
    f_{Z_j | \bm Z_{-j} =\bm z_{-j}}(z_j) = \frac{1}{I_1(z_q)} \e^{z_q - z_j}  f_{\bm \Delta^{(q)}}\left(\bm{\delta}^{(q)} \right) \mathbf{1}_{z_j>0} \, .
\end{equation*}
Finally, 
\begin{equation*}
f_{\Delta^{q,j} | \bm Z_{-j}=\bm z_{-j}}(\delta^{q,j} ) = \frac{1}{I_1(z_q)} \e^{\delta^{q,j}}  f_{\bm \Delta^{(q)}}\left(\bm{\delta}^{(q)} \right) \mathbf{1}_{\delta_j <z_q} \, .
\end{equation*} 
 \end{proof}

Proposition~\ref{prop:RS} guarantees that the use of the rejection sample in a multivariate context. 

\begin{proposition}\label{prop:RS}
    Let $V$ be a random variable with density $f$ and let $\bm Z$ be a $d$-dimensional random vector of density $g$ from which simulation of samples is known. Assuming that there exits a function $R$ in $\mathbb R^{d-1}$ such that $\int (R(\bm z_{-j}))^{-1}  \d \bm z_{-j} < \infty$ and $z_j$-independent.  Then, for all $v \in \mathbb R$,
\begin{equation*}
    \mathbb{P}\left(V < v\right) = \mathbb{P}\left(Z_j < v \lvert U < \frac{f(Z_j)}{R(\bm Z_{-j})g(\bm Z)}\right) \, ,
\end{equation*}
where $U$ is a uniform random variable on $[0,1]$.
\end{proposition}

\begin{proof} 
Let $v \in \mathbb R$. 
\begin{align*}
    \mathbb{P}\left(Z_j < v \mid U < \frac{f(Z_j)}{R(\bm Z_{-j})g(\bm Z)}\right)
    &= \frac{\mathbb{P}\left(Z_j < v , U < \frac{f(Z_j)}{R(\bm Z_{-j})g(\bm Z)}\right)}{\mathbb{P}\left(U < \frac{f(Z_j)}{R(\bm Z_{-j})g(\bm Z)}\right)}\\
    &=\frac{\int \mathbb{P}\left(Z_j < v , U < \frac{f(Z_j)}{R(\bm Z_{-j})g(\bm Z)}\mid \bm Z= \bm z\right) g(\bm z) \d \bm z}{\int \mathbb{P}\left(U < \frac{f(Z_j)}{R(\bm Z_{-j})g(\bm Z)}\mid \bm Z= \bm z\right) g(\bm z) \d \bm z} \\
    &= \frac{\int \mathbf{1}_{z_j < v}  \frac{f(z_j)}{R(\bm z_{-j})g(\bm z)} g(\bm z) \d \bm z}{\int \frac{f(z_j)}{R(\bm z_{-j})g(\bm z)} g(\bm z)  \d \bm z}\\
    & = \frac{\int_{-\infty}^v \mathbf{1}_{z_j < v}  f(z_j) \left(\int (R(\bm z_{-j}))^{-1}  \d \bm z_{-j}\right) \d z_j}{ \int  f(z_j) \left(\int (R(\bm z_{-j}))^{-1}  \d \bm z_{-j}\right) \d z_j}\, .
\end{align*}
Assuming that $ \int (R(\bm z_{-j}))^{-1}  \d \bm z_{-j} < \infty$ and $z_j$-independent, it follows that 
\begin{equation*}
     \mathbb{P}\left(Z_j < v \lvert U < \frac{f(Z_j)}{R(\bm Z_{-j})g(\bm Z)}\right) = \frac{\int (R(\bm z_{-j}))^{-1}  \d \bm z_{-j} \int_{-\infty}^v  f(z_j) \d z_j}{  \int (R(\bm z_{-j}))^{-1}  \d \bm z_{-j} \int f(z_j)\d z_j}=\int_{-\infty}^v  f(z_j) \d z_j = \mathbb{P}\left( V \leq v \right) \, . 
\end{equation*}
\end{proof}

Now, we need to specify the rejection constant $R(\bm z_{-j})$ is our context and verify that \\  ${\int (R(\bm z_{-j}))^{-1}  \d \bm z_{-j} < \infty}$ and that it is $z_j$-independent. $R(\bm z_{-j})$ is defined as follows
\[
R(\bm z_{-j}) = R(\bm \delta_{-j})= \sup_{\delta^{q,j}}\frac{f_{\Delta^{q,j}|\bm Z_{-j}=\bm z_{-j}}\left(\delta^{q,j} \right)}{f_{\bm \Delta^{(q)}}(\bm \delta^{(q)})}
\]
where $\delta^{q,j} = z_q - z_{j}$. Thanks to Proposition~\ref{prp:cond:deltaj}, we obtained that 
\begin{equation}\label{eq:choice:R}
R(\bm z_{-j}) = 
\begin{cases}
I_1(0) + I_2(0) & \text{ if $z_\star>0$ and $z_\star=z_q$} \\
I_1(z_1 +z_{\star}) + I_2(z_1 +z_{\star}) & \text{ if $z_\star>0$ and $z_\star \neq z_q$} \\
I_1(z_q)& \text{ if $z_\star\leq 0$ and $z_j>0$}
\end{cases}
\end{equation}

\begin{proposition}\label{prop:RS2}
    $R(\bm z_{-j})$ chosen as in Equation~\eqref{eq:choice:R} verifies ${\int (R(\bm z_{-j}))^{-1}  \d \bm z_{-j} < \infty}$ and that it is $z_j$-independent. 
\end{proposition}

\begin{proof}
Consider the first case where $\bm z_{\star}= z_q>0$. Then, 
\begin{equation*}
    \int (R(\bm z_{-j}))^{-1}  \d \bm z_{-j} =\int \left(\int_{-\infty}^0 \e^{\zeta} f_{\bm \Delta^{(q)}}\left(\bar{\bm{\delta}}^{(q)} \right) \d \zeta + \int_0^{\infty}   f_{\bm \Delta^{(q)}}\left(\bar{\bm{\delta}}^{(q)} \right) \d \zeta \right) \d \bm z_{-j}   \, .
\end{equation*}

Knowing that $f_{\bm \Delta^{(q)}}(\bar{\bm{\delta}}^{(q)})$ is a probability density function, 
\begin{align}\label{eq:fDelta_yk}
    \int_0^{\infty} f_{\bm \Delta^{(q)}}(\bar{\bm{\delta}}^{(q)}) \d \zeta < \int_{\mathbb{R}} f_{\bm \Delta^{(q)}}(\bar{\bm{\delta}}^{(q)}) \d \zeta < +\infty \, .
\end{align} 
Similarly ,
\begin{align}\label{eq:expfDelta_yk}
    \int_{- \infty}^0 \e^{\zeta}f_{\bm \Delta^{(q)}}(\bar{\bm{\delta}}^{(q)}) \d \zeta&=\int_{0}^{+\infty} \e^{-\xi}f_{\bm \Delta^{(q)}}(\bar{\delta}^{q,1},\ldots,-\xi,\ldots,\bar{\delta}^{q,d}) \d \xi \nonumber\\ 
    &< \int_{0}^{+\infty} f_{\bm \Delta^{(q)}}(\bar{\delta}^{q,1},\ldots,-\xi,\ldots,\bar{\delta}^{q,d}) \d \xi\nonumber\\ 
    &<\int_{\mathbb{R}} f_{\bm \Delta^{(q)}}(\bar{\delta}^{q,1},\ldots,-\xi,\ldots,\bar{\delta}^{q,d}) \d \xi < +\infty
\end{align}

Inequalities~\eqref{eq:fDelta_yk} and~\eqref{eq:expfDelta_yk} show that $R(\bm z_{-j})$ is well defined and finite.  Now, we show that  ${\int (R(\bm z_{-j}))^{-1}  \d \bm z_{-j} < \infty}$, which can be shown through the application of this same argument iteratively on each component of $\bm z_{-j} $, which gives 
\begin{align*}
    \int_{\mathbb{R}}\ldots \int_{\mathbb{R}} \int_{0}^{+\infty}f_{\bm \Delta^{(q)}}(\bar{\bm{\delta}}^{(q)}) \d \zeta \d z_d \ldots \d z_1 < \int_{\mathbb{R}}\cdots \int_{\mathbb{R}} \int_{\mathbb{R}}f_{\bm \Delta^{(q)}}(\bar{\bm{\delta}}^{(q)}) \d u \d z_d \ldots \d z_1 =1.
\end{align*}

These results held true for the remaining cases. 
\end{proof}

\bibliographystyle{chicago}%
\bibliography{mybibsimu}

\end{document}


\maketitle

\section{Asymptotic Dependence} \label{a:chiPlotAppendix}
A key assumption about the vector $\bm X$ is the existence of asymptotic dependence among its components. To ensure that this hypothesis holds true when samples are generated using Gumbel copula as it is the case in the parametric framework of Section~4.1., we use $\chi_{1:d}$ a measure of asymptotic dependence defined as follows 
\begin{equation*}
    \chi_{1:d}(\alpha) := \frac{\mathbb{P}\left(F_1(X_1)>\alpha,\dots,F_d(X_d)>\alpha\right)}{1-\alpha}.
\end{equation*}Figure~\ref{fig:chiplot} shows the empirical version of $\chi_{1:d}$ computed for increasing levels of $\alpha$, on samples simulated with Gumbel copula considering three distinct copula parameter $\theta$ values. One can clearly see that as $\theta$ increases, $\chi_{1:d}$ increases as well and becomes relatively stable. 
\begin{figure}[H]
    \centering
    \begin{tabular}{ccc}
        \centering
        \includegraphics[width=0.3\linewidth]{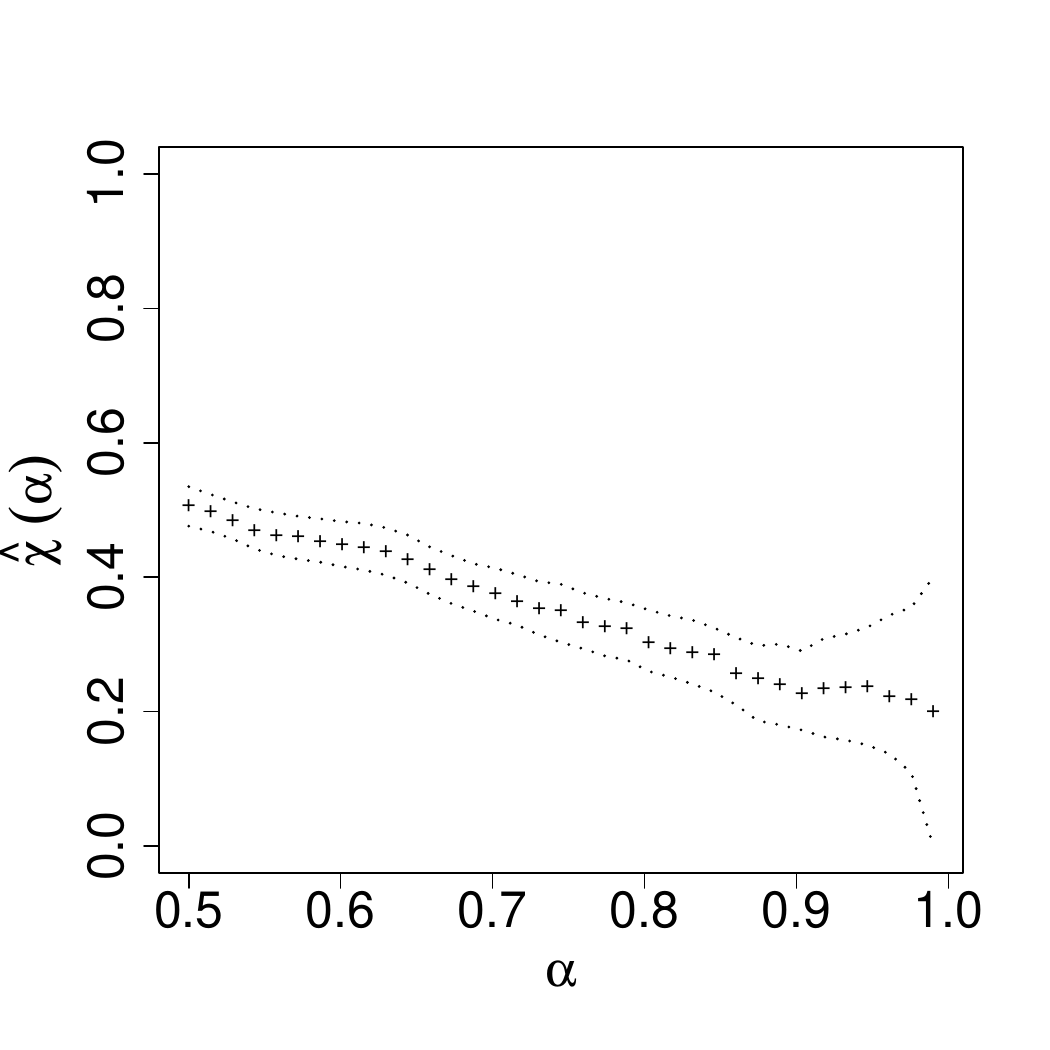} &     \includegraphics[width=0.3\linewidth]{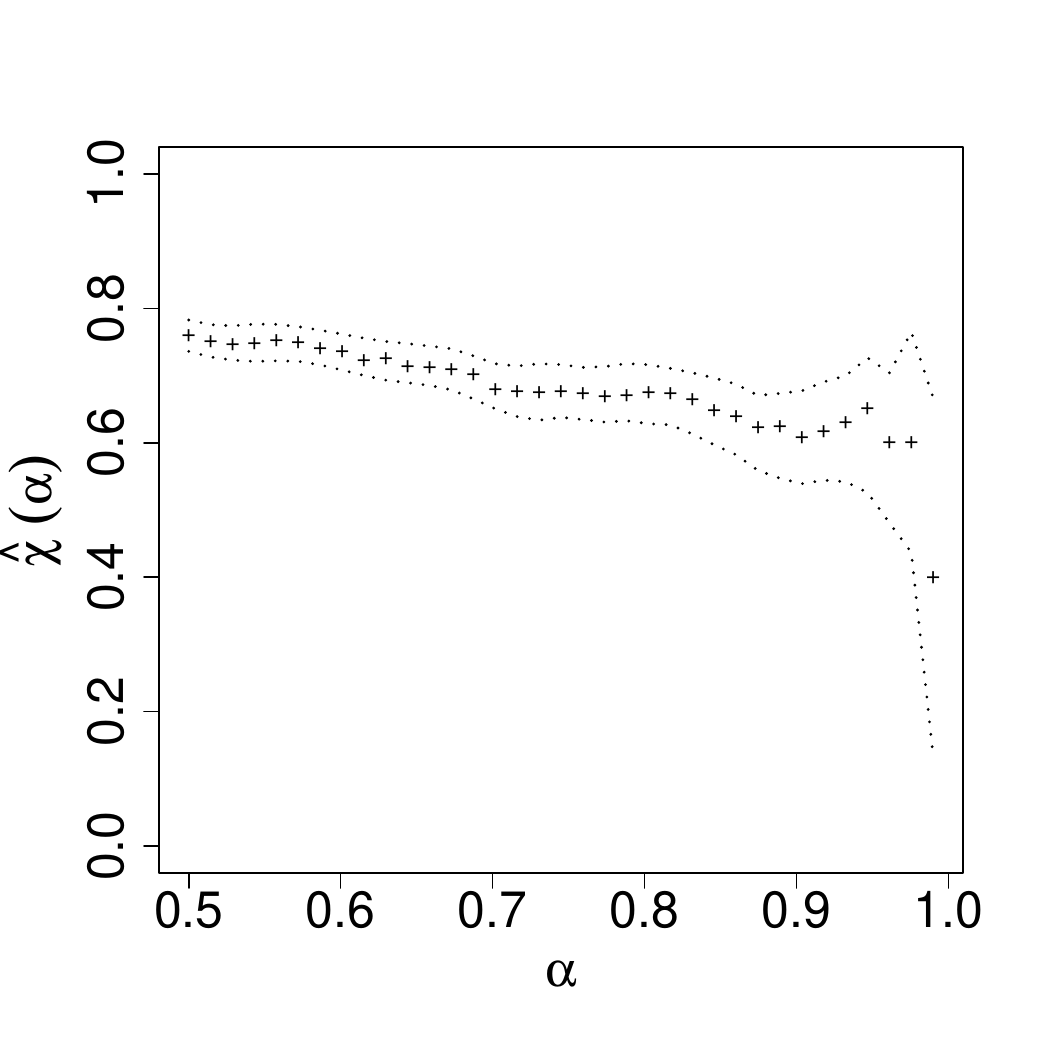} &         \includegraphics[width=0.3\linewidth]{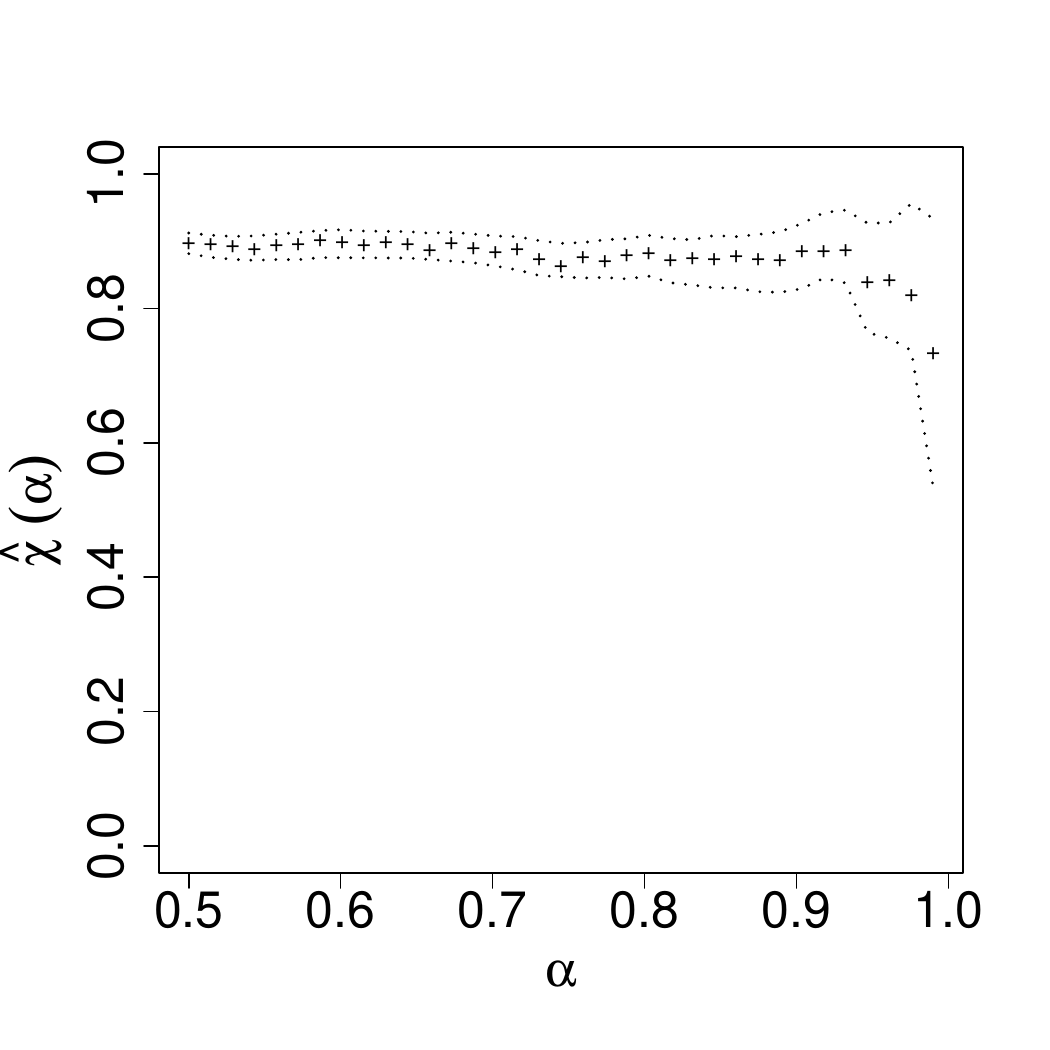}\\~\\
        a) $\theta=1.3$ & b) $\theta=2.6$ & c) $\theta=7.3$
    \end{tabular}
    \caption{$\chi$-plot of data generated with Gumbel copula with $\theta \in \{1.3,2.6,7.3\}$.}
    \label{fig:chiplot}
\end{figure}

\section{Benchmark TRM}\label{a:benchTRM}

Under the specified numerical setting the joint distribution is known, the reference values could be then exactly defined.

If $X\sim \mathcal{T}(\nu)$, then the closed form expression of the $\ES$  at level $\alpha$ \citep{norton2021calculating}  writes as 
\begin{align}\label{eq:ESStudent}
    \mathrm{ES}_{\alpha}\left(X\right) = \left(\frac{\nu + S^{-1}(1-\alpha)^2}{(\nu-1)(\alpha)}\right) s(S^{-1}(1-\alpha)),
\end{align}
where $S^{-1}(\cdot)$ is the inverse of the standardized Student-t CDF and $s(\cdot)$ is the standardized Student-t pdf.

\begin{lemma}\label{lemma:DCTE3}
    Let $X$,$Y$ and $Z$ be three random losses with a joint distribution function entirely specified through the copula density $c_{\theta}$ and their marginal cdfs (pdfs) $F_X,F_Y$ and $F_Z$, ($f_X,f_Y$ and $f_Z$), let $\bm v= \left(x_{\alpha},y_{\alpha},z_{\alpha}\right)$ be the vector of $\VaR$ level $\alpha$ associated with the random vector $\left(X,Y,Z\right)$
    \begin{equation}\label{eq:DCTEform3}
        \mathrm{DCTE}_{\alpha}(X) =   \Psi_{\alpha}^{-1}\int_{x_{\alpha}}^{\infty} xf_X(x) -x \partial_x \left\{ \mathcal{C}\left(F_X(x),1,1-\alpha)\right)+\mathcal{C}\left(F_X(x),1-\alpha,1\right)-\mathcal{C}\left(F_X(x),1-\alpha,1-\alpha)\right)\right\} \d x,
    \end{equation}
where \begin{align*}
    \Psi_{\alpha} &=    1 - F_X(x_{\alpha}) - F_Y(y_{\alpha})-F_Z(z_{\alpha}) + \mathcal{C}\left(1,F_Y(y_{\alpha}),F_Z(z_{\alpha})\right) \\
    &-\mathcal{C}\left(F_X(x_{\alpha}),1,F_Z(z_{\alpha})\right)  +\mathcal{C}\left(F_X(x_{\alpha}),F_Y(y_{\alpha}),1\right)  - \mathcal{C}\left(F_X(x_{\alpha}),F_Y(y_{\alpha}),F_Z(z_{\alpha})\right). 
\end{align*}
And 
    \begin{equation}\label{eq:MESform3}
        \mathrm{MES}_{\alpha}(X) = \frac{\int_{x_{\alpha}}^{\infty} x \left(f_X(x) +  \partial_x \left\{ \mathcal{C}\left(F_X(x),1-\alpha,1-\alpha\right)-\mathcal{C}\left(F_X(x),1,1-\alpha \right) - \mathcal{C}\left(F_X(x),1-\alpha,1 \right)\right\}\right)}{1 - 2 (1-\alpha) +\mathcal{C}\left(1,1-\alpha,1-\alpha \right)} \d x.
    \end{equation}

\end{lemma}

Prior to the  proof of Lemma~\ref{lemma:DCTE3}, Propositions~\ref{prop:forDCTEProof1},\ref{prop:forDCTEProof2} are set. 

\begin{proposition}\label{prop:forDCTEProof1}
    \begin{align*}
        \mathbb{P} \left( X<x,  Y>y,Z>z\right) = \mathbb{P} \left(  Y>y,Z>z\right)  - \mathbb{P} \left( X>x,  Y>y,Z>z\right) 
    \end{align*}
\end{proposition}
\begin{proposition}\label{prop:forDCTEProof2}
    \begin{align*}
    \mathbb{P} \left( X>x,  Y>y,Z>z\right) &=    1 - F_X(x) - F_Y(y)-F_Z(z) + \mathcal{C}\left(1,F_Y(y),F_Z(z)\right) + \mathcal{C}\left(F_X(x),1,F_Z(z)\right)  \\ 
    &+\mathcal{C}\left(F_X(x),F_Y(y),1\right)  - \mathcal{C}\left(F_X(x),F_Y(y),F_Z(z)\right) 
    \end{align*}
\end{proposition}
\begin{proof}

    \begin{align*}
    \mathbb{P} \left( X>x,  Y>y,Z>z\right) &=   \mathbb{P}\left(X>x,Y>y\right) -  \mathbb{P} \left( X>x,  Y>y,Z<z\right)\\
    &=\mathbb{P}\left(X>x,Y>y\right) - \left[\mathbb{P}\left(X>x,Z<z\right) - \mathbb{P} \left( X>x,  Y<y,Z<z\right)\right]\\
    &=\mathbb{P}\left(X>x,Y>y\right) \\
    &- \mathbb{P}\left(Z<z\right) +\mathbb{P}\left(X<x,Z<z\right) \\
    &+ \mathbb{P}\left(Y<y,Z<z\right) - \mathbb{P} \left( X<x,  Y<y,Z<z\right)\\
    &= 1- F_X(x) - F_Y(y) + \mathbb{P}\left(X<x,Y<y\right) \\
    & - F_Z(z) + \mathbb{P}\left(X<x,Z<z\right)  \\
    & + \mathbb{P}\left(Y<y,Z<z\right) - \mathbb{P} \left( X<x,  Y<y,Z<z\right)\\
    &= 1- F_X(x) - F_Y(y) + \mathcal{C}\left(F_X(x),F_Y(y),1\right) \\
    & - F_Z(z) + \mathcal{C}\left(F_X(x),1,F_Z(z)\right)  \\
    & + \mathcal{C}\left(1,F_Y(y),F_Z(z)\right) - \mathcal{C} \left( F_X(x),  F_Y(y),F_Z(z)\right)\\
    \end{align*}
    
\end{proof}

The proof of Lemma~\ref{lemma:DCTE3} is given below
\begin{proof}
    
Under the 3-dimensional case the $\mathrm{DCTE}_{\alpha}(X)$ is defined as 

\begin{align*}
    \E\left[X \lvert X>x_{\alpha},Y>y_{\alpha},Z>z_{\alpha}\right] = \int_{x_{\alpha}}^{\infty} x f_{X|X,Y,Z}(x) \d x
\end{align*}
to be able to compute it one needs to derive the expression of $f_{X|X,Y,Z}$ in terms of the marginals cumulative density functions and the copula. For $x>x_{\alpha}$, 
\begin{align*}
    \mathbb{P} \left( X<x | X>x_{\alpha},Y>y_{\alpha},Z>z_{\alpha}\right) = \frac{\mathbb{P} \left( X<x,  Y>y_{\alpha},Z>z_{\alpha}\right) - \mathbb{P} \left(X<x_{\alpha},Y>y_{\alpha},Z>z_{\alpha}\right)}{\mathbb{P} \left(X>x_{\alpha},Y>y_{\alpha},Z>z_{\alpha}\right)},
\end{align*}where $x_{\alpha},y_{\alpha}$ and $z_{\alpha}$ represent respectively the $\alpha$-quantile of $X,Y$ and $Z$. Noting that 
Using Proposition~\ref{prop:forDCTEProof1} and~\ref{prop:forDCTEProof2}, one could writes \begin{align*}
    \mathbb{P} \left( X<x,  Y>y_{\alpha},Z>z_{\alpha}\right) - \mathbb{P} \left(X<x_{\alpha},Y>y_{\alpha},Z>z_{\alpha}\right)&= F_X(x)-F_X(x_{\alpha}) + \mathcal{C}\left(F_X(x_{\alpha}),1,F_Z(z_{\alpha})\right)\\
    &- \mathcal{C}\left(F_X(x),1,F_Z(z_{\alpha})\right) \\
&+\mathcal{C}\left(F_X(x_{\alpha}),F_Y(y_{\alpha}),1\right) \\
&- \mathcal{C}\left(F_X(x),F_Y(y_{\alpha}),1\right)\\
    &+\mathcal{C}\left(F_X(x),F_Y(y_{\alpha}),F_Z(z_{\alpha})\right)\\
    &-\mathcal{C}\left(F_X(x_{\alpha}),F_Y(y_{\alpha}),F_Z(z_{\alpha})\right),
\end{align*} and \begin{align*}
    \Psi_{\alpha}  := \mathbb{P} \left( X>x_{\alpha},  Y>y_{\alpha},Z>z_{\alpha}\right) &=    1 - F_X(x_{\alpha}) - F_Y(y_{\alpha})-F_Z(z_{\alpha}) + \mathcal{C}\left(1,F_Y(y_{\alpha}),F_Z(z_{\alpha})\right) \\
    &-\mathcal{C}\left(F_X(x_{\alpha}),1,F_Z(z_{\alpha})\right)  +\mathcal{C}\left(F_X(x_{\alpha}),F_Y(y_{\alpha}),1\right)  \\
    &- \mathcal{C}\left(F_X(x_{\alpha}),F_Y(y_{\alpha}),F_Z(z_{\alpha})\right) 
\end{align*}

It follows that 
\begin{align*}
    f_{X|X,Y,Z}(x) = \frac{f_X(x) -  \partial_x \left\{ \mathcal{C}\left(F_X(x),1,F_Z(z_{\alpha})\right)+\mathcal{C}\left(F_X(x),F_Y(y_{\alpha}),1\right)-\mathcal{C}\left(F_X(x),F_Y(y_{\alpha}),F_Z(z_{\alpha})\right)\right\}}{\Psi_{\alpha}}
\end{align*}

The conditional tail expectation could be then computed as

\begin{align*}
    \mathrm{DCTE}_{\alpha}(X) &=  \Psi_{\alpha}^{-1} \int_{x_{\alpha}}^{\infty} x f_X(x) \d x \\ &-  \frac{\int_{x_{\alpha}}^{\infty} x \partial_x \left\{ \mathcal{C}\left(F_X(x),1,F_Z(z_{\alpha})\right)+\mathcal{C}\left(F_X(x),F_Y(y_{\alpha}),1\right)-\mathcal{C}\left(F_X(x),F_Y(y_{\alpha}),F_Z(z_{\alpha})\right)\right\}}{\Psi_{\alpha}} \d x.\\
\end{align*}

Under the 3-dimensional case the $\mathrm{MES}_{\alpha}(X)$ is defined as 

\begin{align*}
    \E\left[X \lvert Y>y_{\alpha},Z>z_{\alpha}\right] = \int_{x} x f_{X|Y,Z}(x) \d x
\end{align*}
to be able to compute it one needs to derive the expression of $f_{X|Y,Z}$ in terms of the marginals cumulative density functions and the copula. The starting point,
\begin{align*}
    \mathbb{P} \left( X<x | Y>y_{\alpha},Z>z_{\alpha}\right) = \frac{\mathbb{P} \left( X<x,  Y>y_{\alpha},Z>z_{\alpha}\right)}{\mathbb{P} \left(Y>y_{\alpha},Z>z_{\alpha}\right)},
\end{align*}where $y_{\alpha}$ and $z_{\alpha}$ represent respectively the $\alpha$-quantile of $Y$ and $Z$. Rewritting the denominator using Proposition~\eqref{prop:forDCTEProof1} and~\eqref{prop:forDCTEProof2} gives 
\begin{align*}
    \mathbb{P} \left( X<x,  Y>y_{\alpha},Z>z_{\alpha}\right) &= \mathbb{P} \left( Y>y_{\alpha},Z>z_{\alpha}\right) - \mathbb{P} \left( X>x,  Y>y_{\alpha},Z>z_{\alpha}\right)\\
    &= F_X(x) + \mathcal{C}\left(F_X(x),F_Y(y_{\alpha}),F_Z(z_{\alpha})\right) - \mathcal{C}\left(F_X(x),1,F_Z(z_{\alpha}) \right) - \mathcal{C}\left(F_X(x),F_Y(y_{\alpha}),1 \right), 
\end{align*}it follows that 
\begin{align*}
    f_{X|Y,Z} = \frac{f_X(x) + \partial_x\{\mathcal{C}\left(F_X(x),1-\alpha,1-\alpha\right) - \mathcal{C}\left(F_X(x),1,1-\alpha \right) - \mathcal{C}\left(F_X(x),1-\alpha,1 \right)\}}{1- 2(1-\alpha) +\mathcal{C}\left(1,1-\alpha,1-\alpha \right) }
\end{align*}
The marginal expected shortfall is computed as 
\begin{equation*}
    \mathrm{MES}_{\alpha}(X) = \int_{x} x \frac{f_X(x) + \partial_x\{\mathcal{C}\left(F_X(x),1-\alpha,1-\alpha\right) - \mathcal{C}\left(F_X(x),1,1-\alpha \right) - \mathcal{C}\left(F_X(x),1-\alpha,1 \right)\}}{1- 2(1-\alpha) +\mathcal{C}\left(1,1-\alpha,1-\alpha \right) } \d x. 
\end{equation*}

\end{proof}

\section{Sensitivity with respect to $m$}

The approach performance is evaluated for several values of $m \in \{5,000;10,000;50,000\}$, for each value of $m$ hundred samples of size $m$ are generated $\widetilde{\mathcal{D}}_l \in \mathbb{R}^{m\times d}$, $l=1,\ldots,100$. For each $\widetilde{\mathcal{D}}_l$, the TRMs are estimated at level $\alpha$ and the mean and standard deviation of TRMs over all samples are computed. Table~\ref{tab:ImpactMTRM} summaries these statistics for $\alpha\in \{0.9975,0.999,0.9997\}$. It is clear that the accuracy of the estimates increases with $m$. However the computational cost of this test routine escalates as $m$ becomes large. For the remainder to be able to run summary statistics on the TRMs estimates, the simulation procedure is run with $m=10,000$ whose results are quite stable. In fact, they are associated with fairly low standard deviations. Moreover, it seems that the approach enhance the TRMs estimation regardless of the value of $\alpha$, with a significant improvement of the estimations over the empirical approach for high values of $\alpha$. 
\begin{table}[H]
\centering
\begin{tabular}{c}
\begin{tabular}{c|c|c|c|c|c}
TRMs     & \multicolumn{1}{c|}{$m$} & \multicolumn{1}{c|}{$\lvert \mathcal{V}_{\alpha}\rvert$ } & \multicolumn{1}{c|}{{\it Orig}} & \multicolumn{1}{c|}{\it Simu} & \multicolumn{1}{c}{{\it Ext}} \\ \hline 

\multirow{3}{*}{$\mathrm{ES}^{\mbox{E}}_{\alpha}(X_1)$ }   
&5,000        & 59.75 (7.29)  & -0.27 (0.0) & -0.02 (0.21) & -0.05 (0.19)   \\ 
&10,000           & 117.28 (10.51) & -0.27 (0.0) & -0.01 (0.11) &-0.02 (0.11)  \\ 
&50,000      & 588.87 (24.59) & -0.27 (0.0) & -0.0 (0.07)  & -0.0 (0.07)        \\ \hline 

\multirow{3}{*}{$\mathrm{MES}^{\mbox{E}}_{\alpha}(X_1)$ }  & 5,000        & 43.44 (5.47)                        & -0.32 (0.0) & 0.0 (0.22)   & -0.02 (0.2)  \\

&10,000      & 86.26 (8.93)   & -0.32 (0.0) & 0.01 (0.19)  & -0.01 (0.19)   \\ 
&50,000      & 427.39 (21.74) & -0.32 (0.0) & 0.0 (0.13)   & 0.0 (0.13)   \\ \hline 

\multirow{3}{*}{$\mathrm{DCTE}^{\mbox{E}}_{\alpha}(X_1)$ } &5,000       & 38.54 (6.23)  &-0.38 (0.0) & -0.01 (0.25) & -0.05 (0.22)   \\ 
&10,000      &77.05 (9.38)  & -0.38 (0.0) & -0.02 (0.17) &-0.04 (0.16) \\ 
&50,000      & 382.86 (16.46) & -0.38 (0.0) & -0.03 (0.07) & -0.03 (0.07) 
\end{tabular}\\~\\
a)\\~\\

\begin{tabular}{c|c|c|c|c|c}
TRMs & \multicolumn{1}{c|}{$m$} & \multicolumn{1}{c|}{$\lvert \mathcal{V}_{\alpha}\rvert$ } & \multicolumn{1}{c|}{{\it Orig}} & \multicolumn{1}{c|}{{\it Simu}} & \multicolumn{1}{c}{{\it Ext}} \\ \hline

\multirow{3}{*}{$\mathrm{ES}^{\mbox{E}}_{\alpha}(X_1)$ }   & 5,000         & 23.53 (4.92)  & -0.41 (0.0) & 0.02 (0.25)  & -0.02 (0.23)  \\ 

&10,000      & 45.47 (6.51)&-0.41 (0.0) & 0.04 (0.44)  & 0.02 (0.42)  \\ 

& 50,000  & 236.54 (15.21) & -0.41 (0.0) & -0.0 (0.1)   & -0.01 (0.1) \\ \hline 

\multirow{3}{*}{$\mathrm{MES}^{\mbox{E}}_{\alpha}(X_1)$ } & 5,000 & 17.35 (4.22)& N.A (N.A)  & 0.05 (0.42)  & 0.05 (0.42)  \\ 
&10,000      & 34.5 (6.67) &N.A (N.A)   & 0.0 (0.28)   & 0.0 (0.28)   \\ 

&50,000       & 171.68 (11.91) & N.A (N.A)  & -0.01 (0.1)  & -0.01 (0.1)  \\ \hline

\multirow{3}{*}{$\mathrm{DCTE}^{\mbox{E}}_{\alpha}(X_1)$ } &5,000       &  15.85 (3.6) &N.A (N.A) & -0.06 (0.23) & -0.06 (0.23) \\ 

&10,000      &  30.96 (4.9)& N.A (N.A)   &-0.01 (0.32) & -0.01 (0.32) \\ 

&50,000      &  155.33 (12.56)& N.A (N.A)   &-0.0 (0.21)  &-0.0 (0.21)\\ \hline

\end{tabular}\\~\\
b) \\~\\

\begin{tabular}{c|c|c|c|c|c}

TRMs                  & \multicolumn{1}{c|}{$m$} & \multicolumn{1}{c|}{$\lvert \mathcal{V}_{\alpha}\rvert$ } & \multicolumn{1}{c|}{{\it Orig}} & \multicolumn{1}{c|}{{\it Simu}} & \multicolumn{1}{c}{{\it Ext}} \\ \hline

\multirow{3}{*}{$\mathrm{ES}^{\mbox{E}}_{\alpha}(X_1)$ }   
&5,000     &  6.97 (2.69) &N.A (N.A)   & 0.09 (0.75)  & 0.09 (0.75)  \\ 
&1,0,000    & 14.17 (3.5)& N.A (N.A)    & 0.05 (0.7)   & 0.05 (0.7)   \\ 
&50,000    & 72.18 (7.97) & N.A (N.A) & 0.0 (0.14)   & 0.0 (0.14)   \\ \hline 
\multirow{3}{*}{$\mathrm{MES}^{\mbox{E}}_{\alpha}(X_1)$ }  
&  5,000 & 5.07 (2.38) &N.A (N.A)& 0.01 (0.9) &0.01 (0.9)   \\ 

&10,000      &  9.7 (3.04)&N.A (N.A)    & 0.04 (0.66)  & 0.04 (0.66)  \\ 
&50,000        & 52.65 (7.04) &N.A (N.A) & -0.01 (0.16) & -0.01 (0.16) \\ \hline
\multirow{3}{*}{$\mathrm{DCTE}^{\mbox{E}}_{\alpha}(X_1)$ } 
&5,000 & 4.77 (2.35)& N.A (N.A)    &-0.02 (0.53) &-0.02 (0.53) \\ 
&10,000 & 9.19 (3.08)& N.A (N.A)&-0.01 (0.52) & -0.01 (0.52) \\ 
&50,000 &  45.81 (7.62)&N.A (N.A)&-0.03 (0.18) & -0.03 (0.18)  \\ 
\end{tabular}\\~\\
c)

\end{tabular}
     \caption{Mean and (standard deviation) of relative errors on TRMs approximations on the original sample {\it (Orig)}, over $100$ simulations of simulated samples of size $m$ {\it (Simu)} and sample including both observations {\it (Ext)}, on synthetic data generated with Gumbel copula parameter $\theta=2.6$. For $\alpha=0.9975$ (a), $\alpha=0.999$ (b), $\alpha=0.9997$ (c) along with the number of observations beyond the theoretical $\VaR$ level $\lvert \mathcal{V}_{\alpha}\rvert$.}
     \label{tab:ImpactMTRM}

\end{table}

\section{Sensitivity with respect to $\VaR$ estimation}

To investigate the impact of the $\VaR$ estimation approach on the estimations infered from the suggested simulation approach, tail-related risk metrics were estimated using $\VaR$ estimation obtained through empirical estimation (Emp.), GPD estimation (GPD) and with the theoretical $\VaR$ level (Theo.). As one could see from Figure~\ref{fig:ESwVaREst}, it seems that the most accurate TRMs estimations are obtained using the suggested method, with outstanding results when the theoretical $\VaR$ level is considered. 

\begin{figure}
\centering
\begin{subfigure}[b]{0.5\textwidth}
  \centering
  \includegraphics[width=\textwidth]{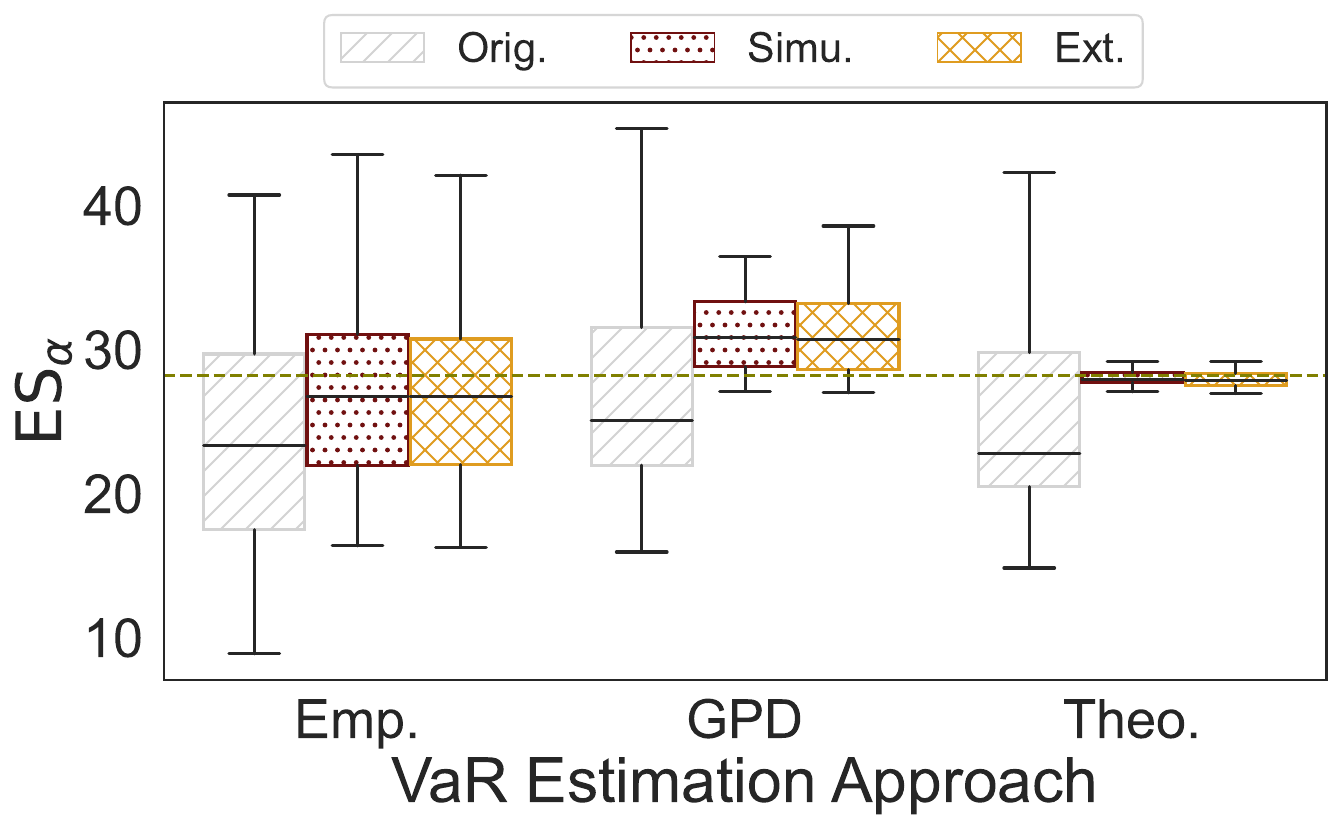}
  \caption{ES}
  
\end{subfigure}%
\hfill
\begin{subfigure}[b]{0.5\textwidth}
  \centering
  \includegraphics[width=\textwidth]{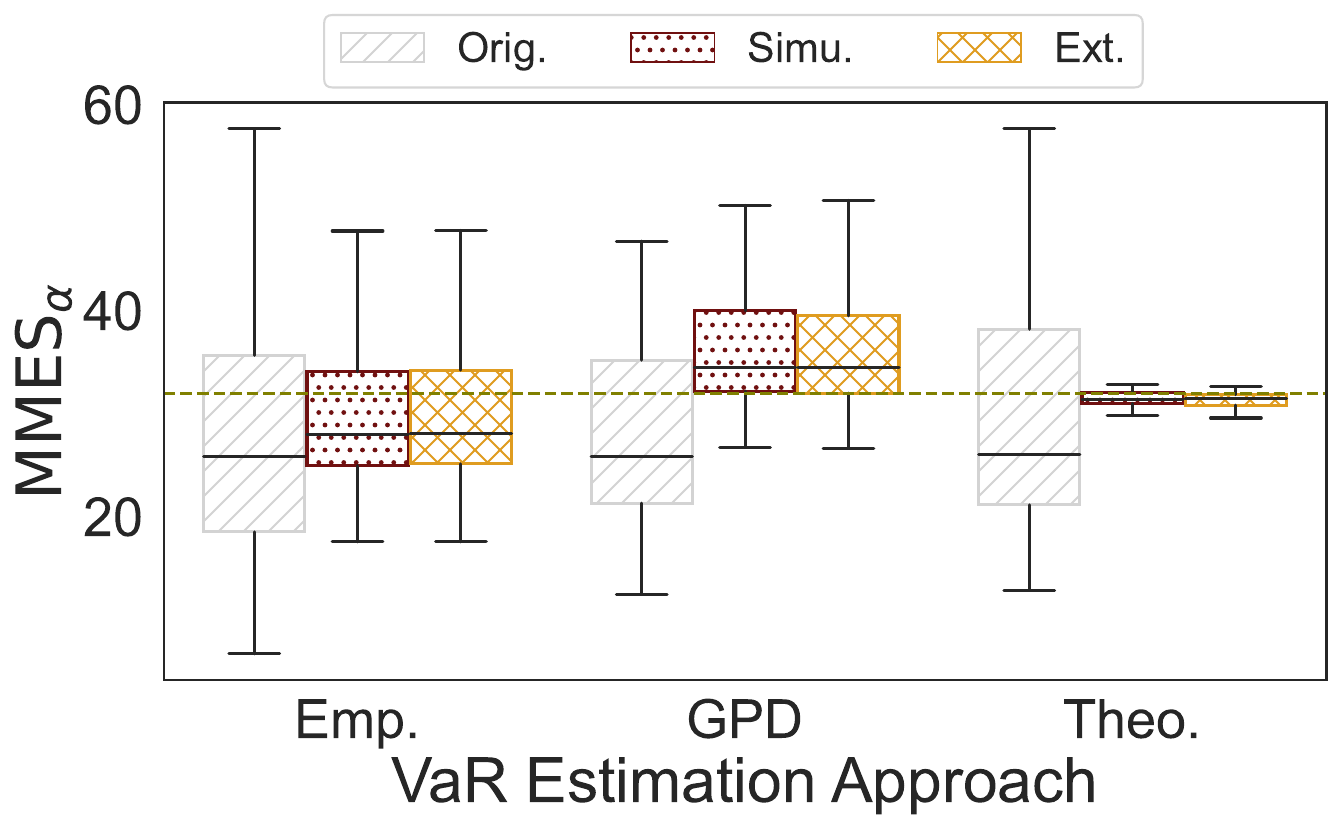}
  \caption{MES}
  \end{subfigure}
\hfill
\begin{subfigure}[b]{0.5\textwidth}
  \centering
  \includegraphics[width=\textwidth]{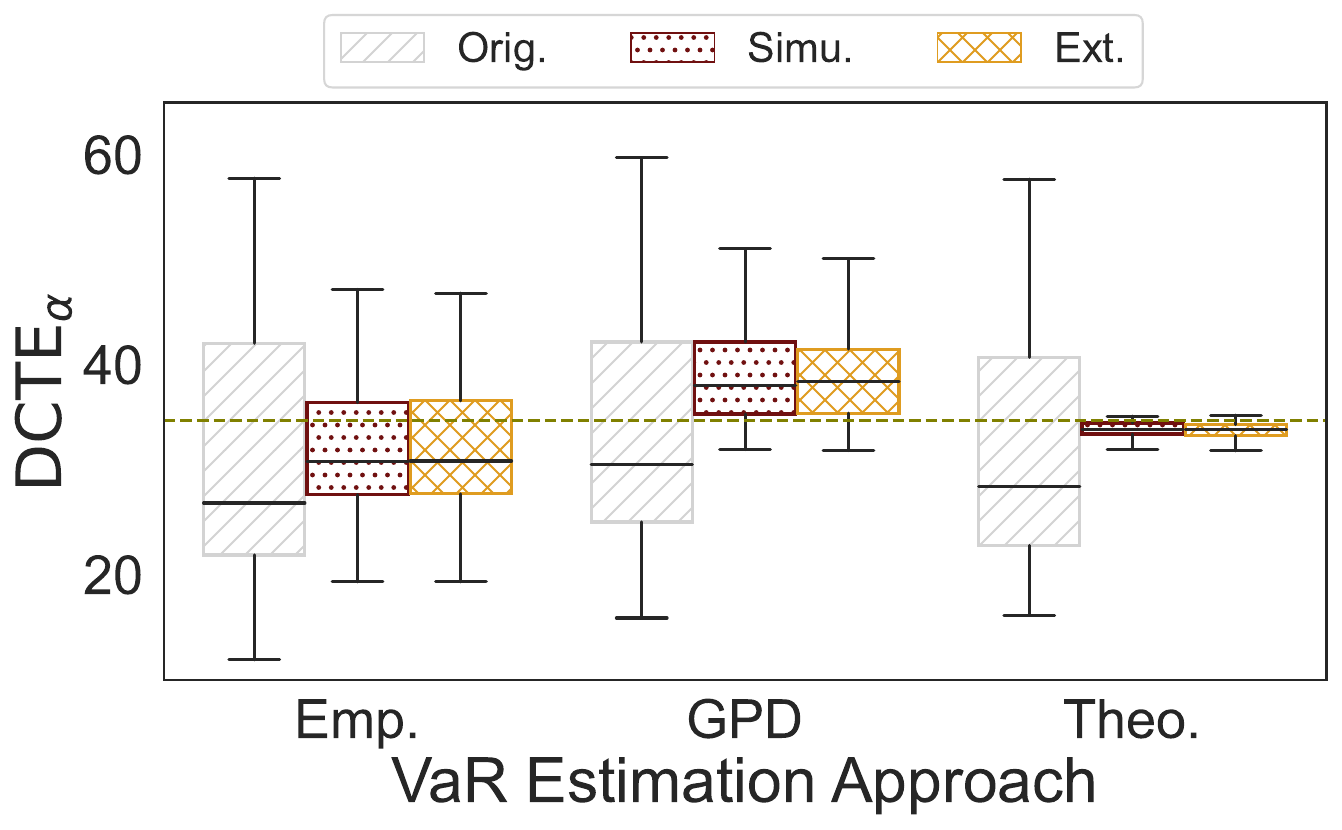}
  \caption{DCTE}
  \end{subfigure}
    \caption{Distribution of $\ES$, $\MES$ and $\DCTE$ estimations at level $\alpha=0.9975$ and $\theta=2.3$, with respect to three $\VaR$ estimation approaches. Three approaches are represented, empirical estimations on original samples {\it (grey oblique lines)}, on simulated samples {\it (red dots)}, on extended samples {\it (yellow grid)}. The green dotted line shows the benchmark value. \maud{à voir}}
  
\label{fig:ESwVaREst}
\end{figure}

\section{Marginal Fit} \label{a:MF}

In this section, we represent through Figure~\ref{fig:QQplotRealData} the QQ plots of the observed samples against fitted Student distribution for the three stocks considered in this example. This modelization step of the margins is crucial in our approach, therefore it is important to ensure that it is handled correctly. Figure~\ref{fig:QQplotRealData} shows that for all the stocks the empirical quantiles are aligned with the theoretical ones with small deviations observed in the tails regions. The $p$-values also suggests that the Student distribution is a good fit for the three stocks negative returns. The associated degrees of freedom are 3.2, 2.2 and 3.1 for HSBC, Llyods and RBS, indicate that these stocks distributions are heavy tailed. 

\begin{figure}[H]
\centering
\begin{tabular}{ccc}
\includegraphics[width=.3\textwidth]{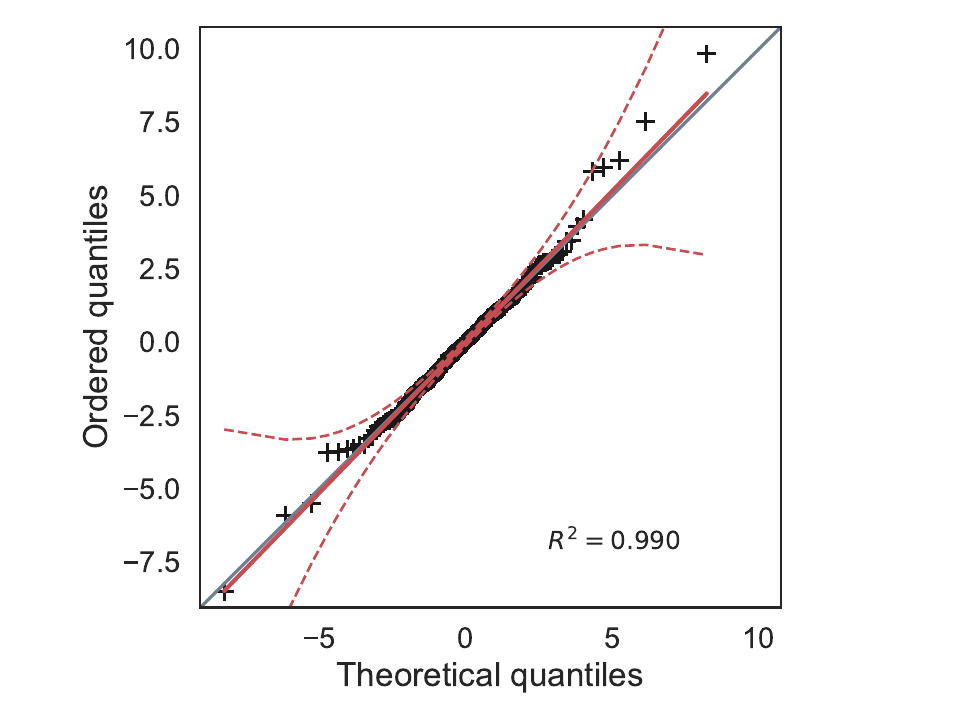}&   \includegraphics[width=.3\textwidth]{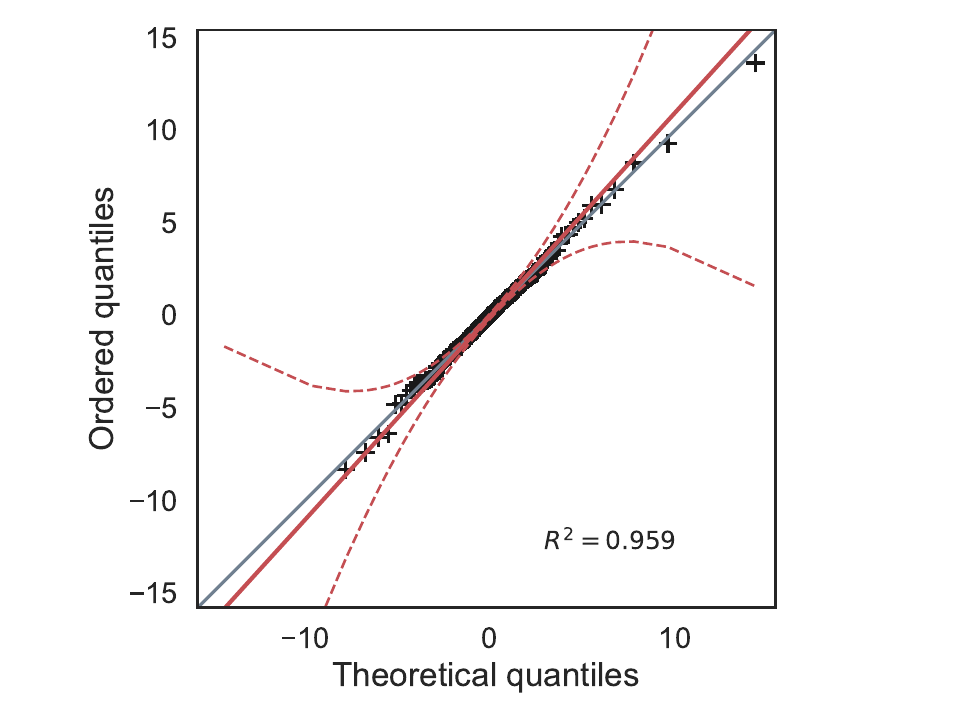} & \includegraphics[width=.3\textwidth]{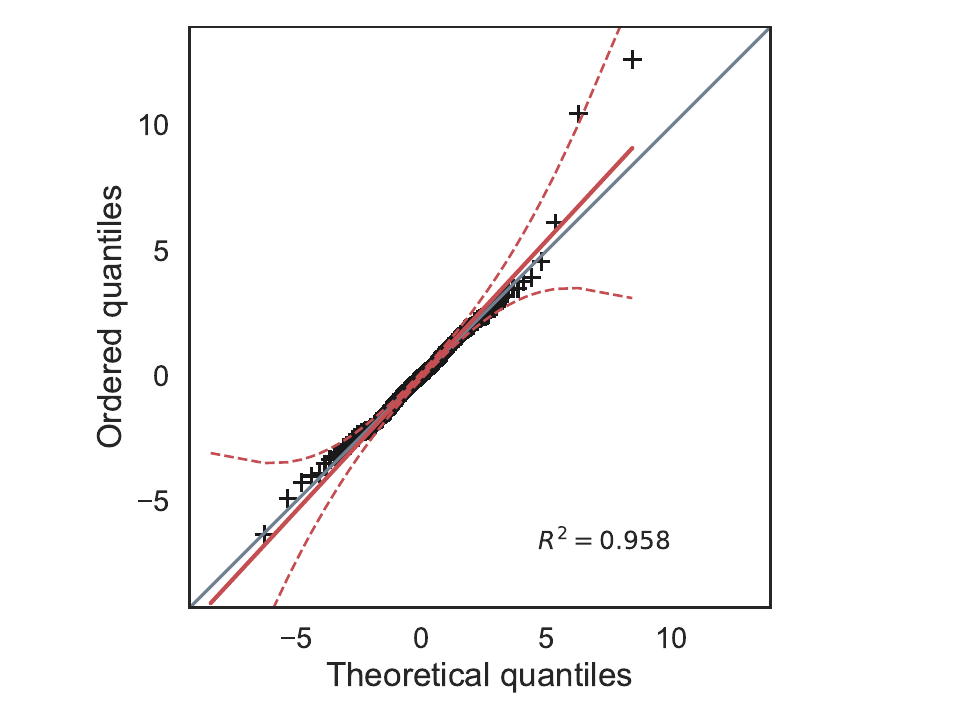}\\~\\
a) HSBC, $p$-value = 0.9345 & b) LLoyds, $p$-value=0.9688 & c) RBS, $p$-value=0.9493
  \end{tabular}
    \caption{QQ plot of observed sample against fitted Student distribution on three stocks HSBC (a), Lloyds (b) and RBS (c), with associated 95\% confidence interval (red dotted lines), the best-fit line (red) and its coefficient of determination $R^2$, the 45° line $y=x$ (grey) and the $p$-value of Kolmogorov-Smirnov test.}
  
\label{fig:QQplotRealData}
\end{figure}

\bibliographystyle{chicago}%
\bibliography{mybibsimu}